\documentclass[11pt,draftcls,draftclsnofoot,onecolumn]{IEEEtran}
\usepackage{amsmath, amsthm, amsfonts, amssymb, amsbsy}
\usepackage{setspace}
\usepackage{graphicx}
\usepackage{pstricks,arydshln}
\usepackage{multirow}

\newtheorem{theorem}{Theorem}
\theoremstyle{plain}
\theoremstyle{plain}
\theoremstyle{plain}\newtheorem{definition}{Definition}

\newtheorem{example}{Example}

\newcommand{\twotriangle}{\hfill $\bigtriangleup \bigtriangleup$  }

\begin{document}
\title{Cyclic and Quasi-Cyclic LDPC Codes on Row and Column Constrained Parity-Check Matrices and Their Trapping Sets
}

\author{(\emph{submitted to IEEE Transactions on Information Theory})\\ 
Qin Huang$^1$, Qiuju Diao$^2$, Shu	Lin$^1$ and Khaled Abdel-Ghaffar$^1$\\
{\small $^1$ Electrical and Computer Engineering Department, University of California, Davis, 95616, USA\\
$^2$ State Key Laboratory of Integrated Service Networks, Xidian University, Xi'an, 710071, CHINA}

Email: \{qinhuang,qdiao, shulin, ghaffar\}@ucdavis.edu\\ 
\thanks{This research was supported by
     NSF under the Grants CCF-0727478 and CCF-1015548, NASA under the Grant NNX09AI21G  and gift grants from Northrop Grumman Space Technology, Intel and Denali Software Inc.. }
}\maketitle

\begin{abstract}
  This paper is concerned with construction and structural analysis of both cyclic and quasi-cyclic codes, particularly LDPC codes.  It consists of three parts.  The first part shows that a cyclic code given by a parity-check matrix in circulant form can be decomposed into descendant cyclic and quasi-cyclic codes of various lengths and rates. Some fundamental structural properties of these descendant codes are developed, including the characterizations of the roots of the generator polynomial of a cyclic descendant code.  The second part of the paper shows that cyclic and quasi-cyclic descendant LDPC codes can be derived from cyclic finite geometry LDPC codes using the results developed in first part of the paper.  This enlarges the repertoire of cyclic LDPC codes. The third part of the paper analyzes the trapping sets of regular LDPC codes whose parity-check matrices satisfy a certain constraint on their rows and columns.  Several classes of finite geometry and finite field cyclic and quasi-cyclic LDPC codes with large minimum weights are shown to have no harmful trapping sets with size smaller than their minimum weights. Consequently, their performance error-floors are dominated by their minimum weights.
 
%%show the relations between cyclic codes and quasi-cyclic codes. From this point, we propose to obtain one side from another side by row and column permutations. By such permutations, we construct a class of quasi-cyclic LDPC codes from cyclic finite geometry codes, which contains [1] and [2] as a special case and a class of cyclic LDPC codes which greatly enriches the family of cyclic LDPC codes. Moreover, we reveal the re
%%%We also point out this class of codes is actually a class of shortened finite geometry codes.

%{\bf Keywords: }Low density parity-check code, cyclic, quasi-cyclic, RS codes, BCH codes
\end{abstract}

%Its $i$-th cyclic shift is denoted by $\psi^i ({\bf o}) $ 

\section{Introduction}

   The rapid dominance of LDPC codes [1] in applications requiring error control coding is due to their capacity-approaching performance which can be achieved with practically implementable iterative decoding algorithms.  LDPC codes were first discovered by Gallager in 1962 [1] and then rediscovered in late 1990's [2], [3].  Ever since their rediscovery, a great deal of research effort has been expended in design, construction, structural analysis, efficient encoding and decoding, performance analysis, generalizations and applications of LDPC codes.  Numerous papers have been published on these subjects.  Many LDPC codes have been chosen as the standard codes for various next generations of communication systems and their applications  to digital  data storage systems are now being seriously considered and investigated.

   Let GF($q$) be a field with $q$ elements. A \emph{regular} $q$-ary LDPC code [1] is given by the null space over GF($q$) of a \emph{sparse} parity-check matrix $\bf H$ that has constant column weight $\gamma$  and constant row weight $\rho$.  Such an LDPC code is said to be ($\gamma$,$\rho$)-regular. If the columns and/or rows of $\bf H$ have \emph{varying} weights, then the null space of $\bf H$ gives an \emph{irregular} $q$-ary LDPC code.  If $\bf H$ is an array of sparse circulants of the same size, then the null space over GF($q$) of $\bf H$ gives a $q$-ary \emph{quasi-cyclic} (QC)-LDPC code.  If $\bf H$ consists of a single sparse circulant or a column of sparse circulants of the same size, then the null space of $\bf H$ over GF($q$) gives a \emph{cyclic LDPC code}.  If $q = 2$, an LDPC code is said to be binary.

   In almost all of the proposed constructions of LDPC codes, the following constraint on the rows and columns of the parity-check matrix $\bf H$ is imposed: \emph{no two rows (or two columns) can have more than one place where they both have non-zero components}.  This constraint on the rows and columns of $\bf H$ is referred to as the \emph{row-column (RC)-constraint}.  This RC-constraint ensures that the Tanner graph [4] of the LDPC code given by the null space of $\bf H$ is free of cycles of length 4 and hence has a girth of at least 6 and that the minimum distance of the code is at least $\gamma_{\min} + 1$, where $\gamma _{\min}$ is the minimum column weight of $\bf H$ [5], [6]. The distance bound $\gamma _{\min} + 1$ is poor for small $\gamma_{\min}$ and irregular LDPC codes, but it is tight for regular LDPC codes  whose parity-check matrices have large column weights, such as finite geometry LDPC codes [5]-[9], and finite field QC-LDPC codes constructed in [10]-[13] and this paper.  A parity-check matrix  $\bf H$  that satisfies the RC-constraint is called an RC-constrained parity-check matrix and the code given by its null space is called an RC-constrained LDPC codes.  An RC-constrained LDPC code is one-step majority-logic decodable [5], [6].  Furthermore, the RC-constraint on the parity-check matrices of  LDPC codes allows us to analyze the \emph{trapping-set} structure  [14], [15] of RC-constrained LDPC codes which affects their error-floor performances. Analysis of trapping-set structure of RC-constrained LDPC codes is a part of investigation in this paper.
 
   LDPC codes can be classified into two general categories: 1) random or pseudo-random codes that are constructed using computer-based algorithms or methods; and 2) algebraic codes that are constructed using algebraic or combinatorial tools such as finite fields, finite geometries and experimental designs.  Codes in these two categories can be classified into two types, codes whose parity-check matrices possess little structure and codes whose parity-check matrices have structures.  A code whose parity-check matrix possesses no structure beyond being a linear code is problematic in that both encoding and decoding implementations become quite complex.  A code whose parity-check matrix has structures beyond being a linear code is in general more easily implemented. Two desirable structures for hardware implementation of encoding and decoding of LDPC codes are cyclic and quasi-cyclic structures.  A cyclic LDPC code can be efficiently and systematically encoded with a single feedback shift-register with complexity linearly proportional to the number of parity-check symbols (or information symbols) [6].  Encoding of a QC-LDPC code can also be efficiently implemented but requires multiple shift-registers [16], [17].  It is in general more complex than encoding of a cyclic code but still enjoys linear complexity.  However, QC-LDPC codes enjoy some advantages in hardware implementation of decoding in terms of wire routing [18].  Furthermore, the QC structure allows partially parallel decoding [19] which offers a trade-off between decoding complexity and decoding speed, while cyclic structure allows either full parallel or serial decoding.   In this paper, we show that a cyclic LDPC code can be put in QC form through column and row permutations.  As a result, a cyclic LDPC code enjoys both encoding and decoding implementation advantages.  Encoding is carried out in cyclic form while decoding is carried out in QC form.

   QC-LDPC codes are more commonly studied than cyclic LDPC codes. There are at least a dozen of or more methods for constructing QC-LDPC codes, including both algebraic and computer-based methods; however, there is only one known class of cyclic LDPC codes which are constructed based on finite geometries [5].

    This paper is concerned with constructions and structural analysis of both cyclic and QC codes, particularly LDPC codes. It consists of three parts.  In the first part, it is shown that a cyclic code given by a parity-check matrix in circulant form can be decomposed, through column and row permutations, into various cyclic and QC codes, called \emph{descendant codes}. Some fundamental structures of the descendant codes are developed, including the characterization of the roots of the generator polynomial of a cyclic descendant code. In the second part of the paper, it is shown that RC-constrained cyclic and QC-LDPC codes can be derived from the class of cyclic finite geometry (FG) LDPC codes based on circulant decomposition presented in the first part. Several new families of RC-constrained cyclic and QC-LDPC codes are presented.  The third part of the paper is concerned with {trapping sets} of RC-constrained regular LDPC codes.  It is shown that for an RC-constrained ($\gamma$,$\rho$)-regular LDPC code, its Tanner graph has no trapping sets of size smaller than or equal to $\gamma$ with numbers of odd-degree check-nodes less than or equal to $\gamma$.   Several classes of cyclic and QC-LDPC codes are shown to have large minimum distances (or minimum weights) and no \emph{elementary} trapping sets [20] with sizes and numbers of degree-1 check-nodes smaller than their minimum weights.

\section{Circulant Decomposition, Cyclic and Quasi-cyclic Codes}

   A circulant is a square matrix over a certain field such that every row is the cyclic-shift one place to the right (or one place to the left) of the row above it and the first row is the cyclic-shift one place to the right (or one place to the left) of the last row. In coding theory, a cyclic-shift commonly refers to the cyclic-shift one place to the right. Hereafter, by a cyclic-shift, we mean a cyclic-shift one place to the right unless explicitly mentioned otherwise.  In this case, every column of a circulant is a downward cyclic-shift the column on its left and the first column is the downward cyclic-shift of the last column.  It is clear that a circulant is uniquely specified (or characterized) by its first row which is called the \emph{generator} of the circulant. The columns and rows of a circulant have the same weight.

\subsection{Circulant Decomposition}
   Let $\bf W$ be an $n\times n$  circulant over the field GF($q$) where $q$ is a power of a prime. We label the rows and columns of $\bf W$ from $0$ to $n - 1$.  Let ${\bf w} = (w_0, w_1, . . . , w_{n-1})$ be the generator of $\bf W$.  We denote $\bf W$ by $\Psi({\bf w}) = \Psi (w_0, w_1, . . . , w_{n-1})$.  Then
\begin{equation}
 {\bf W} = \Psi ({\bf w})=\left[\begin{array}{cccc}
 w_0 & w_1& \cdots& w_{n-1}\\
 w_{n-1}&w_0&\cdots&w_{n-2}\\
 \vdots&\vdots&\ddots&\vdots\\
 w_1&w_2&\cdots&w_0
 \end{array}\right].
\end{equation}
Let $\Psi^{(1)} ({\bf w})$ denote the circulant obtained by simultaneously cyclically shifting all the rows of $\Psi ({\bf w})$ one place to the right.  Let ${\bf w}^{(1)}$ denote the $n$-tuple obtained by cyclic-shifting all the components of $\bf w$ one place to the right.  Then, it is clear that $\Psi^{(1)} ({\bf w}) =\Psi ({\bf w}^{(1)})$. Note that $\Psi ({\bf w})$ and $\Psi ({\bf w}^{(1)})$ have identical set of rows and identical set of columns except that all the columns are cyclically shifted one place to the right and all the rows are cyclically shifted upward one place.  Therefore, $\Psi ({\bf w})$ and $\Psi ({\bf w}^{(1)})$ are isomorphic up to cyclic-shift.

   Suppose $n$ can be factored as a product of two positive integers, $c$ and $l$, such that $c \neq  1$ and $l \neq  1$, i.e., $n = c\cdot l$ and $c$ and $l$ are proper factors of $n$.  Let ${\cal I}   = \{0, 1, 2, \cdots, c\cdot l - 1\}$ be the set of indices (or labels) for the rows and columns of the $n\times n$  circulant $\Psi ({\bf w})$ given by (1). Define the following index sequences:
\begin{eqnarray}
                  \pi ^{(0)} = [0, c, 2c, \cdots, (l - 1)c ], \\
                  \pi  = [\pi ^{(0)},  \pi ^{(0)} + 1, \cdots ,  \pi ^{(0)} + c - 1]. 
                  \end{eqnarray}
Then, $\pi$  gives a permutation of the indices in ${\cal I}$. Suppose we first permute the columns and then the rows of $\bf W$ based on $\pi$. These column and row permutations based on $\pi$  result in the following  $c\times c$  array of circulants of size $l\times l$ over GF($q$):
\begin{equation}
\Phi({\bf w})=\left[\begin{array}{ccccc}
 \Psi ({\bf w}_0)    &            \Psi ({\bf w}_1)    &        \cdots   &    \Psi ({\bf w}_{c - 2})&    \Psi ({\bf w}_{c - 1})\\  
  \Psi^{(1)} ({\bf w}_{c-1})    &            \Psi ({\bf w}_0)    &     \cdots   &    \Psi ({\bf w}_{c - 3}) &     \Psi ({\bf w}_{c - 2})\\ 
       \vdots    &             \vdots   &      \ddots    &     \vdots  &     \vdots \\ 
           \Psi^{(1)}  ({\bf w}_{2})    &            \Psi^{(1)}  ({\bf w}_{3})    &         \cdots   &    \Psi ({\bf w}_{0}) &     \Psi ({\bf w}_{1})\\ 
                      \Psi^{(1)}  ({\bf w}_{1})    &            \Psi^{(1)}  ({\bf w}_{2})    &           \cdots   &    \Psi^{(1)}  ({\bf w}_{c-1}) &     \Psi ({\bf w}_{0})\\       
\end{array}\right],
\end{equation}
where, for $0 \leq  i < c$,
\begin{eqnarray}
                        {\bf w}_i &=& (w_i, w_{c + i}, \cdots, w_{(l - 1)c + i}),\\                
                            \Psi ({\bf w}_i) &=& \left[\begin{array}{cccc}
                                     { w}_i  &         { w}_{c + i}   &    \cdots  &   { w}_{(l - 1)c + i}\\
                             {w}_{(l - 1)c + i}  &      {w}_i  &       \cdots  &   { w}_{(l - 2)c + i}\\
                             \vdots&\vdots&\ddots&\vdots\\
                          {w}_{c + i}&       {w}_{2c + i}     &      \cdots   &         { w}_i
                                  \end{array}\right].                                  
                                  \end{eqnarray}
Each $l\times l$ circulant $\Psi ({\bf w}_i)$ (or $\Psi ({\bf w}^{(1)} _i)$) in $\Phi ({\bf w})$ is called a \emph{descendant} circulant of $\Psi ({\bf w})$. Since $\Psi ({\bf w}_i)$ and $\Psi ({\bf w}^{(1)} _i)$ are isomorphic for $0 \leq i < c$, there are at most $c$ distinct (or non-isomorphic) descendant circulants of $\Psi ({\bf w})$ in $\Phi ({\bf w})$, namely $\Psi ({\bf w}_0), \Psi ({\bf w}_1), \cdots, \Psi ({\bf w}_{c-1}) $. The $l$-tuple ${\bf w}_i$ is called the $i$-th \emph{cyclic section} of $\bf w$.
%and $\Psi ({\bf w}^{(1)} _i)$ is obtained by simultaneously cyclically shifting all the rows of $\Psi ({\bf w}_i)$ one place to the right.  Each $l\times l$ circulant in $\Phi ({\bf w})$ is called a \emph{descendant} of $\Psi ({\bf w})$.  Note that $\Psi ({\bf w}_i)$ and $\Psi ({\bf w}^{(1)} _i)$ have the same set of rows except that the rows of $\Psi ({\bf w}^{(1)}_i)$ are obtained by permuting the rows of $\Psi ({\bf w}_i)$ cyclically one place upward.

   Since $\Phi ({\bf w})$ is obtained by applying the permutation $\pi$  to the columns and rows of the ciruclant $\Psi ({\bf w})$, we write  $\Phi ({\bf w}) = \pi (\Psi ({\bf w}))$.  Let $\pi^{-1}$ be the inverse permutation of $\pi$. Then $\Psi ({\bf w}) = \pi ^{-1}(\Phi ({\bf w}))$.  \emph{From the structure of $\Phi ({\bf w})$ displayed by (4), we see that each row of $l\times l$ circulants is a right cyclic-shift of the row above it, however, when the last circulant on the right is shifted around to the left, all its rows are cyclically shifted one place to the right within the circulant.} This structure is referred to as the \emph{doubly cyclic} structure which is pertinent to the construction of new cyclic codes, especially new cyclic LDPC codes, as will be shown in later sections.   From the expression of (4), we see that the descendant circulant $\Psi ({\bf w}_0)$ of $\Psi ({\bf w})$ appears in the array $\Phi ({\bf w})$  $c$ times on the main diagonal.  For $1 \leq  i < c$, the descendent circulant $\Psi ({\bf w}_i)$ appears $i$ times and its shift $\Psi ^{(1)}({\bf w}_i)$ (or $\Psi ({\bf w}^{(1)}_i)$) appears $c - i$ times in $\Phi ({\bf w})$.  $\Psi ({\bf w}_i)$ and its shifts appear on an off-diagonal of $\Phi ({\bf w})$ starting from the $i$th position of the first row and moving down to the right on a 45$^\circ$ diagonal.  When it reaches to the last (rightmost) column of $\Phi ({\bf w})$, it moves to the left of the next row of $\Phi ({\bf w})$ and continues to move down on a 45$^\circ$ diagonal until it reaches the last row of $\Phi ({\bf w})$.
   
      Summarizing the above results, we have the following theorem.
\begin{theorem}  
Given an $n\times n$  circulant ${\bf W}=\Psi({\bf w})$ over a field with generator $\bf w$, if $n$ can be properly factored, then there is a permutation $\pi$  which puts $\bf W$ into an array of circulants of the same size in the form of (4).  Conversely, if an array $\Phi ({\bf w})$ of circulants of the same size is given in the form (4), then there is a permutation $\pi ^{-1}$ which puts the array $\Phi ({\bf w})$ into a circulant $\bf W$ with generator $\bf w$.     
\end{theorem}      %\hspace{\fill}\IEEEQED 
Theorem 1 gives a basis for decomposing a cyclic code into families of cyclic and QC codes or putting a group of cyclic codes into a longer cyclic code.

\subsection{Cyclic and QC Descendants of a Cyclic Code}

   In the following, we show that cyclic and QC codes can be derived from a given cyclic code using circulant decomposition.  The results developed in this section will be used in Section IV to construct new cyclic and QC-LDPC codes from cyclic FG-LDPC codes.

   Let ${\cal C}_c$ be an ($n$,$n-r$) cyclic code over GF($q$) given by the null space of an $n\times n$ circulant parity-check matrix ${\bf H}_{circ} = \Psi ({\bf w})$ over GF($q$) with rank $r$ where $\bf w$ is the generator of the circulant.  (For every cyclic code, a circulant parity-check matrix ${\bf H}_{circ}$  can always be constructed by using its parity-check vector as the generator $\bf w$ of the circulant [15].  This will be reviewed in the next section.)  Suppose $n$ can be properly factored as the product of two integers, $c$ and $l$.  Then, as shown in Section II. A,  the circulant parity-check matrix ${\bf H}_{circ}  = \Psi ({\bf w})$ of ${\cal C}_c$ can be decomposed as a $c\times c$ array ${\bf H}_{qc}$ of circulants of size $l\times l$ in the form given by (4) through column and row permutations:
\begin{equation}
   {\bf H}_{qc} = \Phi ({\bf w})  =   
   \left[\begin{array}{ccccc}
       \Psi ({\bf w}_0) &   \Psi ({\bf w}_1)  &         \Psi ({\bf w}_2)  & \ldots &    \Psi ({\bf w}_{c -1})\\                        \Psi ^{(1)} ({\bf w}_{c - 1}) &        \Psi ({\bf w}_0)  &         \Psi ({\bf w}_1) &  \ldots&   \Psi ({\bf w}_{c- 2})\\
       \Psi ^{(1)} ({\bf w}_{c - 2}) &     \Psi ^{(1)} ({\bf w}_{c -1}) &    \Psi ({\bf w}_0)&  \ldots &    \Psi ({\bf w}_{c - 3})\\
       \vdots&\vdots&&\ddots&\vdots\\
       \Psi ^{(1)} ({\bf w}_1)  &        \Psi ^{(1)} ({\bf w}_2) &     \Psi ^{(1)} ({\bf w}_3) &  \ldots & \Psi ({\bf w}_0)
                             \end{array} \right], 
 \end{equation}                                                           
 where, for $0 \leq  i < c$, ${\bf w}_i$  and $\Psi ({\bf w}_i )$ are given by (5) and (6). Then, the null space of ${\bf H}_{qc} = \Phi ({\bf w})$  gives an ($n$,$n-r$)  QC code ${\cal C}_{qc}$ over GF($q$) which is \emph{combinatorially equivalent} to ${\cal C}_c $.  We say that $\{ {\cal C}_c, {\cal C}_{qc}\}$ form an equivalent pair. Notation-wise, we express ${\cal C}_{qc}$ and ${\cal C}_c$  as ${\cal C}_{qc} = \pi( {\cal C}_c)$ and ${\cal C}_c = \pi^{-1}({\cal C}_{qc})$, respectively.

   From the array ${\bf H}_{qc} = \Phi ({\bf w}) $, we can construct new cyclic codes of three different types.  These new cyclic codes are called \emph{cyclic descendant codes} (simply \emph{descendants}) of the cyclic code ${\cal C}_c$.  The cyclic code ${\cal C}_c$ itself is called the \emph{mother} code.

 For $0 \leq  i < c$, if $\Psi ({\bf w}_i )$ is a nonzero circulant, then the null space over GF($q$)  of $\Psi ({\bf w}_i )$ gives a cyclic descendant of ${\cal C}_c$, denoted by ${\cal C}_i ^{(1)}$ , of length $l$.  This descendant code is referred to as a \emph{type-1 cyclic descendant} of ${\cal C}_c$. Since there are at most $c$ distinct non-isomorphic descendant circulants of ${\bf H}_{circ}  = \Psi ({\bf w})$  in the array ${\bf H}_{qc} = \Phi ({\bf w}) $.  There are at most $c$ distinct type-1 cyclic descendants of ${\cal C}_c$.

    From (7), we see that each column of the array ${\bf H}_{qc} = \Phi ({\bf w})$  consists of the circulants in the first row of ${\bf H}_{qc}$. For $0 \leq  i < c$, each circulant $\Psi ({\bf w}_i )$ or its cyclic shift $\Psi ^{(1)}({\bf w}_i )$ appears once and only once.  Since a circulant $\Psi ({\bf w}_i )$ and its cyclic shift $\Psi ^{(1)} ({\bf w}_i )$ differ only in permutation of their rows and hence their null spaces are identical.  Consequently,  the null spaces of all the columns of ${\bf H}_{qc} = \Phi ({\bf w})$  are the same.  In fact, the null space of each column of ${\bf H}_{qc} = \Phi ({\bf w})$  is identical to the null space of the following $cl\times l$ matrix:
\begin{equation*}        
{\bf H}_{col}=
\left[
\begin{array}{c}
                                         \Psi  ({\bf w}_0)\\
                                         \Psi  ({\bf w}_1)\\
                                                 \vdots\\
                                       \Psi  ({\bf w}_{c-1})
\end{array}\right].
\end{equation*}
For $1 \leq  k < c$, let $i_1, i_2, \ldots, i_k$ be $k$ distinct integers such that $0 \leq  i_1, i_2, \ldots, i_k < c$.  Let
\begin{equation}        
    {\bf H}_{col, k} =  \left[  
\begin{array}{c}
                   \Psi  ({\bf w}_{i_1} )\\
                                         \Psi ({\bf w}_{i_2} )\\
\vdots\\
                                        \Psi  ({\bf w}_{i_k})
\end{array}\right],
\end{equation}
which is a submatrix of ${\bf H}_{col}$. The null space of   ${\bf H}_{col, k}$ gives a cyclic code of length $l$, denoted by ${\cal C}_k^{(2)}$, which is referred to as a \emph{type-2 cyclic descendant} of the mother cyclic code ${\cal C}_c$.

   For $1 \leq  k < c$, let $i_1, i_2, \ldots, i_k$ be a set of distinct integers such that $0 \leq  i_1, i_2, \ldots, i_k < c$.   Suppose we replace the descendant circulants, $\Psi ({\bf w}_{i_1}), \Psi ({\bf w}_{i_2}), \ldots, \Psi ({\bf w}_{i_k})$ of ${\bf H}_{circ} = \Psi ({\bf w})$  and all their cyclic shifts in the array ${\bf H}_{qc} = \Phi ({\bf w})$ (see (7))  by zero matrices of size $l\times l$ (if $i_1 = 0$, we replace $c$ copies of the circulant, $\Psi ({\bf w}_0)$, by $c$ zero matrices).   By doing this, we obtain a $c\times c$ array ${\bf H}_{qc,mask} = \Phi ({\bf w})_{mask}$ of circulants and zero matrices of size $l\times l$.  Since the cyclic shift of a zero matrix is also a zero matrix, the array $\Phi ({\bf w}) _{mask}$ is still in the form given by (4).  Then ${\bf H}_{circ,mask} = \Psi ({\bf w})_{mask} = \pi^{-1}(\Phi ({\bf w})_{mask})$ gives a new $n\times n$ circulant over GF($q$) .  Let $r_{mask}$ be the rank of ${\bf H}_{circ,mask} = \Psi ({\bf w})_{mask}$. Then the null space of ${\bf H}_{circ,mask} = \Psi ({\bf w})_{mask}$ gives an ($n$,$n - r_{mask}$) cyclic code ${\cal C}_{mask}^{(3)}$ which is referred to as a \emph{type-3 cyclic descendant} of the mother cyclic code ${\cal C}_c$. The replacement of a set of circulants in the array ${\bf H}_{qc} = \Phi ({\bf w})$ by a set of zero matrices  is called \emph{masking} [6], [10], [11]. ${\bf H}_{circ,mask} = \Psi ({\bf w})_{mask}$  and ${\bf H}_{qc,mask}= \Phi ({\bf w})_{mask}$ are called \emph{masked circulant} and \emph{masked array} of ${\bf H}_{circ}  = \Psi ({\bf w})$   and  ${\bf H}_{qc} = \Phi ({\bf w}) $, respectively.  It is clear that different masking pattern results in a different cyclic descendant code of ${\cal C}_c$. In Section III, we will characterize the roots of the generator polynomials of cyclic descendant codes of all three types.

   For any pair ($s$,$t$) of integers with $1 \leq  s, t \leq  c$, let ${\bf H}_{qc}(s,t)$ be a $s\times t$ subarray of ${\bf H}_{qc} = \Phi ({\bf w}) $.  Since ${\bf H}_{qc}(s,t)$ is an array of circulants, its null space gives a QC code.  This QC code is called a QC descendant code of ${\cal C}_c$ (or ${\cal C}_{qc})$.

\subsection{Cyclic- and QC-LDPC Codes Derived From a Cyclic LDPC Code}
   If the circulant parity-check matrix ${\bf H}_{circ}  = \Psi ({\bf w})$  of ${\cal C}_c$  is a sparse circulant over GF($q$)  and satisfies the RC-constraint, then the null space of ${\bf H}_{circ}  = \Psi ({\bf w})$  gives an RC-constrained cyclic-LDPC code over GF($q$) .  Since the $c\times c$ array  ${\bf H}_{qc} = \Phi ({\bf w})$  is obtained from ${\bf H}_{circ}  = \Psi ({\bf w})$  by column and row permutations, it also satisfies the RC-constraint.  Hence, the null space of ${\bf H}_{qc} = \Phi ({\bf w})$  gives an RC-constrained QC-LDPC code ${\cal C}_{qc}$ which is equivalent to the cyclic LDPC code ${\cal C}_c$.  Since the entire array ${\bf H}_{qc} = \Phi ({\bf w})$  satisfies the RC-constraint, any subarray of ${\bf H}_{qc} = \Phi ({\bf w})$  also satisfies the RC-constraint.  Consequently, all the cyclic descendant codes derived from the cyclic-LDPC code ${\cal C}_c$ are cyclic-LDPC codes, i.e., the null space of the $i$th descendant circulant $\Psi ({\bf w}_i )$ (or $\Psi ^{(1)} ({\bf w}_i )$) of ${\bf H}_{circ}  = \Psi ({\bf w})$  in the array ${\bf H}_{qc} = \Phi ({\bf w})$  gives a cyclic-LDPC code of length $l$, the null space of the parity-check matrix ${\bf H}_{col,k}$ given by (8) gives a cyclic-LDPC code of length $l$, and the null space of a $c\times c$ masked circulant ${\bf H}_{circ,mask} = \Psi ({\bf w})_{mask}$  of ${\bf H}_{circ}  = \Psi ({\bf w})$  gives a cyclic-LDPC code of length $n$.  The Tanner graphs of the cyclic descendant LDPC codes of  ${\cal C}_c$ have a girth of length at least 6.

   For any pair ($s$,$t$) of integers with $1 \leq  s, t \leq  c$, let ${\bf H}_{qc}(s,t)$ be a $s\times t$ subarray of ${\bf H}_{qc} = \Phi ({\bf w}) $.  Then the null space of ${\bf H}_{qc}(s,t)$ gives a QC-LDPC code whose Tanner graph has a girth of at least 6.

   Among the classes of LDPC codes that have been constructed or designed, the only class of LDPC codes that are cyclic is the class of finite geometry (FG) LDPC codes [5] whose parity-check matrices are circulants and satisfy the RC-constraint.  Cyclic FG-LDPC codes have large minimum distances (or weights) and perform well with iterative decoding based on belief propagation. Cyclic-LDPC codes constructed based on two-dimensional \emph{projective} geometries have been proved that their Tanner graphs do not have trapping sets of sizes smaller than their minimum weights [20].  As a result, their error-floors are mainly determined by their minimum weights.  Since they have large minimum weights, their error-floors are expected to be very low.  In Section VII, we will show that the Tanner graphs of the cyclic-LDPC codes constructed based on two-dimensional \emph{Euclidean} geometries also do not have trapping sets with sizes smaller than their minimum weights.  Unfortunately, cyclic FG-LDPC codes form a small class of cyclic-LDPC codes.  However, using circulant decomposition presented in this section, we can construct large classes of cyclic and QC descendant LDPC codes from cyclic FG-LDPC codes, as will be shown in Sections IV, V and VI.  These cyclic and QC descendant LDPC codes of cyclic FG-LDPC codes also have good trapping set structures.

   Construction of QC-EG-LDPC codes through decomposition of a single circulant constructed based on lines of a two-dimensional Euclidean geometry was proposed earlier by Kamiya and Sasaki [9].  In this paper, their focus was mainly on construction of high rate QC-LDPC codes and analysis of the ranks of their parity-check matrices.  In this paper, we propose constructions of both cyclic- and QC-LDPC codes through decomposition of  a single or \emph{multiple} circulants constructed based on two and \emph{higher} dimensional Euclidean and projective geometries.  We particularly emphasize on construction of cyclic LDPC codes and characterization of the roots of their generator polynomials.

\section{Decomposition of Cyclic Codes and Characterization of Their Cyclic Descendants} 

   In this section, we first show that a circulant parity-check matrix of a given cyclic code can be expressed as a linear sum of circulants which correspond to the roots of the generator polynomial of the given code. From this linear sum of circulants, we then characterize the roots of the generator polynomials of the cyclic descendants of the given cyclic code.

\subsection{Circulant Parity-Check Matrices of Cyclic Codes}

   For any positive integer $m$, let GF($q^m$) be an extension field of GF($q$).  Let ${\cal C}_{c}$ be an ($n$,$k$) cyclic code over GF($q$) where $n$ is a factor of $q^m - 1$ and $(n,q) = 1$.  Every codeword ${\bf v} = (v_0, v_1, \cdots , v_{n-1})$ in ${\cal C}_c$ is represented by a polynomial ${\bf v}(X) = v_0 + v_1 X + \cdots + v_{n-1}X^{n-1}$ over GF($q$) with degree $n - 1$ or less.  The polynomial ${\bf v}(X)$ is called a code polynomial.  An ($n$,$k$) cyclic code ${\cal C}_c$ over GF($q$) is uniquely specified by its generator polynomial ${\bf g}(X) = g_0 + g_1 X + \cdots + g_{n-k-1}X^{n-k-1} + X^{n-k}$ which is a monic polynomial of degree $n-k$ over GF($q$) and divides $X^n - 1$ [6], [21]-[24] where $g_0 \neq 0$. A polynomial of degree $n - 1$ or less over GF($q$) is a code polynomial if and only if it is divisible by ${\bf g}(X)$.  Hence, every code polynomial ${\bf v}(X)$ is a multiple of ${\bf g}(X)$. 

   The generator polynomial  ${\bf g}(X)$  of ${\cal C}_c$ has $n - k$ roots in GF($q^m$).   The condition $(n,q) = 1$ ensures that all the roots of $X^n - 1$ are distinct elements of GF($q^m$) and hence all the roots of  ${\bf g}(X)$  are distinct elements of GF($q^m$).  In the construction of a cyclic code, its generator polynomial is often specified by its roots.  This is the case for BCH and RS codes [6], [21]-[24]. 

   Let
\begin{eqnarray} 
                       {\bf h}(X) &=& (X^n - 1)/ {\bf g}(X)\nonumber\\     
                                   &=& h_0 + h_1 X + \cdots + h_k X^k
                                   \end{eqnarray}                                           
where $h_j\in \mbox{GF($q$)}$ for $0\leq j\leq k$, $h_k = 1$ and $h_0\neq 0$. The polynomial ${\bf h}(X)$ is called the \emph{parity-check polynomial} of $\cal C$.  Let 
\begin{eqnarray} 
                 { \tilde{\bf h}}(X)&=& \tilde{h}_0 + \tilde{h}_1 X + \cdots + \tilde{h}_k X^k     \nonumber\\                  
                           &=& X^k {\bf h}(X^{-1}) = h_k + h_{k-1}X + \cdots + h_0 X^k,
                            \end{eqnarray}  
which is the \emph{reciprocal polynomial} of ${\bf h}(X)$.  Comparing the coefficients of ${ \tilde{\bf h}}(X)$ and ${\bf h}(X)$, we have
\begin{equation}
                    \tilde{h}_0 = h_k,\quad    \tilde{h}  _1 = h_{k-1},  \quad   \cdots ,  \quad    \tilde{h} _k = h_0.                     \end{equation}
Form the following $n$-tuple over GF($q$):
\begin{equation}
                     \tilde{\bf h}   = ( \underbrace{\tilde{h}  _0,  \tilde{h}  _1, \cdots ,  \tilde{h}  _k}_{k+1},  \tilde{h}  _{k+1}, \cdots,  \tilde{h}  _{n-1}),                            
                     \end{equation}
where the first $k+1$ components are the coefficients of  $\tilde{\bf h}  (X)$ and last $n - k - 1$ components are zeros, i.e.,
\begin{equation}
\tilde{h}  _{k+1} =  \tilde{h}  _{k+2} = \cdots =  \tilde{h}  _{n-1} = 0.   
                     \end{equation}
Using the $n$-tuple $\tilde{\bf h}$ of (12) as the generator, we form the following $n\times n$  circulant over GF($q$):
\begin{equation}
{\bf H}_{circ}= \Psi ( \tilde{\bf h}  )=\left[
 \begin{array}{cccccc}
 \tilde{h}  _0     &  \tilde{h}  _1   &    \tilde{h}  _2   &  \cdots   &    \tilde{h}  _{n-2}  &   \tilde{h}  _{n-1}\\
  \tilde{h}  _{n-1} &    \tilde{h}  _0   &    \tilde{h}  _1&     \cdots  &     \tilde{h}  _{n-3} &    \tilde{h}  _{n-2}\\
  \vdots&\vdots&&\ddots&\vdots&\vdots\\
   \tilde{h}  _1     &  \tilde{h}  _2   &    \tilde{h}  _3   &  \cdots   &    \tilde{h}  _{n-1}  &   \tilde{h}  _{0}\\
\end{array}
 \right].
\end{equation}
In terms of the coefficients of ${\bf h}(X)$, $\Psi ( \tilde{\bf h}  )$ is given as follows:
\begin{equation}
 {\bf H}_{circ}= \Psi ( \tilde{\bf h}  )=\left[
 \begin{array}{cccccccccc}
  {h}  _k     &   {h}  _{k-1}   &     {h}  _{k-2}   &  \cdots  &       {h}  _1 &       {h}  _0&0&0\; & \cdots&0\\
 0     &   {h}  _{k}   &     {h}  _{k-1}&     {h}  _{k-2}   &  \cdots   &     {h}  _1  &    {h}  _0&0\; &\cdots&0\\
    \vdots&\vdots&& &\ddots&\vdots &\vdots && \ddots&\vdots\\
         0  &   0   &     0   & &  \cdots &       {h}_{k} &  {h}_{k-1} &  &\cdots&h_0\\
         \hdashline
     {h}  _0  &   0   &     0  &  &  \cdots    &   0&  {h}_{k} &  {h}_{k-1}  &\cdots&h_1\\      
   \vdots&\vdots& &&\ddots&\vdots &\vdots & &\ddots&\vdots\\
 {h}  _{k-1}   &     {h}  _{k-2} &     {h}  _{k-3}  &  \cdots   &  h_0 &     0&0 &&\cdots &  {h}  _k  \\
\end{array}
 \right].
\end{equation}
   The first $n - k$ rows of ${\bf H}_{circ}= \Psi ( \tilde{\bf h}  )$ are linearly independent which give the conventional parity-check matrix  $\bf H$  of the ($n$,$k$) cyclic code ${\cal C}_c$.  The other $k$ rows of ${\bf H}_{circ}= \Psi ( \tilde{\bf h}  )$ are \emph{redundant rows} (or linearly dependent on the the first $n - k$ rows). Since  ${\bf H}_{circ}$  is a  \emph{redundant expansion} of $\bf H$, the null spaces of  $\bf H$  and  ${\bf H}_{circ}$  give the same  cyclic code ${\cal C}_c$. The $n$-tuple $ \tilde{\bf h} = (h_k, h_{k-1}, \cdots , h_0, 0, 0, \cdots, 0)$ is commonly referred to as the \emph{parity-check vector}.

Note that every row (or every column) of the circulant parity-check matrix ${\bf H}_{circ} = \Psi(\tilde{\bf h} )$ of ${\cal C}_c$ has a \emph{zero-span} of length $n-k-1$ (i.e., $n-k-1$ consecutive zeros).  It is proved in [25] that this zero-span has maximum length and is unique.  The maximum zero-spans of different rows of ${\bf H}_{circ}$ start from different positions (or different columns).  It is shown in [25] that using the parity-check matrix in circulnat form, an ($n$,$k$) cyclic code $\cal C$ can correct bursts of errors up to the code's burst-correction capability or it can correct any burst of erasures of length $n-k$ or less using iterative decoding [15], [25]. Decomposition of a burst-error correction cyclic codes gives new burst-error correction cyclic descendant codes.

   Suppose that $n$ can be properly factored as the product of two positive integers, $c$ and $l$. Then $\Psi ( \tilde{\bf h}  )$ can be decomposed into a  $c\times c$  array of $l\times l$ circulants in the form given by (4) by applying the permutation $\pi$  (defined by (3)) to the columns and rows of  $\Psi ( \tilde{\bf h}  )$,
\begin{equation}
 \Phi ( \tilde{\bf h}  ) = \left[\begin{array}{ccccc}
   \Psi ( \tilde{\bf h}_0)       &      \Psi ( \tilde{\bf h}  _1)   &     \cdots   &   \Psi ( \tilde{\bf h}  _{c - 2})&   \Psi ( \tilde{\bf h}  _{c - 1})\\
         \Psi^{(1)} ( \tilde{\bf h}   _{c - 1})   &   \Psi ( \tilde{\bf h}  _0)   &      \cdots &   \Psi ( \tilde{\bf h}  _{c - 3}) &    \Psi ( \tilde{\bf h}  _{c - 2})\\
         \vdots&&\ddots&\vdots&\vdots\\
             \Psi^{(1)} ( \tilde{\bf h}  _{1})   &   \Psi  ^{(1)}( \tilde{\bf h}  _2)   &      \cdots  &   \Psi^{(1)} ( \tilde{\bf h}  _{1}) &    \Psi ( \tilde{\bf h}  _{0})    
 \end{array}\right],
\end{equation}
where, for $0 \leq  j < c$,
\begin{eqnarray}
                        { \tilde{\bf h}}_j &=& (\tilde{h}_j, \tilde{h}_{c + j}, \cdots, \tilde{h}_{(l - 1)c + j}),\\
\tilde{h}_t&=&h_{k-t}  \quad  \mbox{ for } 0\leq t\leq k,\\
\tilde{h}_t&=&0  \qquad \mbox{ for } t>k.                                 
                                  \end{eqnarray}                                  
   The null space of $\Phi ( \tilde{\bf h}  )$ gives a QC code ${\cal C}_{qc}$ that is combinatorially equivalent to ${\cal C}_c$.

    In code construction, the generator polynomial  ${\bf g}(X)$  of an ($n$,$k$) cyclic code ${\cal C}_c$ over GF($q$) is specified by its roots [6], [21]-[24].  Let $\beta_0, \beta_1, \cdots , \beta_{n-k-1}$ be the roots of  ${\bf g}(X)$ .  Then
\begin{equation}
                            {\bf g}(X)  = \prod\limits_{0\leq i<n-k} (X - \beta_i).                                                  \end{equation}
Since  ${\bf g}(X) | X^n - 1$, $n|(q^m - 1)$ and $(n,q) = 1$, $\beta_0, \beta_1, \cdots , \beta_{n-k-1}$ are distinct nonzero elements of GF($q^m$). Let $\alpha$ be a primitive $n$th root of unity. Then, for $0 \leq  i < n - k$,  $\beta_i$ is a power of $\alpha$.  Since $\alpha^n = 1$, $(\beta _i)^n = 1$ for $0 \leq  i < n - k$. A polynomial ${\bf c}(X)$ of degree $n - 1$ or less over GF($q$)  is a code polynomial if and only if ${\bf c}(X)$ has  $\beta_0, \beta_1, \cdots , \beta_{n-k-1}$ as roots, i.e., ${\bf c}(\beta _i) = 0$ for $0 \leq  i < n - k$. 

   In terms of the roots of  ${\bf g}(X)$, the parity-check matrix of ${\cal C}_c$ generated by  ${\bf g}(X)$  is conventionally given by the following $(n - k)\times n$ matrix over GF($q^m$):
\begin{equation}
{\bf V}=\left[ \begin{array}{c}
\tilde{\bf v}_0\\
\tilde{\bf v}_1\\
\vdots\\
\tilde{\bf v}_{n-k-1}
\end{array}
\right]=
\left[ \begin{array}{ccccc}
 1    &     \beta _0      &     \beta _0^2      &    \cdots   &      \beta_0 ^{n-1}  \\
 1    &     \beta _1      &     \beta _1^2      &    \cdots   &      \beta_1 ^{n-1}  \\
\vdots&\vdots&&\ddots&\vdots\\
 1    &     \beta _{n-k-1}      &     \beta _{n-k-1}^2      &    \cdots   &      \beta_{n-k-1} ^{n-1}  \\
\end{array}
\right].
\end{equation}
The rows are linearly independent over GF($q^m$). An $n$-tuple over GF($q$), ${\bf c} = (c_0, c_1, \cdots , c_{n-1})$, is a codeword in ${\cal C}_c$ if and only if ${\bf c} \cdot {\bf V}^{\sf T} = {\bf 0}$.  This is to say that the null space over GF($q$) of  $\bf V$  gives the cyclic code ${\cal C}_c$.  The null spaces of the circulant parity-check matrix  ${\bf H}_{circ}$  and ${\bf V}$ give the same code ${\cal C}_c$. The parity-check matrix of ${\cal C}_c$ in the form of (21) is commonly used for algebraic decoding,  such as the Berlekamp-Massey algorithm for decoding BCH and RS codes [6], [21]-[24]. 

   In the following, we develop some structural properties of the circulant parity-check matrix  ${\bf H}_{circ}  = \Psi ( \tilde{\bf h}  )$ of ${\cal C}_c$. One such structural property is that  ${\bf H}_{circ}$  can be expressed in terms of the circulants formed by the rows of  $\bf V$.  For $0 \leq  i < n - k$, let
\begin{equation}
                       \tilde{\bf v}_i = (1, \beta_i, \beta_i^2, \cdots,\beta_i ^{n-1}).
\end{equation}
be the $i$th row of  $\bf V$  and $\Psi (\tilde{\bf v}_i)$ be the $n\times n$  circulant over GF($q^m$) with $\tilde{\bf v}_i$ as the generator.  Since $\tilde{\bf v}_0, \tilde{\bf v}_1, \cdots, \tilde{\bf v}_{n-k-1}$ are linearly independent, the circulants, $\Psi (\tilde{\bf v}_0), \Psi (\tilde{\bf v}_1), \cdots , \Psi (\tilde{\bf v}_{n-k-1})$, are also linearly independent (i.e., for $a_i \in \mbox{GF}(q^m)$ with $0\leq i<n-k$, $a_0 \Psi(\tilde{\bf v}_0) + a_1 \Psi(\tilde{\bf v}_1) + \cdots+a_{n-k-1} \Psi(\tilde{\bf v}_{n-k-1}) \neq 0$ unless $a_0=a_1=\cdots=a_{n-k}=0$).

   For $0 \leq  i < n - k$, let
                   \begin{equation}
                      \tilde{\bf v}_i(X) = 1 + \beta_i X + \beta ^2 _i X^2 + \cdots + \beta ^{n-1} _i X^{n-1}                                  \end{equation}
be the polynomial representation of $i$th row  $\tilde{\bf v}_i$ of  $\bf V$  and
\begin{equation}
                    {\bf v}_i(X) = \beta ^{n-1}_i + \beta ^{n-2}_iX + \cdots + \beta _i X^{n-2} + X^{n-1}                             \end{equation}
be the reciprocal of  $\tilde{\bf v} _i(X)$.  For $0 \leq  i < n - k$, since
\[
                    X^n - 1 = (X - \beta_i)( \beta ^{n-1}_i + \beta ^{n-2}_iX + \cdots + \beta _iX^{n-2} + X^{n-1}),
\]
then we have
\begin{equation}
{\bf v}_i (X) = \frac{ X^n - 1}{X - \beta _i}= \beta ^{n-1} _i + \beta ^{n-2} _iX + \cdots + \beta _iX^{n-2} + X^{n-1}. 
\end{equation}
It follows from (9), (20), partial-fraction expansion and (25) that the parity-check polynomial ${\bf h}(X)$ of ${\cal C}_c$ can be expressed as a linear combination of ${\bf v}_i(X)$s as follows:
\begin{eqnarray}
{\bf h}(X) &=&  \frac{ X^n - 1}{\prod\limits_{0\leq i<n-k} X - \beta _i}\nonumber\\
&=&\sum\limits^{n-k-1}_{i=0} \frac{ \sigma_i(X^n - 1)}{X - \beta _i}\nonumber\\
&=&\sum\limits^{n-k-1}_{i=0} \sigma_i{\bf v}_i(X)
\end{eqnarray}
where for $0 \leq  i < n - k$,
\begin{equation}
          \sigma_i = \left(  \prod\limits_{j = 0, j\neq i}^{n-k-1} (\beta_i - \beta_j)\right)^{-1} .  
 \end{equation}
Since $\beta _0, \beta_1, \cdots, \beta_{n-k-1}$ are distinct nonzero elements of GF($q^m$), all the coefficients, $\sigma_0, \sigma_1, \cdots, \sigma_{n-k-1}$, of the linear sum of (26) are nonzero.  

Summarizing the above results, we have the following theorem.

 \begin{theorem}
Let ${\cal C}_c$ be an ($n$,$k$) cyclic code over GF($q$) generated by  ${\bf g}(X)$  which has the following nonzero elements of GF($q^m$), $\beta_0, \beta_1, \cdots, \beta_{n-k-1}$, as roots.  For $0 \leq  i < n - k$, let ${\bf v}_i(X) = \beta ^{n-1}_i + \beta ^{n-2}_iX + \cdots + \beta_i X^{n-2} + X^{n-1}$.  Then the parity-check polynomial ${\bf h}(X)$ of ${\cal C}_c$ can be expressed as a linear sum of ${\bf v}_0 (X), {\bf v}_1 (X), \cdots, {\bf v}_{n-k-1}(X)$ as follows:
\begin{equation}
 {\bf h}(X)=\sum\limits^{n-k-1}_{i=0} \sigma_i{\bf v}_i(X),
 \end{equation}
where, for $0 \leq  i < n - k$,
\begin{equation}
          \sigma_i = (  \prod\limits_{j = 0, j\neq i}^{n-k-1} (\beta_i - \beta_j))^{-1} .  
 \end{equation}
 \end{theorem}

Replacing $X$ in (28) by $X^{-1}$, multiplying both sides by $X^{n-1}$, using (10) and (23), the expression of (28) can be put in the following form:
\begin{equation}
        X^{n-k-1} \tilde{\bf h}  (X) =  \sum\limits_{i=0}^{n-k-1}   \sigma _i \tilde{\bf v} _i (X).               
\end{equation}
The vector representation of the polynomial $X^{n-k-1} \tilde{\bf h}  (X)$ is
\begin{equation}
                  \tilde{\bf h}  ^{(n-k-1)} = (0, 0, \cdots, 0, h_k, h_{k-1}, \cdots , h_0),
\end{equation}
which is the $(n-k-1)$th right cyclic-shift of the vector representation $\tilde{\bf h}=(h_k, h_{k-1}, \cdots, h_0, 0, 0, \cdots, 0)$ of the reciprocal polynomial $\tilde{\bf h}(X)$ of the parity-check polynomial ${\bf h}(X)$ of ${\cal C}_c$.   Putting (30) in vector form, we have 
\begin{equation}
 \tilde{\bf h}^{(n-k-1)} = \sum\limits^{n-k-1}_{i=0} \sigma_i \tilde{\bf v}_i.  
\end{equation}
If we cyclically shift the components of all the vectors in (32) $k + 1$ places to the right, then we have
\begin{equation}
 \tilde{\bf h} = \sum\limits^{n-k-1}_{i=0} \sigma_i \tilde{\bf v}^{(k+1)}_i  
\end{equation}
where 
\begin{eqnarray}
               \tilde{\bf v}^{(k+1)}_i  &=& (\beta^{n-k-1} _i, \cdots , \beta^{n-1} _i, 1, \cdots, \beta^{n-k-2} _i)      \nonumber \\                 
                                   &=& \beta^{n-k-1} _i (1, \beta_i, \beta_i^2, \cdots, \beta_i^{n-1})      \nonumber \\                                       
                               &=& \beta^{n-k-1} _i \tilde{\bf v}_i, 
\end{eqnarray}
is  the $(k+1)$th right cyclic-shift of $\tilde{\bf v}_i$, for $0 \leq i < n-k$. It follows from (33) and (34) that we have 
\begin{equation}
 \tilde{\bf h} = \sum\limits^{n-k-1}_{i=0} \lambda_i \tilde{\bf v}_i  
 \end{equation}
where, for $0 \leq  i < n - k$,
\begin{equation}
              \lambda_i = \sigma_i \beta^{n-k-1} _i.                             
\end{equation}                
Then, it follows from (28), (35) and (36) that we have Theorem 3.

%Let ${\bf H}^{n-k-1}_{circ} = \Psi ( \tilde{\bf h}  ^{n-k-1})$ be the $n\times n$  circulant generated by  $\tilde{\bf h}  ^{n-k-1}$.  This circulant ${\bf H}^{n-k-1}_{circ} = \Psi ( \tilde{\bf h}  ^{n-k-1})$ is obtained by simultaneously cyclic-shifting all the rows of the circulant parity-check matrix  ${\bf H}_{circ}  = \Psi ( \tilde{\bf h}  )$ $n - k - 1$ times.  It follows (31), (32) and the above facts that we have Theorem 3 that  ${\bf H}_{circ}$  is a linear sum of the circulants, $\Psi ( \tilde{\bf v} (0)), \Psi ( \tilde{\bf v} (1)), \cdots , \Psi ( \tilde{\bf v} (n-k-1))$.

\begin{theorem}
 For an ($n$,$k$) cyclic code ${\cal C}_c$ over GF($q$) whose generator polynomial has elements $\beta_0, \beta_1, \cdots, \beta_{n-k-1}$ of GF($q^m$), as roots, then
 \begin{equation}
 \tilde{\bf h}(X) = \sum\limits^{n-k-1}_{i=0} \lambda_i \tilde{\bf v}_i (X) 
 \end{equation}  
where, for $0 \leq  i < n - k$,  
\begin{eqnarray}
                      \lambda _i & = &\sigma_i \beta ^{n-k-1}_i \nonumber\\
                               & =& \beta ^{n-k-1}_i \left(  \prod\limits_{j=0, j\neq i}^{n-k-1} (\beta _i - \beta _j) \right)^{-1} .    
\end{eqnarray}
The circulant parity-check matrix  ${\bf H}_{circ}$  of ${\cal C}_c$ given by (15) can be expressed as the following linear sum of circulants, $\Psi ( \tilde{\bf v} _0), \Psi ( \tilde{\bf v} _1), \cdots, \Psi ( \tilde{\bf v} _{n-k-1})$,
\begin{equation}
                 {\bf H}_{circ}  = \Psi ( \tilde{\bf h}  )  =  \sum\limits_{i=0}^{n-k-1}   \lambda _i \Psi ( \tilde{\bf v} _i).    
\end{equation}
where for $0 \leq i< n-k$, $\tilde{\bf v}_i=(1, \beta_i, \beta_i^2, \cdots, \beta_i ^{n-1})$.
 \end{theorem}
 The circulants, $\Psi (\tilde{{\bf v}}_{0}), \Psi (\tilde{{\bf v}}_1), \ldots, \Psi (\tilde{{\bf v}}_{n-k-1})$, are called the \emph{root circulants} of the cyclic code ${\cal C}_c$.  It follows from (9), (10), (23) and (37) that the coefficients of the parity-check polynomial ${\bf h}(X)$ are:
 \begin{eqnarray}    
   {h}_j &=&    \sum\limits^{n-k-1}_{i = 0} \lambda_i\beta^{j}_i,  \mbox{ for } 0\leq j\leq k,\\
   {h}_j &=&  0,  \mbox{ for } k < j < n.
 \end{eqnarray}

\subsection{Characterization of Cyclic Descendants of a Cyclic Code}

   In the following, we characterize the roots of the generator polynomial of a cyclic descendant of an ($n$,$k$) cyclic code ${\cal C}_c$ over GF($q$) whose parity-check matrix is given in terms of roots of the form given by (21). Consider the circulant $\Psi(\tilde{\bf v} _{i})$ with $\tilde{\bf v}_i = (1, \beta_i, \beta_i^2, \cdots, \beta_i^{n-1})$ as the generator.  Decompose $\Psi(\tilde{\bf v}_i)$ into a $c\times c$ array of $l\times l$ circulants.  The descendant circulants in the first row of $\Psi(\tilde{\bf v}_i)$ are $\Psi(\tilde{\bf v}_{i,0}), \Psi(\tilde{\bf v}_{i,1}), \cdots, \Psi(\tilde{\bf v}_{i,c-1})$ where for $0 \leq j < c$,
   \begin{eqnarray}
             \tilde{\bf v}_{i,0} &=& (1, \beta^c _i, \beta^{2c} _i, \cdots, \beta^{(l-1)c} _i),\\
              \tilde{\bf v}_{i,j} &=& \beta^j _i \tilde{\bf v}_{i,0}.                            
              \end{eqnarray}
              If follows from (42) and (43) that we have
              \begin{equation}
              \Psi(\tilde{\bf v}_{i,j}) = \Psi(\beta^j _i\tilde{\bf v}_{i,0}) = \beta^j _i \Psi(\tilde{\bf v}_{i,0}).           
              \end{equation}
The equality of (44) implies that if $\Psi(\tilde{\bf v}_{i,0})$ is known, all the descendant circulants, $\Psi(\tilde{\bf v}_{i,j})$'s and $\Psi(\tilde{\bf v}^{(1)}_{i,j})$'s can be constructed from $\Psi(\tilde{\bf v}_{i,0})$ using (44).
    
    It follows from Theorem 3 that the circulant generated by $\tilde{\bf h}_j$ is given as follows:
\begin{equation}
\Psi ( \tilde{\bf h}_j)  =  \sum\limits^{n-k-1}_{i=0}   \lambda_i \Psi ( \tilde{\bf v} _{i,j}) =   \sum\limits^{n-k-1}_{i=0}     \lambda_i\beta^j_i \Psi ( \tilde{\bf v} _{i,0}),    
\end{equation}
where  $\tilde{\bf h}  _j$, the $j$th cyclic section of  $\tilde{\bf h}$, is given by (17).  The null space of $\Psi ( \tilde{\bf h}_j)$ gives a cyclic code ${\cal C}_j^{(1)}$ over GF($q$) of length $l$, a type-1 descendant of ${\cal C}_c$. 
 
 For $0 \leq  i_1, i_2 < n - k$, suppose there exists an integer $t$ with $0 < t < c$ such that $\beta _{i_{2}} = \alpha^{tl}\beta_{i_1}$. In this case, since $\alpha^{cl} = \alpha^n = 1$, we must have $\beta _{i_{1}}^c = \beta _{i_2}^c$.  We say that $\beta _{i_1}$ and $\beta _{i_2}$ are \emph{equal in $c$th power}. Then, it follows from (42) and (44) that  $\tilde{\bf v} _{i_1,0} =  \tilde{\bf v} _{i_2,0}$ and $\Psi ( \tilde{\bf v} _{i_1,0}) = \Psi ( \tilde{\bf v} _{i_2,0})$. Let $m$ be the number of distinct circulants among $\Psi( \tilde{\bf v} _{0,0}), \Psi ( \tilde{\bf v} _{1,0}), \cdots , \Psi ( \tilde{\bf v} _{n-k-1,0})$.  Then, we can partition the $n - k - 1$ roots, $\beta _0, \beta_1, \cdots , \beta_{n - k - 1}$, into $m$ \emph{equal classes} in $c$th power.  For $0 \leq  e < m$, let
\begin{equation}
                       \Omega_e = \{\beta_{e,0}, \beta_{e,1}, \cdots, \beta_{e,r_e-1} \}
\end{equation}
be the $e$th class of equal roots in $c$th power where each $\beta _{e,f}$ in $\Omega_e$ is one of the roots, $\beta _0, \beta _1, \cdots, \beta_{n - k - 1}$, and $r_e$ is the number of equal roots in $\Omega_e$. It is clear that $1 \leq r_e \leq c$.  For $0 \leq f < r_e$, let
\begin{equation}
                       \tilde{\bf v}^*_{e,f} = (1, \beta^c _{e,f}, \beta^{2c}_{e,f}, \cdots, \beta^{(l-1)c}_{e,f}).   
                       \end{equation}
Since $\beta^c _{e,0} = \beta^c _{e,1} = \cdots = \beta^c _{e,r_e-1}$, we have $\tilde{\bf v}^*_{e,0} = \tilde{\bf v}^*_{e,1} =\cdots= \tilde{\bf v}^*_{e,r_e-1}$ and $\Psi(\tilde{\bf v}^* _{e,0}) = \Psi(\tilde{\bf v}^* _{e,1}) = \cdots = \Psi(\tilde{\bf v}^*_{e,r_e-1})$.  For $0 \leq f < r_e$, $\Psi(\tilde{\bf v}^*_{e,f})$ is one of the circulants $\Psi(\tilde{\bf v}_{0,0}), \Psi(\tilde{\bf v}_{1,0}), \cdots, \Psi(\tilde{\bf v}_{n-k-1,0})$ in the second sum of (45).  For $0 \leq e < m$, let 
\begin{equation}
                      L = \{\lambda_{e,0}, \lambda_{e,1}, \cdots, \lambda_{e,r_e-1} \}  
                      \end{equation}
be the set of coefficients, $\lambda_i$, of the circulants, $\Psi(\tilde{\bf v}^* _{e,0}), \Psi(\tilde{\bf v}^* _{e,1}), \cdots, \Psi(\tilde{\bf v}^* _{e,r_e-1})$, in the second sum of (45).  Grouping the identical circulants in the second sum of (45) together and for each $e$ with $0 \leq e < m$, using $\Psi(\tilde{\bf v}^* _{e,0})$ to represent the $e$th group of identical circulants, we have 
\begin{eqnarray}
                \Psi(\tilde{\bf h}_j)  &=&  \sum\limits^{n-k-1} _{i=0}  \lambda_i \beta^j _i \Psi(\tilde{\bf v}_{i,0}) \nonumber\\                
                            & =&  \sum\limits^{m-1}_{e=0}  \lambda^*_{e,j} \Psi(\tilde{\bf v}^*_{e,0}), 
                                \end{eqnarray}
where
\begin{equation}
              \lambda^* _{e,j} =      \sum\limits^{r_e -1}_{f=0}  \lambda_{e,f} \beta^j _{e,f}.                                   
                               \end{equation}
From (49), we see that the circulant $\Psi(\tilde{\bf h} _j)$ with generator $\tilde{\bf h}_j$ is a linear sum of the $m$ circulants, $\Psi(\tilde{\bf v}^* _{0,0}), \Psi(\tilde{\bf v}^* _{1,0}), \cdots, \Psi(\tilde{\bf v}^*_{m-1,0})$, where for $0 \leq e < m$, the circulant $\Psi(\tilde{\bf v}^* _{e,0})$ is generated by $\tilde{\bf v}^* _{e,0} = (1, \beta^c _{e,0}, \beta^{2c} _{e,0}, \cdots, \beta^{(l-1)c} _{e,0})$.  Then, it follows from (21), (28), (37), (38), (42), (49) and (50) that we have the following theorem.

\begin{theorem}
The generator polynomial ${\bf g}^{(1)}_j(X)$ of the type-1 cyclic descendant code ${\cal C}^{(1)}_j$ of the cyclic mother code ${\cal C}_c$ given by the null space of the $l\times l$ circulant $\Psi(\tilde{\bf h}_j)$ has $\beta^c _{e,0}$, $0\leq e<m$, as a root if and only if $\lambda^* _{e,j}  \neq 0$.
\end{theorem}

   Theorem 4 characterizes the roots of the generator polynomial of a type-1 cyclic descendant of a given cyclic mother code ${\cal C}_c$.

\begin{example}
 Let $\alpha$ be a primitive element of GF($2^{11}$).  Consider the binary primitive (2047,2025) BCH code whose generator polynomial ${\bf g}(X)$ has $\alpha, \alpha^2, \alpha^3, \alpha^4$ and their conjugates as roots. The length 2047 of the code can be factored as a product of $c = 89$ and $l = 23$. The $2047\times 2047$ circulant parity-check matrix ${\bf H}_{circ}$ of this BCH code can be decomposed into an $89\times 89$ array $\Phi (\tilde{\bf h})$ of circulants of size $23\times 23$ by column and row permutations $\pi$ defined by (3).  The null space of each $23\times 23$ descendant circulant of $\Phi (\tilde{\bf h})$ gives the (23,12) Golay code with generator polynomial $1 + X + X^5 + X^6 + X^7 + X^9 + X^{11}$ [6], which has $\beta=\alpha^{89}, \beta^2, \beta^3, \beta^4$ and their conjugates as roots.
\twotriangle  \end{example}

   The next theorem characterizes a type-2 cyclic descendant ${\cal C}_k ^{(2)}$ of ${\cal C}_c$ given by the null space of the parity-check matrix ${\bf H}_{col,k}$.

\begin{theorem}  For $1 \leq  k < c$, let $i_1, i_2, \ldots, i_k$ be a set of distinct integers such that $0 \leq  i_1, i_2, \ldots, i_k < c$.   For $1 \leq  t \leq  k$, let ${\bf g}_{i_t}^{(1)} (X)$ be the generator polynomial of  $i_t$-th type-1 cyclic descendant code ${\cal C}_{i_t}^{(1)}$  of ${\cal C}_c$ given by the null space of $i_t$-th descendant circulant $\Psi ({\bf w}_{i_t})$ of ${\bf H}_{circ}  = \Psi ({\bf w}) $.  Then the generator polynomial ${\bf g}_k ^{(2)}(X)$ of the type-2 cyclic descendant code ${\cal C}_k ^{(2)}$ of ${\cal C}_c$ given by the null space of  the parity-check matrix ${\bf H}_{col,k}$ of (8) is  the {least common multiple} of ${\bf g}_{i_1} ^{(1)} (X), {\bf g}_{i_2}^{(1)} (X), \ldots,  {\bf g}_{i_k}^{(1)} (X)$, i.e.,
\begin{equation}
                     {\bf g}_k^{(2)} (X) = LCM\{ {\bf g}_{i_t}^{(1)} (X),  0 \leq  t < k\}.      
                     \end{equation}        
                     The roots of ${\bf g}_k^{(2)} (X)$ is the union of the roots of ${\bf g}_{i_1} ^{(1)} (X), {\bf g}_{i_2}^{(1)} (X), \ldots,  {\bf g}_{i_k}^{(1)} (X)$.
\end{theorem}
   Consider the parity-check matrix ${\bf H}_{circ,mask}^{(3)}$ of a type-3 cyclic descendent ${\cal C}_{mask}^{(3)}$ of ${\cal C}_c$.  Express each row of ${\bf H}_{circ,mask}^{(3)}$ as a polynomial of degree $n-1$ or less with the leftmost component as the constant term and the rightmost component as the coefficient of the term of degree $n - 1$.  This polynomial is call a row polynomial.  Find the greatest common divisor $\tilde{\bf h}_{ mask}^{(3)}(X)$ of all the \emph{row polynomials}.  Let ${\bf h}_{mask}^{(3)}(X)$ be the reciprocal polynomial of $\tilde{\bf h}_{mask} ^{(3)}(X)$. Then the generator polynomial of ${\cal C}_{mask}^{(3)}$ is given by
 \begin{equation}
                      {\bf g}_{mask}^{(3)}(X) = (X^n - 1)/{\bf h}_{mask}^{(3)}(X).
\end{equation}

\section{Decomposition of Cyclic Euclidean Geometry LDPC Codes}

In this section, we give constructions of new cyclic and QC-LDPC codes by decomposing the circulant parity-check matrices of the cyclic Euclidean geometry (EG) LDPC codes.   
\subsection{Cyclic Descendants of Two-Dimensional EG-LDPC Codes}

   Consider a two-dimensional Euclidean geometry EG(2,$q$) over the field GF($q$), where $q$ is a power of a prime [6], [22], [26].  This geometry consists of $q^2$ points and  $q(q + 1)$ lines.  A point in EG(2,$q$) is simply a two-tuple ${\bf a} = (a_0, a_1)$ over GF($q$) and the zero two-tuple (0,0) is called the origin.  A line in EG(2,$q$) is simply a one-dimensional subspace, or its coset, of the vector space of all the $q^2$ two-tuples over GF($q$).   A line contains $q$ points. If a point $\bf a$ is on a line $\cal L$ in EG(2,$q$), we say the line $\cal L$ passes through $\bf a$.   Any two points in EG($2$,$q$) are connected by one and only one line. For every point $\bf a$ in EG(2,$q$), there are $(q + 1)$ lines that intersect at (or pass through) the point $\bf a$.  These lines are said to form an \emph{intersecting bundle} of lines at the point $\bf a$. For each line in EG(2,$q$), there are $q-1$ lines parallel to it.  Two parallel lines do not have any point in common. The $q(q + 1)$ lines in EG($2$,$q$) can be partitioned into $(q + 1)$ groups, each group consists of $q$ parallel lines. A group of $q$ parallel lines is called a \emph{parallel bundle}.

   The field GF($q^2$), as an extension field of the ground field GF($q$), is a realization of EG(2,$q$).  Let $\alpha$  be a primitive element of GF($q^2$). Then, the powers of $\alpha$, $\alpha ^{-\infty} \triangleq 0,\; \alpha ^0 = 1,\;  \alpha ,\; \alpha ^2,\; \cdots ,\; \alpha ^{q^2-2}$, give all the $q^2$ elements of GF($q^2$) and they represent the  $q^2$ points of  EG(2,$q$).  The 0-element represents the origin of EG(2,$q$).

    Let EG*(2,$q$) be the subgeometry obtained from EG(2,$q$) by removing the origin and the $q+1$ lines passing through the origin.  This subgeometry consists of $q^2 - 1$ non-origin points and $q^2 - 1$ lines not passing through the origin.  Each line in EG*($2$,$q$) has only $q - 2$ lines parallel to it.  Hence, each parallel bundle of lines in EG*($2$,$q$) consists of $q -1$ parallel lines not passing through the origin.  Each intersecting bundle of lines at a non-origin point consists of $q$ lines.  Let ${\cal L} = \{\alpha ^{j _1},\;  \alpha ^{j _2},\; \cdots,\; \alpha ^{j _{q}}\}$ be a line in EG*($m$,$q$).  For $0 \leq  i < q^2 -1$, let $\alpha^i {\cal L} = \{ \alpha ^{j_1+i},\;  \alpha ^{j_2 +i},\; \cdots,\; \alpha ^{j_q +i} \}$.
Then, $\alpha^i {\cal L}$ is also a line in EG*(2,$q$) and $\alpha ^0 {\cal L},\; \alpha {\cal L},\; \cdots,\; \alpha ^{q^2 -2} {\cal L}$ give all the $q^2 - 1$ lines in EG*(2,$q$).  This structure of lines is called \emph{cyclic structure} [6], [7].

    Let ${\cal L}$ be a line EG*(2,$q$).  Based on $\cal L$, we define the following $(q^2 -1)$-tuple over GF(2), 
\[
                   {\bf v}_{\cal L} = (v_0, v_1, \cdots ,  v_{q^2 -2}), 
\]
whose components correspond to the $q^2 - 1$ non-origin points $\alpha ^0, \alpha , \alpha ^2, \cdots, \alpha ^{q^2-2}$ of  EG*(2,$q$), where $v_j = 1$ if $\alpha ^j$ is a point on $\cal L$ and $v_j = 0$ otherwise.  It is clear that the weight of ${\bf v}_{\cal L}$ is $q$. This $(q^2-1)$-tuple ${\bf v}_{\cal L}$ is called the \emph{incidence vector} of the line $\cal L$ [5], [6].  Due to the cyclic structure of the lines in EG*(2,$q$) (i.e., if $\cal L$ is a line, $\alpha {\cal L}$ is also a line), the incidence vector ${\bf v}_{\alpha {\cal L}}$ of the line $\alpha {\cal L}$ is the cyclic-shift (one place to the right) of the incidence vector ${\bf v}_{\cal L}$ of the line $\cal L$.

    Let $n = q^2 - 1$. Form an $n\times n$ matrix ${\bf H}_{EG}$ over GF(2) with the incidence vectors of the $n$ lines, $\alpha ^0{\cal L}, \alpha {\cal L}, \cdots, \alpha ^{n-1} {\cal L}$, of EG*(2,$q$) as rows. Then, ${\bf H}_{EG}$  is an $n\times n$ circulant with both column and row weights $q$.  ${\bf H}_{EG}$ can be obtained by using the incidence vector ${\bf v}_{\cal L}$ of the line $\cal L$ as the generator and cyclically shifting ${\bf v}_{\cal L}$ $n - 1$ times.  Since two lines in EG*(2,$q$) have at most one \emph{point} in common, their incidence vectors have at most one position where they both have 1-components.  Consequently, ${\bf H}_{EG}$ satisfies the RC-constraint and its null space gives a cyclic EG-LDPC code ${\cal C}_{EG}$ [5], [6], [15] whose Tanner graph is free of cycles of length 4 and hence has a girth of at least 6.  The RC-constraint on the parity-check matrix ${\bf H}_{EG}$ ensures that the minimum weight (or distance) of  ${\cal C}_{EG}$ is at least $q + 1$.  To find the generator polynomial ${\bf g}_{EG}(X)$ of ${\cal C}_{EG}$, we express each row of ${\bf H}_{EG}$  as a polynomial over GF(2) of  degree $n - 1$ or less with leftmost entry as the constant term and rightmost entry as the coefficient of $X^{n-1}$. Let $\tilde{\bf h}(X)$ be greatest common divisor of the row polynomials of ${\bf H}_{EG}$.  The reciprocal ${\bf h}_{EG}(X)$ of $\tilde{\bf h}_{EG}(X)$ is the parity-check polynomial.  Then, the generator polynomial ${\bf g} _{EG}(X) = (X^n - 1)/{\bf h}_{EG}(X)$.

   For the special case with $q = 2^s$, the rank of ${\bf H}_{EG}$ is $3^s - 1$ [27], [6], [9] and the minimum weight of ${\cal C}_{EG}$ is exactly $2^s + 1$ [6], [28].  An integer $h$ with $0 \leq  h < 2^{2s}$, can be expressed in radix-$2^s$ form as follows: $h = c_0 + c_1 2^s$, where $0 \leq  c_0, c_1 < 2^s$.  The sum $W_{2^s} (h) = c_0  +c_1$ is called the $2^s$-weight of $h$.  For any non-negative integer $l$, let $h^{(l)}$ be the remainder resulting from dividing $2^{l}h$ by $2^{2s} - 1$.  Then $0 \leq   h^{(l)} < 2^{2s} - 1$. The radix-$2^s$ form and $2^s$-weight of $h^{(l)}$ are $h^{(l)} = c^{(l)} _0 + c^{(l)} _1 2^s$ and $W_{2^s}(h^{(l)}) = c^{(l)}_0 + c^{(l)} _1$, respectively.  Then, $\alpha^{h}$ is root of the generator polynomial ${\bf g}_{EG}(X)$ of ${\cal C}_{EG}$ if and only if [6], [28]
\begin{equation}
                           0 < \max\limits_{       0 \leq  l < s} W_{2^s}(h^{(l)}) < 2^s.            
\end{equation}                
The smallest integer that does not satisfy the condition given by (53) is $2^s +1$.  Hence, ${\bf g}_{EG}(X)$ has the following consecutive powers of $\alpha , \alpha ^2, \cdots, \alpha ^{2^s}$, as roots. 

   Constructions of cyclic LDPC codes based on finite geometries, Euclidean and projective, were first presented in [5].  In [5], the authors showed that cyclic finite geometry (FG) codes perform very well over the AWGN channel with iterative decoding based on belief propagation (IDBP) using the sum-product algorithm (SPA) and the decoding of these codes converges very fast. 

   Let $c$ and $l$ be two proper factors of $n$ such that $n = c\cdot l$.  Decompose the $n\times n$ circulant parity-check matrix ${\bf H}_{EG}$ into a $c\times c$ array $\pi({\bf H}_{EG}) = \pi( \Psi ({\bf v}_{\cal L}))$ of circulants of size $l\times l$ in the form of (4) through column and row permutation $\pi$ defined by (3).  Note that every row of $\pi ({\bf H}_{EG})$, as a $(q^2- 1)\times (q^2- 1)$ matrix, still corresponds to a line in EG*(2,$q$) not passing through the origin of EG(2,$q$).  Since ${\bf H}_{EG}$ satisfies the RC-constraint, each descendant circulant in $\pi({\bf H}_{EG})$ also satisfies the RC-constraint.
   
   Based on the array $\pi ({\bf H}_{EG} )$ of circulants, three types of cyclic descendant LDPC-codes of the cyclic EG-LDPC code ${\cal C}_{EG}$   can be constructed.  Note that the first row of ${\bf H}_{EG}$  is not the parity-check vector.  For $q = 2^s$, the roots of the generator polynomial ${\bf g}_{EG}(X)$ of ${\cal C}_{EG}$ can be determined from (53).  Then, it follows from Theorems 4 and 5, the roots of the generator polynomials of a type-1 and type-2 cyclic descendant codes can be determined.  QC-EG-LDPC codes can also be constructed by taking the null spaces of subarrays of $\pi ({\bf H}_{EG} )$.

For $q = 2^s$, let $2^s - 1 = c\cdot l$.  Let ${\bf v} = (v_0, v_1, . . . , v_{2^s -2})$ be the incidence vector of a chosen line in EG(2,$2^s$) not passing through the origin as the generator of the $(2^s - 1)\times (2^s - 1)$ circulant ${\bf H}_{EG} = \Psi ({\bf v})$ over GF(2).  For $0 \leq  i < c$, let ${\bf v}_i = (v_i, v_{c+i}, . . ., v_{(l-1)c+i})$ be a cyclic section of $\bf v$. The ranks of ${\bf H}_{EG}  = \Psi ({\bf v})$ and its type-1 circulant descendant $\Psi ({\bf v}_i)$ and type-3 circulant descendant ${\bf H}_{EG,mask} = \Psi ({\bf v})_{mask}$ (masked circulant of ${\bf H}_{EG} = \Psi ({\bf v}))$ as defined in Section II.B can be determined easily.  Let $\alpha$  be a primitive element of GF($2^s$). Define the following two $(2^s - 1)\times (2^s -1)$ matrices over GF($2^s$): ${\bf V} = [\alpha ^{-ij}]$ and ${\bf V}^{-1} = [\alpha ^{ij}]$, $0 \leq  i, j < 2^s - 1$.  Both $\bf V$ and ${\bf V}^{-1}$ are \emph{Vandermonde matrices} [23], [24] and non-singular.  Furthermore, ${\bf VV}^{-1} = {\bf I}$ where $\bf I$ is a $(2^s - 1)\times (2^s - 1)$ identity matrix. Hence, ${\bf V}^{-1}$ is the inverse of $\bf V$ and vice versa.  Then, the matrix
 \[
\begin{array}{ll}
                         {\bf H}_{EG}^{\mathcal F}& = {\bf V H }_{EG} {\bf V}^{-1} = {\bf V} \Psi ({\bf v}) {\bf V}^{-1}   \vspace{0.4cm}  \\                                                                  &  = diag (\sum\limits^{2^{s}-2}_{j=0} v_j, \sum\limits^{2^s-2}_{j=0}  \alpha ^j v_j, \ldots , \sum\limits^{2^s-2}_{j=0}  \alpha ^{(2^s-2)j} v_j)
\end{array}
\]              
is a $(2^s - 1)\times (2^s -1)$ diagonal matrix over GF($2^s$) whose $i$th diagonal element, $0 \leq  i < 2^s - 2$, equals $\sum\limits^{2^s-2}_{j=0}  \alpha ^{ij} v_j$.      The vector composed of the diagonal elements of ${\bf H}_{EG}^{\mathcal F}$ is the Fourier transform [23] of the incidence vector ${\bf v} = (v_0, v_1, . . . , v_{2^s -2})$.  ${\bf H}_{EG} ^{\mathcal F}$ is called the Fourier transform of ${\bf H}_{EG} $. $ {\bf H}_{EG} ^{\mathcal F}$ and ${\bf H}_{EG}$ have the same rank.  Since ${\bf H}_{EG} ^{\mathcal F}$ is a diagonal matrix, its rank, denoted by $rank({\bf H}_{EG} ^{\mathcal F})$, is equal to the number of nonzero diagonal elements in ${\bf H}_{EG} ^{\mathcal F}$ which is  $3^s - 1$, same as that of ${\bf H}_{EG}$.
   
Similarly, the rank of the type-3 circulant descendant ${\bf H}_{EG ,mask} = \Psi ({\bf v})_{mask}$ of ${\bf H}_{EG} = \Psi ({\bf v})$ is equal to the number of nonzero diagonal elements of its Fourier transform $({\bf H}_{EG,mask})^{\mathcal F}$ of ${\bf H}_{EG,mask}$.

   To determine the rank of a type-1 descendant circulant $\Psi ({\bf v}_i)$ of ${\bf H}_{EG} = \Psi ({\bf v})$.  We define ${\bf V} = [\beta ^{-ij}]$ and ${\bf V}^{-1} = [\beta ^{ij}]$, $0 \leq  i, j < l$ where $\beta  = \alpha ^c$. The order of $\beta$ is $l$.  Then, for $0 \leq  i < c$, the Fourier transform of $\Psi ({\bf v}_i)$ is
 \[
\begin{array}{ll}
                         (\Psi ({\bf v}_i))^{\mathcal F} &= {\bf V}\Psi ({\bf v}_i) {\bf V}^{-1}\\ 
                                 &= diag (\sum\limits^{l-1}_{j=0} v_{jc+i}, \sum\limits^{l-1}_{j=0}  \beta ^j v_{jc+i}, \ldots, \sum\limits^{l-1}_{j=0}  \beta ^{(l-1)j} v_{jc+i})
                                 \end{array}
   \]
is an $l\times l$ diagonal matrix over GF($2^s$) whose $i$th diagonal element, $0 \leq  i < l$, equals $\sum\limits^{l-1}_{j=0}  \beta ^{ij} v_{jc+i}$.      The vector  composed of the diagonal elements of $(\Psi ({\bf v}_i))^{\mathcal F}$ is the Fourier transform of ${\bf v}_i = (v_i, v_{c+i}, \ldots, v_{(l-1)c+i})$.  Then, for $0 \leq  i < c$, $(\Psi ({\bf v}_i))^{\mathcal F}$ is the Fourier transform of the type-1 descendant circulant $\Psi ({\bf v}_i)$ of ${\bf H}_{EG} = \Psi ({\bf v})$.  $(\Psi ({\bf v}_i))^{\mathcal F}$ and $\Psi ({\bf v}_i)$ have the same rank.  Hence the rank, $rank(\Psi ({\bf v}_i))$, is equal to the number of nonzero diagonal elements in $(\Psi ({\bf v}_i))^{\mathcal F}$.

  To determine the rank of the parity-check matrix ${\bf H}_{col,k}$ of a type-2 cyclic descendant code given by (8).  We first find the Fourier transform  of each $l\times l$ circulant descendant in ${\bf H}_{col,k}$.   Divide the rows of the Fourier transforms of the $k$ descendant circulants in ${\bf H}_{col,k}$ into $l$ groups.  Each group $\theta _j$, $1 \leq  j \leq  l$, consists of the $j$th rows of the $k$ descendant circulants in ${\bf H}_{col,k}$. A group is called a nonzero group if not all its $k$ rows are zero rows, otherwise called a zero group.  Then the rank of ${\bf H}_{col,k}$ is equal to the number of nonzero groups of rows in the Fourier transforms of the $k$ descendant circulants in ${\bf H}_{col,k}$.

\begin{example} 
Let the two-dimensional Euclidean geometry EG(2, $2^6$) over GF($2^6$) be the code construction geometry.  The field GF($2^{12}$) is a realization of EG(2, $2^6$). Based on the incidence vectors of the $2^{2\times 6} - 1 = 4095$ lines not passing the origin of EG(2, $2^6$), we can construct a $4095\times 4095$ RC-constrained circulant ${\bf H}_{EG}$ with both column and row weights 64. Any line not passing through the origin of EG(2,$2^6$) can be used to construct the generator (the first row) of ${\bf H}_{EG}$.  The rank of ${\bf H}_{EG}$ is $3^6 -1 =728$. The null space of ${\bf H}_{EG}$ gives a (4095,3367) cyclic EG-LDPC code ${\cal C}_{EG}$ with minimum distance 65.  Its error performances decoded with 50 iterations of the sum-product algorithm (SPA) [3], [6], [15] and the scaled min-sum (MS) algorithm [29] over the binary AWGN channel are shown in Figure 1. We see that the error performance of the code decoded with 50 iterations of SPA is slightly better than that of 50 iterations of the scaled MS-algorithm.  Furthermore, decoding of the code with the MS algorithm converges very fast.  The performance curves with 5, 10 and 50 iterations of the scaled MS-algorithm almost overlap with each other.  Also included in Figure 1 is the error performance of the code decoded with the soft-reliability based iterative majority-logic decoding (SRBI-MLGD) devised in [30]. We see that, at bit-error rate (BER) of $10^{-6}$, the SRBI-MLGD performs only 0.6 dB from the scaled MS with 50 iterations.  The SRBI-MLGD requires only integer and binary logical operations with a computational complexity much less than that of the SPA and the MS-algorithm.  It offers more effective trade-off between error-performance and decoding complexity compared to the other reliability-based iterative decoding, such as the weighted bit-flipping (WBF) algorithms [5], [6], [15], [31], [32].

\twotriangle\end{example}

 \begin{example} Consider the $4095\times 4095$ circulant ${\bf H}_{EG}$ constructed in Example 2. Suppose we factor 4095 as the product of $c = 3$ and $l = 1365$.  By column and row permutations, the $4095\times 4095$ circulant ${\bf H}_{EG}$ can be decomposed into a $3\times 3$ array $\pi({\bf H}_{EG})$ of descendant circulants of size $1365\times 1365$ in the form given in (4).    Let $\Psi_0$,  $\Psi_1$ and $\Psi_2$ denote the 3 descendant circulants of ${\bf H}_{EG}$ in the first row of $\pi({\bf H}_{EG})$.  Then
\[
                  \pi ({\bf H}_{EG})  = \left[\begin{array}{lll} 
                                                \Psi _0   &        \Psi _1 &       \Psi _2\\
      \Psi^{(1)} _2  &  \Psi _0   &        \Psi _1 \\
      \Psi^{(1)} _1  &  \Psi^{(1)} _2   &        \Psi _0 
\end{array}\right].
 \]
The descendant circulants $\Psi_0$ and $\Psi_2$ both have column and row weights 24. The descendant circulant $\Psi_1$ has both column and row weights 16. The rank of $\Psi_1$ is 600 (the number of nonzero diagonal elements of its Fourier transform $\Psi_1 ^{\mathcal F}$).  Consider the cyclic LDPC code ${\cal C}_{EG}^{(1)}$ given by the null space of $\Psi_1$. This code is a (1365,765) cyclic EG-LDPC code with rate 0.56 and minimum weight   at least 17, the column weight of $\Psi_1$ plus 1. The code is a type-1 cyclic descendant of the cyclic (4095,3367) EG-LDPC code given in Example 2.  Its generator polynomial has $\beta  = \alpha^3, \beta^2, . . . ,\beta^{16}$ consecutive power of $\beta$  as roots where $\alpha$  is a primitive element of GF($2^{12}$).  It follows from the BCH bound [6], that the minimum weight is again at least 17 which agrees with bound of column weight plus one. By extensive computer search, we find that ${\cal C}_{EG}^{(1)}$ has no trapping set with size smaller than 17 (see Section VI), however, we do find a (17,0) trapping set which gives a codeword of weight 17.  Therefore, the minimum weight of ${\cal C}_{EG}^{(1)}$ is exactly 17 and the error-floor of this code is dominated by the minimum weight of the code.  The error performance of the code over the AWGN channel using BPSK signaling decoded with 50 iterations of the SPA (or MSA) is shown in Figure 2(a).  At the block error rate (BLER) of $10^{-5}$, the code performs 1.6 dB from the sphere packing bound.

  Suppose we use
\[                                               
{\bf H}_{col,3} =  \left[\begin{array}{c}
 \Psi _0\\
 \Psi _1\\
 \Psi _2
 \end{array} \right]
 \]
as a parity-check matrix.  This matrix is a $4095\times 1365$ matrix over GF(2) with constant column weight 64 but two different row weights, 16 and 24.  Its rank is 664 and hence it has a large row redundancy (3431 redundant rows).  The null space of ${\bf H}_{col,3}$ gives a (1365,701) cyclic-EG-LDPC code ${\cal C}_{EG}^{(2)}$ with rate 0.5135 and minimum distance at least 65.  It is a type-2 cyclic descendant of the (4095,3367) cyclic EG-LDPC code given in Example 2.  The error performances of this code over the AWGN channel decoded with 50 iterations of the SPA and the SRBI-MLGD-algorithm are shown in Figure 2(b). This code is one-step majority-logic decodable and it can corrects 32 errors with simple one-step (OS) majority-logic decoding (MLGD) [6].

 Suppose we replace the circulants, $\Psi_2$  and its cyclic-shift $\Psi ^{(1)}_2$, in $\pi ({\bf H}_{EG})$ by two $1365\times 1365$ zero matrices ${\bf O}$.  We obtain the following $3\times 3$ masked array of circulants of size $1365\times 1365$:
\[
                  \pi ({\bf H}_{EG})_{mask}  =  \left[\begin{array}{lll} 
                                                \Psi _0   &        \Psi _1 &       {\bf O}\\
       {\bf O}&  \Psi _0   &        \Psi _1 \\
      \Psi^{(1)} _1  &   {\bf O}&        \Psi _0 
\end{array}\right].
 \]
The above array is still  in the form of (4) with doubly cyclic structure.  It is a $4095\times 4095$ matrix over GF(2) with both column and row weights 40.  Applying the inverse permutation $\pi ^{-1}$ to the rows and columns of $\pi ({\bf H}_{EG})_{mask}$, we obtain an RC-constrained $4095\times 4095$ circulant ${\bf H}_{EG,mask}$ with both column and row weights 40. The rank of ${\bf H}_{EG,  mask}$ is 1392. The null space of ${\bf H}_{EG,mask}$ gives a (40,40)-regular (4095,2703) cyclic-EG-LDPC code with minimum distance at least 41. It is a type-3 cyclic descendant code of the (4095,3367) cyclic EG-LDPC code given in Example 2.  The error performances of this code over the AWGN channel decoded with 3, 5 and 50 iterations of the SPA is shown in Figure 2(c).               
\twotriangle\end{example}  

We can factor 4095 as the product of 15 and 273.  Setting $c = 15$ and $l = 273$, we can decompose the $4095\times 4095$ circulant ${\bf H}_{EG}$  given in Example 2 into a $15\times 15$ array $\pi ({\bf H}_{EG} )$ of circulants of size $273\times 273$.  From this array of circulants, we can construct many type-1,-2 and -3 cyclic descendant LDPC codes of the (4095-3367) cyclic EG-LDPC code ${\cal C}_{EG}$  given by the null space of ${\bf H}_{EG} $.
 
   In this section, we have shown that given a two-dimensional Euclidean geometry, many cyclic EG-LDPC codes with large minimum weights can be constructed.

\subsection{Quasi-Cyclic Descendants of Two-Dimensional Cyclic EG-LDPC Codes}
 
   In the previous subsection, we have considered constructions of cyclic descendant LPDC codes of cyclic EG-LDPC codes based on two-dimensional Euclidean geometries.  In this subsection, we consider constructions of QC descendant LDPC codes of cyclic EG-LDPC codes based on two-dimensional Euclidean geometry.  As pointed out earlier that construction of QC descendant EG-LDPC codes based on two-dimensional Euclidean geometries was also proposed in [9].  However, the approach to construction proposed in this section is different, mathematically simpler and more general than that in [9].  The approach in conjunction with masking allows us to construct both high and low rate codes. Furthermore, a fundamental theorem on decomposition of a circulant parity-check matrix ${\bf H}_{EG}$ constructed based on a two dimensional Euclidean geometry into an array of \emph{circulant permutation matrices} (CPMs) is proved.  This theorem will be generalized for constructing QC-EG-LDPC codes based on high-dimensional Euclidean geometries. Therefore, the construction of QC descendant EG-LDPC codes is a generalization of that proposed in [9].

    In the following, we will present two types of QC descendant EG-LDPC codes.  First, we consider the RC-constrained $c\times c$ array $\pi ({\bf H}_{EG} )$ of circulants over GF(2) of size of $l\times l$ constructed in the previous subsection where $cl = n = q^2 - 1$ and $l > q - 1$.  For a pair of positive integers, ($s$,$t$) with $1 \leq  s, t \leq  c$,  let $\pi ({\bf H}_{EG} ) (s,t)$ be a $s\times t$ subarray of $\pi ({\bf H}_{EG} )$.  This subarray also satisfies the RC-constraint and its null space gives a QC descendant LDPC code ${\cal C}_{EG,qc}^{(1)}$ of the cyclic EG-LDPC code ${\cal C}_{EG}$  given by the null space of the $n\times n$ circulant ${\bf H}_{EG} $.  The QC-LDPC code ${\cal C}_{EG,qc}^{(1)}$ is referred to as a \emph{type-1 QC descendant code} of ${\cal C}_{EG}$.  Note that $q$ does not divide $q^2 - 1$.  For $l > q -1$ and be a factor of $q^2 - 1$, the smallest $l$  is $q + 1$.

   Notice that the transpose of the parity-check matrix ${\bf H}_{col,k}$ of a type-2 cyclic descendant EG-LDPC code gives the parity-check matrix $\pi({\bf H}_{EG})(1,k)$ of a type-1 QC descendant EG-LDPC code.  Both parity-check matrices ${\bf H}_{col,k}$ and $\pi({\bf H}_{EG})(1,k)$ have the same rank which is equal to the number of nonzero groups of rows in the Fourier transforms of the $k$ circulants in ${\bf H}_{col,k}$ (or the number of nonzero groups of columns in the Fourier transforms of the $k$ circulants in $\pi ({\bf H}_{EG})(1,k))$. 

\begin{example}
Consider the $3\times 3$ array $\pi ({\bf H}_{EG} )$ of circulants of size $1365\times 1365$ given in Example 3 constructed based on the two-dimensional Euclidean geometry EG(2,$2^6$).  Set $s = 1$ and $t = 3$.  Take the first row $[\Psi_0 \;  \Psi_1 \; \Psi_2]$ of $\pi ({\bf H}_{EG} )$ as a $1\times 3$ subarray $\pi ({\bf H}_{EG} )(1,3)$ of $\pi ({\bf H}_{EG} )$, i.e., $\pi ({\bf H}_{EG} )(1,3) = [\Psi _0 \; \Psi _1 \; \Psi _2]$ which is the transpose of the parity-check matrix ${\bf H}_{col,3}$ of the type-2 cyclic LDPC code given in Example 3.   $\pi ({\bf H}_{EG} )(1,3)$ is a $1365\times 4095$ matrix over GF(2) with constant row weight 64 but two different column weights 16 and 20.  The null space of this subarray gives a (4095,3431) QC-EG-LDPC code, a QC descendant of (4095,3367) cyclic EG-LDPC code given in Example 2.  The bit and block error performances with 3, 5, and 50 iterations of the SPA are shown in Figure 3.                         
\twotriangle\end{example}
 
   For a type-1 QC descendant of a cyclic EG-LDPC codes ${\cal C}_{EG}$  given by the null space of a $(q^2 - 1)\times (q^2 - 1)$ circulant ${\bf H}_{EG}$  constructed based on the 2-dimensional Euclidean geometry EG(2,$q$), the size of each circulant in its parity-check matrix is at least $q + 1$.  
   
   Next, we consider type-2 QC descendants of ${\cal C}_{EG}$ .   Suppose $q - 1$ can be factored as a product of two integers, $b$ and $l$ with $1\leq b, l < q$, i.e., $q - 1 = bl$.  Then $n = q^2 - 1$ can be factored as the following product: $n = (q  + 1)(q - 1) = (q + 1)bl$.  Let $c = (q + 1)b$.  Then, the circulant parity-check matrix ${\bf H}_{EG}$ of the cyclic EG-LDPC code ${\cal C}_{EG}$ of length $n = q^2 - 1$ can be decomposed into an RC-constrained $(q + 1)b \times (q + 1)b$ array $\pi ({\bf H}_{EG})_{cpm}$ of circulants over GF(2) of size $l\times l$.  Since $\pi ({\bf H}_{EG})_{cpm}$ is obtained from ${\bf H}_{EG}$ by column and row permutations, the rank of $\pi ({\bf H}_{EG})_{cpm}$ is the same as the rank of ${\bf H}_{EG}$. The following theorem gives a fundamental structure of the array $\pi ({\bf H}_{EG})_{cpm}$ which allows us to construct a large class of QC-LDPC codes which are QC descendants of the cyclic EG-LDPC code ${\cal C}_{EG}$. We will show that each circulant in $\pi ({\bf H}_{EG})_{cpm}$ is either a \emph{circulant permutation matrix} (CPM) or a zero matrix of size $l\times l$. (A CPM is a permutation matrix for which each row is the cyclic-shift of the row above it and the first row is the cyclic-shift of the last row.) We call the array $\pi ({\bf H}_{EG})_{cpm}$ the \emph{CPM-decomposition} of ${\bf H}_{EG}$, where the subscript ``CPM'' stands for ``CPM-decomposition''.
   
   \begin{theorem}
Let ${\bf H}_{EG}$ be the $(q^2 - 1 )\times (q^2 - 1)$ circulant over GF(2) constructed based on the $q^2 - 1$ lines of the two-dimensional Euclidean geometry EG(2,$q$) over GF($q$) not passing through the origin.  Suppose $q - 1$ can be factored as a product of two integers, $b$ and $l$ with $1\leq b, l < q$, i.e., $q - 1 = bl$.    Let $c = (q + 1)b$. Then, ${\bf H}_{EG}$ can be decomposed as a $(q + 1)b \times (q + 1)b$ array $\pi ({\bf H}_{EG})_{cpm}$ of circulants of size $l\times l$.  Each circulant is either an $l\times l$ CPM or an $l\times l$ zero matrix (ZM).  Each row (or column) block of $\pi ({\bf H}_{EG})_{cpm}$ consists of exactly $q$ CPMs and $(q + 1)b - q$  ZMs.
\end{theorem}
 \begin{proof}  It follows from the definition of the incidence vector of a line in EG*(2,$q$) that the $q^2 - 1$ columns of  ${\bf H}_{EG}$ correspond to the $q^2 - 1$ non-origin points, $\alpha ^0 = 1,  \alpha , \alpha^2, . . . , \alpha ^{q-2}$, of EG*(2,$q$).  Permuting the columns and rows based on the permutation $\pi$  defined by (2) and (3), we decompose the circulant ${\bf H}_{EG}$ into a $c\times c$ array $\pi ({\bf H}_{EG})_{cpm}$ of circulants of size $l\times l$ in the form of (4).  For $0 \leq j < c$, consider the $j$th circulant $\Psi_j$ in the first row block of the array $\pi ({\bf H}_{EG})_{cpm}$.  It follows from the column permutation $\pi$  that the columns of $\Psi_j$ correspond to the non-origin points, $\alpha ^j, \alpha ^{c+j}, \alpha ^{2c+j}, . . . , \alpha ^{(l-1)c+j}$.  Suppose that $\Psi_j$ is neither an $l\times l$ CPM nor an $l\times l$ ZM. Then, the first row of $\Psi_j$ must have at least two 1-components.  Let ${\bf y}_1 = \alpha ^{l_1 c+j}$ and ${\bf y}_2 =  \alpha ^{l_2 c+j}$ with $0\leq l_1 < l_2 < l$, be the points that correspond to two positions where the first row of  $\Psi_j$ have 1-components.  Then,
\begin{equation} 
                                    {\bf y}_2 = \lambda {\bf y}_1,     
 \end{equation}
where $\lambda  = \alpha ^{(l_2-l_1)c}$ which is a nonzero element in GF($q$).  Since $0 < l_2  - l_1 < l$, $\lambda  \neq  1$. Let ${\bf y} = \eta {\bf x} + {\bf z}$ be the line in EG*(2,$q$) that contains the points (or connects) ${\bf y}_1$ and ${\bf y}_2$  where ${\bf x}$ and ${\bf z}$ are two linearly independent points in EG*(2,$q$) and $\eta \in \mbox{GF}(q)$.  Then,
\begin{equation}
\begin{array}{l} 
                                 {\bf y}_1 = \eta_1 {\bf x} + {\bf z},\\
                                 {\bf y}_2 = \eta_2 {\bf x} + {\bf z}.
 \end{array}
 \end{equation}
It follows from (54) and (55) that we have
 \begin{equation}
                              {\bf y}_2 = \lambda \eta_1 {\bf x} + \lambda {\bf z}. 
 \end{equation}
where $\lambda \eta _1$ is a nonzero element in GF($q$).  Since $\lambda  \neq  1$, the point $\lambda {\bf z}$ is different from the point ${\bf z}$.  Equality (56) implies that ${\bf y}_2$  is also a point on the line ${\bf y}' = \eta {\bf x} + \lambda {\bf z}$ that is parallel to the line ${\bf y} = \eta {\bf x} + {\bf z}$.  However, a point cannot be on two parallel lines. Consequently, the first row of $\Psi_j$ cannot have more than one 1-component and $\Psi _j$ is either a CPM or a zero matrix. 
 
   As a $(q^2- 1)\times (q^2 - 1)$ matrix over GF(2), the first row of ${\bf H}_{EG}$ (the incidence vector of a line in EG*(2,$q$)) has $q$ one-components.  Since $\pi ({\bf H}_{EG})_{cpm}$ is obtained from ${\bf H}_{EG}$ through column and row permutations, the first row of $\pi ({\bf H}_{EG})_{cpm}$, as a $(q^2- 1)\times (q^2 - 1)$ matrix over GF(2), also has $q$ one-components. Based on the result proved above, these $q$ one-components must distribute in $q$ CPMs in the first row block of the array $\pi ({\bf H}_{EG})_{cpm}$, one in each.  Consequently, the first row block of the array $\pi ({\bf H}_{EG})_{cpm}$ consists of $q$ CPMs and $c - q = (q + 1)b - q$ ZMs of size $l\times l$.  Since $\pi ({\bf H}_{EG})_{cpm}$ has the cyclic structure as displayed in (4), every row block of the array is the cyclic-shift of the row block above it and the first row block is the cyclic-shift of the last row block.  This cyclic structure implies that every row (or column) block of $\pi ({\bf H}_{EG})_{cpm}$ has $q$ CPMs and $(q + 1)b - q$  ZMs.  This proves the theorem.   \end{proof}
  
    The array $\pi ({\bf H}_{EG})_{cpm}$ of CPMs and ZMs of size $l\times l$ can be used as the base to construct QC-LDPC codes.  For any pair of integers, ($\gamma$,$\rho $) with $1\leq \gamma, \rho  \leq (q + 1)b$, let $\pi ({\bf H}_{EG})(\gamma, \rho)_{cpm}$ be a $\gamma \times \rho$  subarray of $\pi ({\bf H}_{EG})_{cpm}$.  It is an RC-constrained $\gamma l \times \rho l$ matrix over GF(2). Then, the null space of $\pi ({\bf H}_{EG})(\gamma , \rho )_{cpm}$ gives a QC-EG-LDPC code ${\cal C}_{EG,qc} (\gamma , \rho )$ of length $\rho l$ whose Tanner graph has a girth of at least 6.  If $\pi({\bf H}_{EG}) (\gamma,\rho)_{cpm}$ has constant column and row weights, then ${\cal C}_{EG,qc} (\gamma,\rho)$ is a regular QC-EG-LDPC code. Otherwise, $\pi({\bf H}_{EG}) (\gamma,\rho)_{cpm}$ has multiple column and/or row weights. In this case, the null space of $\pi({\bf H}_{EG}) (\gamma,\rho)_{cpm}$ gives an irregular QC-EG-LDPC code. 

    Here we consider a very special subclass of type-2 QC descendant LDPC codes of the two-dimensional cyclic EG-LDPC code ${\cal C}_{EG}$. The entire array $\pi({\bf H}_{EG})_{cpm}$ is a $(q^2 - 1)\times (q^2- 1)$ matrix over GF(2) with both column and row weights equal to $q$.  The null space of $\pi({\bf H}_{EG})_{cpm}$ gives a QC-EG-LDPC code ${\cal C}_{EG,qc}((q+1)b,(q+1)b)$ of length $n = q^2 - 1$ with minimum distance $q + 1$.  If $q = 2^s$, then the rank of $\pi({\bf H}_{EG})_{cpm}$ is $3^s - 1$ (the rank of $\pi({\bf H}_{EG})_{cpm}$ is the same as that of ${\bf H}_{EG})$.  In this case, the null space of $\pi({\bf H}_{EG})_{cpm}$ gives a QC-EG-LDPC code with the following parameters:
    
    \begin{center}
    
 \begin{minipage}{5cm}
             Length: $n  = 4^s - 1$,
 
             Dimension $= 4^s - 3^s$,
 
             Minimum distance $= 2^s + 1$.
 \end{minipage}
 
 \end{center}
 
   For a given two-dimensional Euclidean geometry EG(2,$q$) over GF($q$), the above construction gives a family of structurally compatible QC-EG-LDPC codes.   
   
   Each factor $l$ of $q - 1$ results in a CPM-decomposition of the circulant ${\bf H}_{EG}$ with CPMs of size $l\times l$. A special case of CPM-decomposition of ${\bf H}_{EG}$ is $l = q - 1$.  In this case, the CPM-decomposition of ${\bf H}_{EG}$ is a $(q+1)\times (q+1)$ array $\pi ({\bf H}_{EG})_{cpm}$ of CPMs and ZMs of size $(q - 1)\times (q - 1)$.  Each row (or column) block of $\pi ({\bf H}_{EG})_{cpm}$ consists of $q$ CPMs and one single ZM. There are a total of $q + 1$ ZMs in $\pi ({\bf H}_{EG})_{cpm}$.  In constructing the circulant ${\bf H}_{EG}$, we can choose a line $\cal L$ such that, after decomposition, the $q + 1$ ZMs in $\pi ({\bf H}_{EG})_{cpm}$ lie on its main diagonal.  This special case with $l = q - 1$ was first presented in [8] as an array of permutation matrices (PMs) of size $(q - 1)\times (q -1)$ and was later formulated as an array of CPMs of size $(q - 1)\times (q -1)$ in [9].

\begin{example} Consider the $4095\times 4095$ circulant ${\bf H}_{EG}$ over GF(2) constructed based the two-dimensional Euclidean geometry EG(2,$2^6$) given in Example 2. Factor $2^{2\times 6} - 1 = 4095$ as the product of $q+1=2^6 + 1 = 65$ and $q-1=2^6 - 1= 63$. Let $c = 65$ and $l = 63$.  Decompose the $4095\times 4095$ circulant ${\bf H}_{EG}$ into a $65\times 65$ array $\pi ({\bf H}_{EG}) _{cpm}$ of CPMs and ZMs of size $63\times 63$.  Suppose ${\bf H}_{EG}$ is constructed by choosing a line $\cal L$ not passing through the origin of EG(2,$2^6$) such that, after decomposition of  ${\bf H}_{EG}$, the 65 ZMs of $\pi ({\bf H}_{EG}) _{cpm}$ lie on its main diagonal.  The null space of $\pi ({\bf H}_{EG}) _{cpm}$ gives (4095,3367) QC-EG-LDPC code which is combinatorially equivalent to the (4095,3367) cyclic EG-LDPC code given in Example 2.  Suppose we choose a $6\times 65$ subarray $\pi ({\bf H}_{EG}) (6,65)_{cpm}$ of of $\pi ({\bf H}_{EG}) _{cpm}$. The null space of this subarray gives a (4095,3771) code with rate 0.921. The error performance of this code with 50 iterations of the SPA is shown in Figure 4.  At the BLER of $10^{-4}$, the (4095,3771) code performs 0.75  dB from the sphere packing bound.           
\twotriangle\end{example}

\begin{example}
Continue Example 5.  Suppose we factor $q - 1 = 63$ as the product of 9 and 7.  Set $b = 9$, $l = 7$ and $c = (q + 1)b = 65\times 9 = 585$.  Decompose the $4095\times 4095$ circulant ${\bf H}_{EG}$ given in Example 2 into a $585\times 585$ array $\pi ({\bf H}_{EG})_{cpm}$ of CPMs and ZMs of size $7\times 7$.  Choose $\gamma  = 72$ and $\rho  = 585$.  Take a $72\times 585$ subarray $\pi ({\bf H}_{EG})(72,585)_{cpm}$ from $\pi ({\bf H}_{EG})_{cpm}$.  The subarray $\pi ({\bf H}_{EG})(72,585)_{cpm}$ is a $504\times 4095$ matrix over GF(2).  The null space of this matrix gives a (4095,3591) QC-EG-LDPC code with rate 0.877 whose error performance over the AWGN decoded with 50 iterations of the SPA is shown in Figure 5. 
      \twotriangle\end{example}

\begin{example}
In this example, we construct a long high-rate code and show how close the code performs to the Shannon limit.  Let the two-dimensional Euclidean geometry EG(2,257) over the prime field GF(257) be the code construction geometry.  Based on the incidence vectors of the lines in EG(2,257) not passing through the origin of the geometry, we construct a $66048\times 66048$ circulant ${\bf H}_{EG}$ with both column and row weights 257.  The null space of ${\bf H}_{EG}$ gives a cyclic-EG-LDPC code of length of 66048 with minimum distance at least 258.

   Set $c = q + 1 = 257 + 1 =258$ and $l = q - 1 = 257 - 1 = 256$. Decompose ${\bf H}_{EG}$ into a $258\times 258$ array $\pi ({\bf H}_{EG})_{cpm}$ of CPMs and ZMs of size $256\times 256$.  In this CPM-decomposition, every row and every column consists of 257 CPMs and a single ZM. Suppose ${\bf H}_{EG}$ is constructed by choosing a line not passing through the origin of EG(2,$2^8$) such that the 258 ZMs lie on the main diagonal of the array $\pi ({\bf H}_{EG})_{cpm}$.

   Let $\gamma  = 4$ and $\rho = 128$. Take a $4\times 128$ subarray $\pi ({\bf H}_{EG})(4,128)_{cpm}$ from $\pi ({\bf H}_{EG})_{cpm}$, avoiding the ZMs on the main diagonal of $\pi ({\bf H}_{EG})_{cpm}$.  This subarray $\pi ({\bf H}_{EG})(4,128)_{cpm}$ is a $1024\times 32768$ matrix with column and row weights 4 and 128, respectively.  The null space of $\pi ({\bf H}_{EG})(4,128)_{cpm}$ gives a (4,128)-regular (32768,31747) QC-EG-LDPC code with rate 0.969.  The error performance of this code over the AWGN channel decoded with 50 iterations of the SPA is shown in Figure 6.  At the BER of $10^{-6}$, the code performs 0.6 dB from the Shannon limit.  
    \twotriangle\end{example}

   If we select a set of CPMs and their cyclic-shifts in $\pi ({\bf H}_{EG})_{cpm}$ and replace them by zero matrices of size $l\times l$, we obtain an array $\pi ({\bf H}_{EG,mask})_{cpm}$ of CPMs and ZMs which has the form of (4) with doubly cyclic structure.  Applying inverse permutation $\pi ^{-1}$ to the rows and columns of $\pi ({\bf H}_{EG,mask})_{cpm}$, we obtain a $(q^2- 1)\times (q^2- 1)$ masked circulant ${\bf H}_{EG,mask}$ over GF(2). The null space ${\bf H}_{EG,mask}$ gives a cyclic-EG-LDPC code of length $q^2 - 1$.
   
   \subsection{Masking}
    For a pair of two positive integers, $(\gamma , \rho )$ with $1 \leq  \gamma , \rho  \leq  q + 1$, let
 \begin{equation}
           \pi ({\bf H}_{EG})(\gamma ,\rho )_{cpm}  =  \left[\begin{array}{cccc}
                                                {\bf B}_{0,0}     &    {\bf B}_{0,1}   &   \cdots &   {\bf B}_{0,\rho -1}\\
                                               {\bf B}_{1,0}     &    {\bf B}_{1,1}   &    \cdots &      {\bf B}_{1,\rho -1}\\
                                                \vdots & &\ddots &\vdots\\
                                            {\bf B}_{\gamma -1,0} &     {\bf B}_{\gamma -1,1} &    \cdots   &{\bf B}_{\gamma -1,\rho-1}                                               \end{array}\right]. 
                                               \end{equation}                                               
be a $\gamma \times \rho$  subarray of $\pi ({\bf H}_{EG})_{cpm}$. A set of CPMs in $\pi ({\bf H}_{EG})(\gamma ,\rho )_{cpm}$  can be replaced by a set of ZMs. This replacement is referred to as masking [6], [8], [10], [11], [15].  Masking results in a sparser matrix whose associated Tanner graph has fewer edges and hence fewer short cycles and probably a larger girth than that of the associated Tanner graph of the original $\gamma \times \rho$  subarray $\pi ({\bf H}_{EG})(\gamma ,\rho )_{cpm}$.  To carry out masking, we first design a low density $\gamma \times \rho$  matrix ${\bf Z}(\gamma ,\rho ) = [z_{i,j}]$ over GF(2).  Then take the following matrix product: $\pi ({\bf M}_{EG})(\gamma ,\rho )_{cpm} = {\bf Z}(\gamma ,\rho )\otimes \pi ({\bf H}_{EG})(\gamma ,\rho )_{cpm} = [z_{i,j}{\bf B}_{i,j}]$, where $z_{i,j}{\bf B}_{i,j} = {\bf B}_{i,j}$ for $z_{i,j} = 1$ and $z_{i,j}{\bf B}_{i,j} = {\bf O} (a (q - 1)\times (q - 1)$ zero matrix) for $z_{i,j} = 0$.  We call ${\bf Z}(\gamma ,\rho )$ the masking matrix, $\pi ({\bf H}_{EG})(\gamma ,\rho )_{cpm}$ the base array  and $\pi ({\bf M}_{EG})(\gamma ,\rho )_{cpm}$ the masked array.  Since the base array $\pi ({\bf H}_{EG})(\gamma ,\rho )_{cpm}$ satisfies the RC-constraint, the masked array $\pi ({\bf M}_{EG})(\gamma ,\rho )_{cpm}$ also satisfies the RC-constraint, regardless of the masking matrix.  Hence, the associated Tanner graph of the masked matrix $\pi ({\bf M}_{EG})(\gamma ,\rho )_{cpm}$ has a girth at least 6.  The null space of the masked array $\pi ({\bf M}_{EG})(\gamma ,\rho )_{cpm}$ gives a new  QC-EG-LDPC code.  If both the masking matrix and the base array are regular, the masked array is also regular and its null space gives a regular QC-LDPC code.  However, if the masking matrix is irregular and base array is regular, the masked array is irregular and its null space gives an irregular code.  A well designed masking matrix results in a good LDPC code.  Design and construction of masking matrices for constructing binary LDPC codes are discussed in [6], [8], [10], [11].

\begin{example}
 In this example, we construct a long irregular QC-EG-LDPC code using the masking technique presented above.  Consider the $258\times 258$ array $\pi ({\bf H}_{EG})_{cpm}$ of CPMs and ZMs of size $256\times 256$ constructed in Example 7.  Take a $128\times 256$ subarray $\pi ({\bf H}_{EG})(128,256)_{cpm}$ from $\pi ({\bf H}_{EG})_{cpm}$. We use this subarray as a base array for masking to construct an irregular code of rate 1/2. Next we construct a $128\times 256$ masking matrix ${\bf Z}(128,286)$ (by computer search) with column and row weight distributions  close to the following variable-node and check-node degree distributions (node perspective) of a Tanner graph optimally designed for an irregular code of rate 1/2 and infinite length (using density evolution [33]):
\[\begin{array}{cc}         
    \lambda(X) = 0.4410X + 0.3603X^2 + 0.00171X^5 + 0.03543X^6 + 0.09331X^7 + 0.0204 X^8 
    \\+ 0.0048X^9+ 0.000353 X^{27} + 0.04292X^{29},
        \end{array}                 \]
and
\[
             \rho (X) = 0.00842X^7 + 0.99023X^8 + 0.00135X^9.
\]
where the coefficient of $X^i$ represents the percentage of nodes with degree $i +1$.  The column and row weight distributions of the constructed masking matrix ${\bf Z}(128,256)$ are given below:
\[
            v(X) = 106X + 105X^2 + 35X^8 + 10X^{29},
            \]
            \[
            c(X) = 10X^7 + 118X^8,
\]
where the coefficient $X^i$ gives the number of columns (or rows) of ${\bf Z}(128,256)$ with weight $i + 1$.

   Masking the $128\times 256$ subarray $\pi ({\bf H}_{EG})(128,256)_{cpm}$ with ${\bf Z}(128,256)$, we obtain a $128\times 256$ masked array $\pi ({\bf M}_{EG})(128,256)_{cpm} = {\bf Z}(128,256)\otimes \pi ({\bf H}_{EG})(128,256)_{cpm}$ of $256\times 256$ CPMs and ZMs.  It is a $32768\times 65536$ matrix over GF(2) with average column and row weights 3.875 and 7.75, respectively.  The null space of $\pi ({\bf M}_{EG})(128,256)_{cpm}$ gives an irregular (65536,32768) QC-EG-LDPC code.  The error performance of this code with 50 iterations of the SPA is shown in Figure 7.  We see that at a BER of $10^{-9}$, the code performs 0.6  dB from the Shannon limit without visible error floor.  Also include in Figure 7 is the performance of the DVB S-2 standard (64800,32400) LDPC code [34] with a BCH outer code. The DVB S-2 LDPC code is an IRA (irregular repeat-accumulated) code [15], [35]. The BCH code is a (32400,32208) shortened BCH code with error-correction capability 12.  The BCH outer code is  used to push down the error-floor of the DVB S-2 code.  We see that the (65536,32768) QC-EG-LDPC code outperforms DVB S-2 code with the BCH outer code. 
    \twotriangle\end{example}

\section{Construction of QC-LDPC Codes Based on Decomposition of Multiple Circulants Constructed from High-Dimensional Euclidean Geometries}
 
   In the last subsection, we considered decomposition of the single RC-constrained circulant constructed based on the lines of a two-dimensional Euclidean geometry EG(2,$q$) over a finite field GF($q$) not passing through the origin of the geometry into a $(q+1)b\times (q+1)b$ array of CPMs and ZMs of size $l\times l$ where $b$ and $l$ are factors of $q - 1$ and $bl = q - 1$.  From this array of CPMs and ZMs, we can construct a family of RC-constrained QC-EG-LDPC codes of various lengths and rates and a family of cyclic LDPC codes.  
   
   In this section, we consider decomposition of multiple circulants constructed based on lines of an $m$-dimensional Euclidean geometry EG($m$,$q$) over the Galois field GF($q$) into arrays of CPMs and ZMs of size $l\times l$.  From these arrays, we can construct a very large array of CPMs and ZMs which forms a base array to construct a large family of RC-constrained QC-EG-LDPC codes.

    Consider the $m$-dimensional Euclidean geometry EG($m$,$q$) over GF($q$).  This geometry consists of $q^m$ points and
 $                                J = q^{m-1} (q^m - 1)/(q - 1)                $
lines.  Each line consists of $q$ points.  The field GF($q^m$) as an extension field of the ground field GF($q$) is a realization of the geometry EG($m$,$q$) [6], [26].  Let $\alpha$  be a primitive element of GF($q^m$).  Then, the powers, $\alpha ^{-\infty}   \triangleq 0, \alpha ^0= 0, \alpha , . . . , \alpha ^{q^m-2}$, represent $q^m$ points of EG($m$,$q$).  Again, the element $\alpha ^{-\infty}  = 0$ represents the origin of EG($m$,$q$).  Let EG*($m$,$q$) be the sub-geometry obtained by removing the origin and the line passing through the origin from EG($m$,$q$).  This sub-geometry consists of $q^m - 1$ non-origin points and
$                             J_0 = (q^{m-1} - 1)(q^m - 1)/(q - 1)$
lines not passing through the origin of EG($m$,$q$).
 
   Let ${\cal L} = \{\alpha ^{j_1}, \alpha ^{j_2}, . . . , \alpha ^{j_q}\}$ with $0 \leq  j_1, j_2, . . . ,  j_q < q^m - 1$ be a line in EG*($m$,$q$) consisting of the points, $\alpha ^{j_1}, \alpha ^{j_2}, . . . ,\alpha ^{j_q}$.  For $0 \leq  t < q^m - 1$, $\alpha ^t {\cal L} = \{\alpha ^{j_1+t}, \alpha ^{j_2+t}, . . . , \alpha ^{j_q+t} \}$
is also a line in EG*($m$,$q$) [6], [7], [15]. The lines ${\cal L}, \alpha {\cal L}, \alpha ^2 {\cal L}, . . . , \alpha ^{q^m-2}{\cal L}$ are $(q^m - 1)$  different lines in EG*($m$,$q$).  Since $\alpha ^{q^m-1} = 1, \alpha ^{q^m-1}{\cal L} = {\cal L}$.  The $q^m - 1$ lines, ${\cal L}, \alpha {\cal L}, \alpha ^2{\cal L}, . . . , \alpha ^{q^m-2} {\cal L}$, are said to form a \emph{cyclic class}, denoted by $Q_{\cal L}$.  The $J_0$ lines in EG*($m$,$q$) can be partitioned into $K_0 = (q^{m-1} - 1)/(q - 1)$ cyclic classes.
 
   For any line $\cal L$ in EG*($m$,$q$) not passing through the origin, the incidence vector of $\cal L$ is a ($q^m - 1$)-tuple over GF(2) defined as follows: $                      {\bf v}_{\cal L} = (v_0, v_1, . . . , v_{q^m-2}),$
whose components correspond to the $q^m - 1$ non-origin points, $\alpha ^0 = 0, \alpha , . . . , \alpha ^{q^m-2}$, of EG*($m$,$q$) , where $v_j = 1$ if $\alpha ^j$ is a point on $\cal L$, otherwise $v_j = 0$.  The weight of the  incidence vector of a line is $q$.  Due to the cyclic structure of the lines in EG*($m$,$q$), the incidence vector ${\bf v}(\alpha ^{i+1}{\cal L})$ of the line $\alpha ^{i+1}{\cal L}$ is right cyclic-shift of the incidence vector ${\bf v}(\alpha ^i {\cal L})$ for $0 \leq  i < q^m - 1$. 
 
   Denote the $K_0$ cyclic classes of lines in EG*($m$,$q$) with $Q_{{\cal L}_0}, Q_{{\cal L}_1}, . . . , Q_{{\cal L}_{K_0-1}} $. For each cyclic class $Q_{{\cal L}_i}$ of $q^m - 1$ lines with $0 \leq  i < K_0$, we form a $(q^m - 1)\times (q^m - 1)$ circulant ${\bf H}_{EG,i}$ with the incidence vectors of the lines ${\cal L}_i, \alpha {\cal L}_i, \alpha ^2 {\cal L}_i, . . . , \alpha ^{q^m-2} {\cal L}_i$ as columns such that each column is downward cyclic-shift of the column on its left and the first column is the downward cyclic-shift of the last column.  This $(q^m - 1)\times (q^m - 1)$ ciculant ${\bf H}_{EG,i}$ satisfies the RC-constraint and has both column and row weights equal to $q$.  Let $q = p^s$ where $p$ is a prime.  For $s \geq  3$ and $m \geq  3$, $q$ is very small compared to $q^m - 1$.  Therefore, ${\bf H}_{EG,i}$ is a very sparse circulant.
  
    Form the following $(q^m - 1)\times K_0(q^m - 1)$ matrix over GF(2) with circulants, ${\bf H}_{EG ,1}, {\bf H}_{EG ,2}, \ldots, {\bf H}_{EG,K_0}$ as submatrices:
 \begin{equation}
                    {\bf H}_{EG,qc} = [{\bf H}_{EG ,0}\;  {\bf H}_{EG,1}, \ldots,  {\bf H}_{EG ,K_0 -1}] .
\end{equation} 
This matrix has column and row weights $q$ and $q K_0$, respectively.  Since the columns of ${\bf H}_{EG,qc}$ correspond to the lines of $EG^*(m,q)$, ${\bf H}_{EG ,qc}$ satisfies the RC-constraint.  Its null space gives an RC-constrained QC-EG-LDPC code ${\cal C}_{qc,m}$ of length $K_0(q^m - 1)$ with minimum distance at least $q + 1$.  The subscript "$m$" stands for the dimension of the Euclidean geometry $EG(m,q)$ used for code construction.
  
     Suppose $q - 1$ can be factored as a product of $b$ and $l$ with $0\leq  b, l < q$, i.e., $q - 1 = bl$.  Then $q^m - 1$ can be factored as follows:
 \begin{eqnarray*}
           q^m - 1 &= (q^{m-1} + q^{m-2} + . . . + q + 1)(q - 1)  \nonumber \\ 
                    & = (q^{m-1} + q^{m-2} + . . . + q + 1)bl. 
                    \end{eqnarray*} 
Let
 \begin{equation}
          c = (q^{m-1} + q^{m-2} + . . . + q + 1)b.
 \end{equation}

\begin{theorem}
For $0 \leq  i < K_0 $, each $(q^m - 1)\times (q^m - 1)$ circulant ${\bf H}_{EG,i}$ constructed based on $i$th cyclic class $Q_{L_i}$ of lines of the sub-geometry EG*($m$,$q$) can be decomposed into a $c\times c$ array $\pi ({\bf H}_{EG,i})_{cpm}$ of CPMs and MZs of size $l\times l$ by applying the $\pi $-permutation to both the columns and rows of ${\bf H}_{EG,i}$.  Each row (column) block of $\pi ({\bf H}_{EG,i})_{cpm}$ consists of $q$ CPMs and $c - q$ ZMs.
\end{theorem} 
\begin{proof}
The proof of this theorem is similar to the proof of Theorem 5.
\end{proof}
  
   Again, we call $\pi ({\bf H}_{EG,i})_{cpm}$ the CPM-decomposition of ${\bf H}_{EG,i}$. Replacing each circulant ${\bf H}_{EG ,i}$ in (58) by its CPM-decomposition $\pi ({\bf H}_{EG} ,i)$, we obtain the following $c\times cK_0$ array of CPMs and ZMs of size $l\times l$ over GF(2):
 \begin{equation}
                   \pi ({\bf H}_{EG ,qc})_{cpm}  = [\pi ({\bf H}_{EG ,0})_{cpm} \; \pi ({\bf H}_{EG ,1})_{cpm}\; \ldots \; \pi ({\bf H}_{EG ,K_0 -1})_{cpm} ] .               \end{equation} 
The array $\pi ({\bf H}_{EG} ,qc)_{cpm}$  is a sparse array with relatively small number of CPMs compared to the number of ZMs.  It also satisfied the RC-constraint.  Its null space gives a QC-EG-LDPC code which is combinatorially equivalent to the QC-EG-LDPC code ${\cal C}_{qc,m}$  given by the null space of ${\bf H}_{EG,qc}$ of (58).  For $1 \leq  \gamma  \leq  c$ and $1 \leq  \rho  \leq  cK_0$ , take a $\gamma \times \rho$  suarray $\pi ({\bf H}_{EG ,qc})(\gamma ,\rho )_{cpm}$  from $\pi ({\bf H}_{EG,qc})_{cpm} $.  This subarray is $\gamma l\times  \rho l$ matrix over GF(2). Its null space gives a QC-EG-LDPC code of length $\rho l$ which is referred to as a QC descendant of the QC-EG-LDPC code ${\cal C}_{qc,m}$  given by the null space of ${\bf H}_{EG,qc}$ of (58).  The above construction gives a large family of QC descendant LDPC codes of ${\cal C}_{qc,m }$.

   Again, a special case is $b = 1$ and $l = q - 1$.  In this case, $c = (q^{m-1} + q^{m-2} + . . . + q + 1)$ and $\pi ({\bf H}_{EG,qc})_{cpm}$ is a $c\times cK_0$  array of CPMs and ZMs of size $(q - 1)\times (q - 1)$ over GF(2).
 
   Consider the $c\times c$ subarray $\pi ({\bf H}_{EG,i})_{cpm}$ of CPMs and ZMs. As stated in Theorem 7, each column (or row block) consists of $q$ CPMs and $c - q$ ZMs.  Suppose $q$ can be factored as a product $e$ and $f$, i.e, $q = ef$.  We can split each column block of $\pi ({\bf H}_{EG,i})_{cpm}$ into $e$ column blocks of the same length with the $q$ CPMs evenly distributed  into the new $e$ column blocks, each with $f$ CPMs.  This column splitting operation is referred to \emph{column block splitting}.  In distributing the CPMs into $e$ new column blocks, their relative positions are not changed.  This column block splitting results in a $c\times c e$ array  ${\bf M}_{col,i} (e)$ of CPMs and ZMs of size $l\times l$, each column block consisting of $f$ CPMs and each row block consisting of $q$ CPMs.  Next, we split each row block of ${\bf M}_{col,i}(e)$ into $e$ new row blocks of the same length with the $q$ CPMs evenly distributed among the $e$ new row blocks, each with $f$ CPMs.  This row splitting operation is referred to as the \emph{row block splitting}.  This row block splitting of  ${\bf M}_{col,i} (e)$ results in a $ce\times ce$ array ${\bf M}_{col,row,i}(e,e)$ of CPMs and ZMs of size $l\times l$. The array ${\bf M}_{col,row,i}(e,e)$ is called the $e\times e$ \emph{expansion} of $\pi ({\bf H}_{EG,i})_{cpm}$. Each column block and each row block of ${\bf M}_{col,row,i}(e,e)$ consists of $f$ CPMs.  If we replace each $c\times c$ subarray $\pi ({\bf H}_{EG,i})_{cpm}$ in $\pi ({\bf H}_{EG})_{cpm}$ given by (60) with its $e\times e$ expansion ${\bf M}_{col,row,i}(e,e)$, we obtain the following $ce\times ce K_0$  array:
\begin{equation}
             {\bf M}_{EG,qc} = [{\bf M}_{col,row,0}(e,e)\;  {\bf M}_{col,row,1}(e,e)\; \cdots\; {\bf M}_{col,row,K_0 - 1}(e,e)].  
\end{equation}
Note that ${\bf M}_{EG,qc}$ has a much smaller density of CPMs than that of the array $\pi ({\bf H}_{EG,qc})_{cpm}$.

 \begin{example}
 Let $q = 2^3$.  Consider the 3-dimensional Euclidean geometry EG(3,$2^3$) over GF($2^3$).  This geometry has $q^3-1 = 2^{3\times 3} - 1 = 511$ non-origin points and $4599$ lines not passing through the origin of the geometry.  The $4599$ lines not passing through the origin can be partitioned into $9$ cyclic classes, each consisting of 511 lines.  Using the incidence vectors of the lines in these 9 cyclic classes, we can form 9 circulants, ${\bf H}_{EG,0}, {\bf H}_{EG,1}, . . . , {\bf H}_{EG,8}$, of sized $511\times 511$.  Factor 511 as the product of $b = 73$ and $l=q - 1 = 7$.  It follows from Theorem 6, each $511\times 511$ circulant ${\bf H}_{EG,i}$ can be decomposed into a $73\times 73$ array $\pi ({\bf H}_{EG,i})_{cpm}$ of CPMs and ZMs of size $7\times 7$.  Each column (row) block consists of 8 CPMs and 65 ZMs.  Form the following $73\times 657$ array of CPMs and ZMs of size $7\times 7$:
\[ 
        \pi ({\bf H}_{EG,qc})_{cpm} = [\pi ({\bf H}_{EG,0})_{cpm}\;  \pi ({\bf H}_{EG,1})_{cpm}\; \cdots\; \pi ({\bf H}_{EG,8})_{cpm}].   
 \]
This array is a $511\times 4599$ matrix with column and row weights 8 and 72, respectively.  The null space of this matrix gives a (8,72)-regular (4599,4227) QC-EG-LDPC code with rate 0.9191. 

Suppose we factor $q = 8$ as the product of $e = 2$ and $f = 4$.  Using column and row block splittings, each $73\times 73$ array $\pi ({\bf H}_{EG,i})_{cpm}$ can be expanded into a $146\times 146$ array ${\bf M}_{col,row,i}(2,2)$ of CPMs and ZMs of size $7\times 7$, each row and column block consisting of 4 CPMs and 142 ZMs.  Suppose we take first 8 of these $146\times 146$ arrays and form the following $146\times 1168$ array of CPMs and ZMs of size $7\times 7$:
 \[
             {\bf  M}_{EG}(8) = [{\bf M}_{col,row,0}(2,2)\;  {\bf M}_{col,row,1}(2,2)\; \cdots\; {\bf M}_{col,row,7}(2,2)].  
 \]
It is a $1022\times 8176$ matrix over GF(2) with column and row weight 4 and 32, respectively.  The null space of this matrix gives a (4,32)-regular (8176,7156) QC-EG-LDPC code with rate 0.8752.  This code is actually equivalent to the (4,32)-regular QC-EG-LDPC code adopted by NASA as the standard code for LANDSAT high-speed communications and other missions [15], [36] where the bit error rate requirement is $10^{-12}$.  The error performance of this code decoded with 50 iterations of the SPA and 15 iterations of the MSA are shown in Figure 8.  We see that there is no visible error-floor down to the BER of $10^{-14}$.  The estimated error-floor of this code is below the BER of $10^{-15}$.  At the BER of $10^{-14}$, it performs only 1.6 dB from the Shannon limit.  A hardware decoder for the NASA code has been built.  \twotriangle\end{example}

\section{Decomposition of Projective Geometry LDPC Codes}
 
  RC-constrained cyclic LDPC codes can also be constructed based on the incidence vectors of lines of finite projective geometries. For detail construction of this class of codes, the readers are referred to [5], [6], [15].  In the following, we consider the decomposition of a subclass of cyclic projective geometry (PG)-LDPC codes constructed based on the lines of two-dimensional projective geometries over finite fields (often called projective planes). 
 
   Consider the 2-dimensional projective geometry PG(2,$q$) pver GF($q$).  This geometry has $n = q^2 + q + 1$ points and $n = q^2 + q + 1$ lines [6], [15], [22], [26].  Each line contains of $q + 1$ points.  Two lines can have at most one point in common.  Let $\alpha$  be a primitive element of GF($q^3$).  Since $q^3-1= (q-1)(q^2 + q + 1)$, $n$ is a factor of $q^3-1$.  The $n$ points of PG(2,$q$) can be represented by the $n$ elements of $\{\alpha ^0, \alpha , \cdots , \alpha ^{n-1}\}$ [5], [6], [15]. The $q + 1$ points on a line are represented by the $q + 1$ elements in $\{\alpha ^0, \alpha , \cdots , \alpha ^{n-1}\}$.  Let $\cal L$ be a line in PG(2,$q$).  The incidence vector of this line $\cal L$ is an $n$-tuple over GF(2) defined as follows: ${\bf v}_{\cal L} = (v_0, v_1, \cdots  , v_{n-1})$ where $v_j = 1$ if $\alpha ^j$ is a point on $\cal L$, otherwise $v_j = 0$ for $0 \leq  j < n$.  Since $\cal L$ consists of $q +1$ points, the weight of ${\bf v}_{\cal L}$ is $q + 1$.  It is known that the cyclic-shift of ${\bf v}_{\cal L}$ is the incidence  of another line in PG(2,$q$) [6], [15].  The incidence vector ${\bf v}_{\cal L}$ and its $n-1$ cyclic-shifts are all different and give the incidence vectors of all the $n$ lines in PG(2,$q$).
 
   Form an $n\times n$ circulant  ${\bf H}_{PG}$  over  GF(2) with ${\bf v}_{\cal L}$ and its $n-1$ cyclic-shifts as rows.  The columns and rows of ${\bf H}_{PG}$ correspond to the points and lines of PG(2,$q$), respectively.  Both column and row weights of  ${\bf H}_{PG}$  are equal to $q + 1$.  Since two lines in a projective geometry can have at most one point in common, their incidence vectors can have at most one place where they both have 1-components.  Hence,  ${\bf H}_{PG}$  satisfies the RC-constraint. Therefore, the null space of  ${\bf H}_{PG}$  gives an RC-constrained cyclic-PG-LDPC code  ${\cal C}_{PG}$  of length $n = q^2 + q +1$ and minimum distance at least $q + 2$, whose Tanner graph has a girth of at least 6. 
 
   For the special case $q = 2^s$, the rank of  ${\bf H}_{PG}$  is $3^s + 1$ [5], [6], [15], [27] and the cyclic PG-LDPC code  ${\cal C}_{PG}$  has the following parameters: 1) Length $n = 2^{2s} + 2^s + 1$; 2) Dimension $n-3^s-1$; 3) Minimum distance $\geq 2^s + 2$. The roots of the generator ${\bf g}(X)$ of  ${\cal C}_{PG}$  can be determined and are given in [5], [6],  [37]. 
   
    Let $c$ and $l$ be two proper factors of $n$ such that $n = c\cdot l$.  Then, through column and row permutation $\pi$  defined by (2) and (3), the circulant  ${\bf H}_{PG}$  can be decomposed into an RC-constrained $c\times c$ array $\pi ( {\bf H}_{PG} )$ of circulants of size of $l\times l$.  The null space of each nonzero $l\times l$ circulant in $\pi ( {\bf H}_{PG} )$ gives an RC-constrained cyclic PG-LDPC code of length $l$.  For any pair $(\gamma ,\rho )$ of integers with $1 \leq  \gamma , \rho  \leq  l$, the null space of any $\gamma \times \rho$  subarray of $\pi ( {\bf H}_{PG} )$ gives a QC-PG-LDPC code of length $\rho l$.
 
\begin{example} Let the two-dimensional projective geometry PG(2,$2^6$ ) over GF($2^6$) be the code construction geometry.  This geometry has $(2^{3s}-1)/(2^s-1) = 4161$ points and 4161 lines.  Each line consists of 65 points.  Based on the lines of PG(2,$2^6$), we can construct an RC-constrained $4161\times 4161$ circulant  ${\bf H}_{PG}$  with both column and row weights equal to 65.  The null space of this ciruclant gives a (65,65)-regular (4161,3431) cyclic PG-LDPC code with minimum distance at least 66.  The error performances of this code over the AWGN channel decoded with 5, 10 and 50 iterations of the SPA are shown in Figure 9(a).  We see that the decoding of this code converges very fast.    Since 4161 can be factored as the product of 3 and 1387.  Let $c = 3$ and $l = 1387$.  Then  ${\bf H}_{PG}$  can be decomposed into a $3\times 3$ array $\pi ( {\bf H}_{PG} )$ of circulants of size $1387\times 1387$ in the form of (4).  Let $\Psi _0$, $\Psi _1$ and $\Psi _2$ be the 3 circulants in the first row block of $\pi ( {\bf H}_{PG} )$.  The column and row weights of the circulant $\Psi _1$ are both 19.  The null space of $\Psi _1$ gives an RC-constrained (1387,720) cyclic-PG-LDPC code with minimum distance at least 20.  Its error performance over the AWGN channel decoded with 50 iterations of SPA is shown in Figure 9(b).   \twotriangle \end{example}

   Note that $n = q^2 + q + 1$ is not divisible by $q-1$. The PG-circulant  ${\bf H}_{PG} $ cannot be decomposed into an array of CPMs of size $(q-1)\times (q-1)$.  Decomposition of circulants constructed based on projective geometries of dimensions higher than two can be carried out similar to the decomposition of high dimensional Euclidean geometries, except for the CPM-decomposition.

\section{Trapping Sets of RC-Constrained LDPC Codes}

   It has been observed for most LDPC codes, decoded with iterative message-passing decoding algorithms such as the SPA or the MSA, that as the SNR continues to increase, the error probability \emph{suddenly drops at a rate much slower than} that in the region of low to moderate SNR (or even stops to drop, i.e., \emph{the error performance curve flattens out}). This phenomenon, known as \emph{error-floor}, may preclude LDPC codes from applications requiring very low error rates.  High error-floors most commonly occur for unstructured random or pseudo-random LDPC codes constructed using computer based methods or algorithms.  Structured LDPC codes constructed algebraically, such as finite geometry and finite field LDPC codes [5]-[13], in general have much lower error-floors. 

   Ever since the phenomenon of the error-floors of LDPC codes with iterative decoding became known [38], a great deal of research effort has been expended in finding its causes and methods to resolve or mitigate the error-floor problem. For the AWGN channel, the error-floor of an LDPC code is mostly caused by an undesirable structure, known as \emph{trapping set} [14], [15], in the Tanner graph of the code based on which the decoding is carried out. 

\subsection{Concepts and Definitions}

   Let $\cal C$ be an LDPC code of length $n$ given by the null space of a sparse $m\times n$ parity-check  matrix ${\bf H} = [h_{i,j}], 0\leq i<m, 0\leq j<n$ over GF(2) with $m$ rows and $n$ columns. The Tanner graph [4] $\cal G$ of $\cal C$ is a \emph{bipartite graph} with two sets of nodes, the \emph{variable nodes} (VNs) and the \emph{check nodes} (CNs). The VNs, denoted by $v_0, v_1, . . . , v_{n-1}$, represent the $n$ code bits of a codeword ${\bf v} = (v_0, v_1, . . . , v_{n-1})$ in the code and the CNs, denoted by $c_0, c_1, . . . , c_{m-1}$, represent the $m$ \emph{(parity) check-sum constraints} that the code bits must satisfy (they must be all equal to zero).  For convenience, we do not distinguish a ``code bit'' and a ``VN'', or a ``check-sum'' and a ``CN''.  We will use the notation $v_j$ for both the $j$th code bit and its corresponding VN and the notation $c_i$ for both the $i$th check-sum and its corresponding CN.  A VN $v_j$ is connected to a CN $c_i$ by an \emph{edge} if and only if the code bit $v_j$ is contained in the check-sum $c_i$.  Basically, the VNs correspond to the $n$ columns of the parity-check matrix $\bf H$  and CNs correspond to the $m$ rows of $\bf H$. The $j$th VN $v_j$ is connected to the $i$th CN $c_i$ by an edge if and only if $h_{i,j} = 1$.  The degree $d_{v_j}$ of the VN $v_j$ is defined as the number of CNs connected to $v_j$ and the degree $d_{c_i}$ of the CN $c_i$ is defined as the number of VNs connected to the CN $c_i$. The degree $d_{v_j}$ of the VN $v_j$ is simply equal to the number of 1-entries in the $j$th column of the parity-check matrix ${\bf H} = [h_{i,j}]$ and the degree $d_{c_i}$ of the CN $c_i$ is simply equal to the number of 1-entries in the $i$th row of ${\bf H} = [h_{i,j}]$.  For a ($\gamma$,$\rho$)-regular LDPC code, all the VNs have the same degree $\gamma$  and all the CNs have the same degree $\rho$.  For an irregular code, its Tanner graph has varying VN degrees and/or varying CN degrees.  It is clear that the number of edges in the Tanner graph of an LDPC code is equal to the total number of 1-entries in the code's parity-check matrix $\bf H$.

      Figure 10(a) shows that the Tanner graph of a (3,3)-regular (7,3) LDPC code given by the null space of the following RC-constrained parity-check matrix:
                    \[
                     {\bf H} = \left[
                     \begin{array}{lllllll}
                                1 & 0 & 1 & 1 & 0 & 0 & 0\\
                                          0& 1& 0 &1 &1& 0 &0\\
                                          0& 0& 1& 0 &1 &1& 0\\
                                   0 &0 &0& 1& 0 &1 &1\\
                                          1 &0& 0& 0& 1& 0& 1\\
                                          1 &1& 0& 0& 0& 1& 0\\
                                          0 &1& 1& 0& 0& 0& 1
                                          \end{array}\right].
\]

       \begin{definition}[]                                                                    
 Let $\cal G$ be the Tanner graph of a binary LDPC code $\cal C$ given by the null space of an $m\times n$ matrix $\bf H$ over GF(2). For $1 \leq  \kappa \leq  n$ and $0 \leq \tau \leq  m$, a ($\kappa,\tau$) trapping set [14] is a set ${\cal T}(\kappa,\tau)$ of $\kappa$ VNs in $\cal G$  which induce a subgraph of $\cal G$ with exactly $\tau$ odd-degree CNs (and an arbitrary number of even-degree CNs). An elementary $(\kappa ,\tau )$ trapping set [20] is a trapping set for which all CNs in the induced subgraph of the Tanner graph have degree one or degree two, and there are exactly $\tau$  degree-one CNs.
\end{definition}

 In an elementary trapping set, every CN of degree 1 is connected to a single VN and every CN of degree 2 (if exists) is connected to two VNs. Figures 10(b) and 10(c) shows two subgraphs of the Tanner graph of a (3,3)-regular LDPC code shown in Figure 10(a) which are induced by a (3,3) trapping set and a (4,4) trapping set, respectively.  The (3,3) trapping set consists of 3 VNs, $v_1$, $v_4$ and $v_6$.  The subgraph induced by this trapping set has 3 CNs of degree 1 and 3 CNs of degree 2.  Therefore, this trapping set is an elementary trapping set.  The (4,4) trapping set consists of 4 VNs, $v_1$, $v_2$, $v_4$ and $v_6$.  The subgraph induced by this trapping set has 3 CNs of degree 1, one CN of degree 3 and 3 CNs of degree 2.

   Suppose, in transmission of a codeword, an error pattern $\bf e$ with $\kappa$ errors at the locations of the $\kappa$ VNs of a $(\kappa, \tau)$ trapping set occurs. This error pattern will cause $\tau$ parity-check failures (i.e., the check-sums are not equal to zeros, because each of these $\tau$  check-sums contain an odd number of errors in $\bf e$). In this case, for iterative decoding, another decoding iteration must be carried out to correct the failed check-sum.  Iterative decoding, such as the SPA and MSA, is very susceptible to trapping sets of a code because it works locally in a distributed-processing manner. Each CN has a local processor unit to process the messages received from the VNs connected to it and each VN has a local processor unit to process the messages received from the CVs connected to it.  Hopefully, these local processor units through iterations and message exchanges collect enough information to make a global optimum decision of the transmitted code bits.

   In each decoding iteration, we call a CN a \emph{satisfied} CN if it satisfies its corresponding check-sum constraint (i.e., its corresponding check-sum is equal to zero), otherwise, call it an \emph{unsatisfied} CN.  During the decoding process, the decoder undergoes \emph{state transitions} from one state to another until all the CNs satisfy their corresponding check-sum constraints or a predetermined maximum number of iterations is reached. The $i$th state of an iterative decoder is represented by the hard-decision sequence obtained at the end of $i$th iteration. In the process of a decoding iteration, the messages from the satisfied CNs try to \emph{reinforce} the current decoder state, while the messages from the unsatisfied CNs try to \emph{change} some of the bit decisions to satisfy their check-sum constraints. If errors affect the $\kappa$ code bits (or the $\kappa$ VNs) of a $(\kappa,\tau)$ trapping set ${\cal T}(\kappa,\tau)$, the $\tau$ odd-degree CNs, each connected to an odd number of VNs in ${\cal T}(\kappa,\tau)$, will not be satisfied while all other CNs will be satisfied.  The decoder will succeed in correcting the errors in ${\cal T}(\kappa,\tau)$  if the messages coming from the unsatisfied  CNs connected to the VNs in ${\cal T}(\kappa,\tau)$  are \emph{strong enough} to overcome the  (false or inaccurate) messages coming from the satisfied CNs.  However, this may not be the case if $\tau$ is \emph{small}.  As a result, the decoder may not converge to a valid codeword even if more decoding iterations are performed and this non-convergence of decoding results in an error-floor.   In this case, we say that the decoder is trapped.

   For the binary-input AWGN channel, error patterns with small number of errors (or low-weight error-patterns) are more probable to occur than error patterns with larger number of errors.  Consequently, in message-passing decoding algorithms, the most harmful $(\kappa,\tau)$ trapping sets are usually those with small values of $\kappa$ and $\tau$, especially when the value of $\tau$ is very small compared to that of $\kappa$.  Extensive study and simulation results [13], [38]-[68] show that the trapping sets that result in high decoding failure rates and contribute significantly to high error-floors are those with small values $\kappa$ and small ratios $\tau/\kappa$.  We call these trapping sets \emph{small trapping sets}. The trapping sets with large values $\tau$ relative to values $\kappa$ in general result in relatively small decoding failure rates and contribute little to error-floor. From extensive computer simulations reported in the literature [14], [38]-[68], it has been observed that most trapping sets that exert a strong influence on the error-floor are of the elementary trapping sets and trapping sets with $\tau/\kappa \leq 1$.

   Besides small trapping sets and their distributions, undetected errors caused by small minimum weight of a code also contribute considerably to the error-floor of the code.  If there are no trapping sets with size $\kappa$ smaller than the minimum weight of an LDPC code, then the error-floor of the code decoded with iterative decoding is dominated by the minimum weight of the code.  For $\tau  = 0$, ${\cal T}(\kappa ,0)$ is a special trapping set with no odd-degree CN. Such a trapping set is induced by an error pattern which is identical to a codeword of weight $\kappa$ .  When such a trapping set occurs, the decoder converges into an incorrect codeword and commits an undetected error.  In this case, we say that the decoder is  trapped into a fixed point.

 The notion of a small trapping set given above is loosely defined.  A more quantitative definition of small trapping set was given by L\"adner and Milenkovic [20].
\begin{definition}   A $(\kappa ,\tau )$ trapping set in the Tanner graph of an LDPC code of length $n$ is said to be small if $\kappa  \leq \sqrt{n}$ and $\tau  \leq  4\kappa $ (i.e. the ratio $\tau/\kappa \leq 4$).     \end{definition} 

   Since Richardson introduced the notion of trapping sets and their effect on error-floor in 2003 [14], a great deal of research effort has been expended in analyzing the general structure of trapping sets of LDPC codes, developing methods (or algorithms) for finding trapping sets (especially the harmful ones), techniques to remove small trapping sets, and devising decoding strategies to remove or reducing the degrading effect of harmful trapping sets,[13], [38]-[68].  The research effort expended so far still leaves the trapping set problem largely unsolved.  However, study and extensive computer simulations have shown that among the trapping sets contribute significantly to the error-floor, the harmful ones are mostly the small trapping sets, especially the small elementary trapping sets with $\tau/\kappa \leq 1$.

   Constructing (or designing) codes to avoid harmful trapping sets to mitigate error-floor problem is a hard combinatorial problem, just like finding the number of the minimum weight codewords (or the weight distribution) of a linear code.  Consequently, to lower the error-floor of an LDPC code caused by (small) trapping sets, an alternate approach is taken.  A most commonly taken approach is a \emph{decoder-based strategy} to remove or reduce the effect of harmful trapping sets on error-floor.  Several such decoder-based strategies have been recently proposed [53], [54], [56], [58], [61], [67], [68].  Among them, the most effective decoding strategy is the \emph{backtracking iterative decoding algorithm} recently presented in [68].
   
\subsection{An Analysis of Trapping Sets of the RC-Constrained LDPC Codes}

   In this section, we present an analysis of trapping set structure of an RC-constrained regular LDPC code.  The analysis is based on the RC-constraint on the rows and columns of the parity-check matrix $\bf H$ and its column weight $\gamma$. For such an RC-constrained LDPC code, its minimum weight is at least $\gamma  + 1$.  We will show that there is no $(\kappa,\tau)$ trapping set with $\kappa \leq  \gamma$  and $\tau < \gamma$.  More precisely, any trapping set $(\kappa,\tau)$ with $\kappa$ VNs, if $\kappa \leq  \gamma $, then the number of odd-degree CNs is at least  $\gamma  + 1$, i.e.,$\tau > \gamma $.  This is to say that for an RC-constrained ($\gamma $,$\rho $)-regular LDPC code, there is no harmful trapping set with size smaller than $\gamma$.  Particularly, we will show that an RC-constrained ($\gamma $,$\rho $)-regular LDPC code whose parity-check matrix has column weight $\gamma$  has no small elementary trapping sets of the type defined by Definition 2. Cyclic EG- and PG-LDPC codes given in [5] and their cyclic and QC descendants given in Sections IV and VI of this paper are RC-constrained LDPC codes and whose parity-check matrices have large column weights, hence they don't have harmful small trapping sets with size smaller than $\gamma $. Besides the FG-LDPC codes, LDPC codes constructed based on finite fields and experimental designs in [10]-[13], [69]-[76] are also RC-constrained LDPC codes.

   Let $\cal C$ be a binary ($\gamma $,$\rho $)-regular LDPC code of length $n$ given by the null space of  an RC-constrained $m\times n$ matrix ${\bf H} = [h_{i,j}]_{0\leq i<m, 0\leq j<n}$ over GF(2) with column and row weights $\gamma$  and $\rho $, respectively.   Let ${\bf h}_0, {\bf h}_1, . . . , {\bf h}_{m-1}$ denote the rows of $\bf H$, where the $i$th row ${\bf h}_i$ is given by the following $n$-tuple over GF(2): $          {\bf h}_i = (h_{i,0}, h_{i,1}, . . . ,  h_{i,n-1}),$
for $0\leq  i < m$.  An $n$-tuple ${\bf v} = (v_0, v_1, . . . , v_{n-1})$  over GF(2) is a codeword in ${\cal C}$ \emph{if and only if} ${\bf v}\cdot {\bf H}^T = 0$ (a zero $m$-tuple). The condition ${\bf v}\cdot {\bf H}^T = 0$ gives the following $m$ constraints on the bits of the codeword ${\bf v}$:
\begin{equation}
          c_i = {\bf v}\cdot {\bf h}_i = v_0 h_{i,0} + v_1 h_{i,1} + . . . + v_{n-1} h_{i,n-1} = 0, 
          \end{equation}
for $0 \leq  i < m$, where ${\bf v}\cdot {\bf h}_i$ is the inner product of $\bf v$ and ${\bf h}_i$. The above $m$ linear sums of code bits are called \emph{parity-check-sums} (or simply check-sums).  The $m$ check-sums of the code bits equal to 0 are the constraints that the code bits of any codeword must satisfy.
 
   For $0 \leq  j < n$, if $h_{i,j} = 1$, then the $j$th code bit $v_j$  participates (or is contained) in the $i$th check-sum $c_i$ given by (62).  In this case, we say that the $i$th check-sum $c_i$ \emph{checks on} the $j$th code bit $v_j$ of ${\bf v}$ (or the $j$th code bit $v_j$ of $\bf v$ is \emph{checked} by the $i$th check-sum $c_i$).  Since $\bf H$ has constant column weight $\gamma$ , there are $\gamma$  check-sums containing (or checking on) the code bit $v_j$. Since every row of $\bf H$ has weight $\rho $, each check-sum $c_i$ checks on $\rho$  code bits.  Since $\bf H$ satisfies the RC-constraint, no two different rows of $\bf H$ have more than one position where they both have 1-components.  This implies that no two different code bits, $v_{j_1}$ and $v_{j_2}$, are \emph{simultaneously checked} by two different check-sums, $c_{i_1}$ and $c_{i_2}$.
   
   Suppose a codeword ${\bf v} = (v_0, v_1, . . . , v_{n-1})$ in $\cal C$ is transmitted over the binary-input AWGN channel. Let ${\bf z} = (z_0, z_1, . . . , z_{n-1})$ over GF(2) be the \emph{hard-decision received vector  (or sequence)}.  The $j$th received bit $z_j$ of $\bf z$ is simply an estimate of the $j$th code bit $v_j$ of the transmitted codeword $\bf v$.  If $z_j = v_j$ for $0 \leq  j <n$, then $\bf z = v$; otherwise, $\bf z$ contains \emph{transmission errors}.  Therefore, $\bf z$ is an estimate of the transmitted codeword $\bf v$ prior channel decoding.  Let
  \[  \begin{array}{lll}
              {\bf e} & = &(e_0, e_1, . . . , e_{n-1}),\\
                 &= &(z_0, z_1, . . . , z_{n-1})  + (v_0, v_1, . . . , v_{n-1}),\\
                 &= &(z_0 + v_0, z_1 + v_1 +\cdots+ z_{n-1} + v_{n-1}).
                 \end{array}\]
where, for $0 \leq  j < n$,  $e_j = z_j + v_j$ and $``+''$ is  modulo-2 addition.  If $z_j \neq  v_j$, then $e_j = 1$ otherwise $e_j = 0$.  Therefore, the positions in $\bf e$ where the components equal to ``1'' are the erroneous positions. The $n$-tuple $\bf e$ gives the pattern of errors contained in the received sequence $\bf z$ and is called the \emph{error pattern} contained in $\bf z$ [6]. Hence $\bf z=v+e$.

   For any decoding algorithm (soft or hard), the first step is to compute the syndrome of $\bf z$ [6],
   \begin{equation}
                          {\bf s} = (s_0, s_1, . . . , s_{m-1}) = {\bf z}\cdot {\bf H}^{\sf T},                        
                          \end{equation}
where,
\begin{equation}
                        s_i = {\bf z}\cdot {\bf h}_i = z_0 h_{i,0} + z_1 h_{i,1} + . . . + z_{n-1}h_{i,n-1},   
                        \end{equation}
for $0 \leq  i < m$, which is called a \emph{syndrome-sum} of $\bf z$.  If $\bf s = 0$, then the received bits in $\bf z$ satisfy all the $m$ check-sum constraints given by (62) and $\bf z$ is a codeword.  In this case, the receiver assumes that $\bf z$ is the transmitted codeword and accepts it as the \emph{decoded} codeword. If $\bf s \neq  0$, the received bits in $\bf z$ do not satisfy all the $m$ check-sum constraints given by (62) and $\bf z$ is not a codeword.  In this case, we say that errors in $\bf z$ are \emph{being detected} and the error pattern is called a \emph{detectable error pattern}.  Then an error correction process is initiated. Since $\bf z = v + e$ and ${\bf v}\cdot {\bf h}_i = 0$, it follows from (64) that each syndrome- sum is actually a linear sum of a set of error bits contained in the received sequence $\bf z$,
\begin{equation}
                       s_i = {\bf e}\cdot{\bf h}_i = e_0h_{i,0} + e_1 h_{i,1} + . . . + e_{n-1}h_{i,n-1}, 
                       \end{equation}
If $\bf z$ is error-free, $s_i = c_i = 0$ for $0 \leq  i < m$. If $\bf z$ is not error-free but the error pattern $\bf e$ happens to be identical to a nonzero codeword in $\cal C$, all the $m$ syndrome-sums will be equal to 0.  In this case, the received sequence $\bf z$ contains an undetected error pattern and decoding results in an incorrect decoding.  Decoding process is initiated (or continues in iterative decoding) only if not all the syndrome-sums are equal to zero.

   From (65), we see that a syndrome-sum $s_i$ is equal to ``1'' if and only if the number of nonzero error digits checked by $s_i$ (or participate in the sum $s_i$) is \emph{odd}.  A syndrome-sum $s_i$ is equal to zero if and only if either all the error bits checked by $s_i$ are error-free or the number of nonzero error bits checked by $s_i$ is even. Let $\bf e$ be an error pattern with $\kappa$ nonzero error bits which cause $\tau$ nonzero syndrome-sums and an arbitrary number of zero syndrome-sums.  Construct a subgraph ${\cal G}(\kappa,\tau)$ of the Tanner graph ${\cal G}$ of the code with a set ${\cal T}(\kappa,\tau)$ of $\kappa$ VNs.  These $\kappa$ VNs correspond to the $\kappa$ nonzero error bits in the detectable error pattern $\bf e$ and are connected to $\tau$ CNs which correspond to the $\tau$ nonzero (failed) syndrome-sums and/or some CNs which correspond to zero syndrome-sums but are adjacent to the VNs in ${\cal T}(\kappa,\tau)$.  In this subgraph, the CNs corresponding to the nonzero (failed) syndrome-sums have odd degrees and the other CNs have even-degrees.  This subgraph ${\cal G}(\kappa,\tau)$ is said to be induced by the detectable error pattern $\bf e$ and the set ${\cal T}(\kappa,\tau)$ is a trapping set as defined in Definition 1.

   A syndrome-sum $s_i$ that contains an error bit $e_j$ is said to \emph{check on}  $e_j$.  Since each column of the parity-check matrix $\bf H$ has column weight $\gamma $, there are $\gamma$  syndrome-sums checking on every error bit $e_j$, i.e, every error bit is checked by $\gamma$  syndrome-sums (or contained in $\gamma$  syndrome-sums).
 Since each row of $\bf H$ has weight $\rho $, each syndrome-sum checks on $\rho$  error bits.  Since $\bf H$ satisfies the RC-constraint, \emph{no two error bits can be checked simultaneously by two syndrome-sums}. 

  For $0 \leq  i < m$ and $0 \leq  j < n$, we define the following two \emph{index sets}:
\begin{eqnarray}
                        {\cal N}_i = \{j: 0\leq  j <n, h_{i,j} = 1\},         \\
                        {\cal M}_j = \{i: 0\leq  i <m, h_{i,j} = 1\}.           
                        \end{eqnarray}
The indices in ${\cal N}_i$ are simply the locations of the 1-components in the $i$th row ${\bf h}_i$ of $\bf H$. ${\cal N}_i$ is called the \emph{support} of ${\bf h}_i$.  The indices in ${\cal M}_j$ give the rows of $\bf H$ whose $j$th components are equal to ``1''.  We call ${\cal M}_j$ the \emph{support} of $j$th code bit $v_j$. Since $\bf H$ satisfies the RC-constraint, it is clear that: 1) for $0 \leq  i_1, i_2 < m$ and $i_1 \neq  i_2$, ${\cal N}_{i_1}$ and ${\cal N}_{i_2}$ have \emph{at most one} index in common; and 2) for $0 \leq  j_1, j_2 < n$ and $j_1 \neq  j_2$, ${\cal M}_{j_1}$ and ${\cal M}_{j_2}$ have at most one index in common.  Since $\bf H$ has constant column weight $\gamma$  and constant row weight $\rho $, $|{\cal M}_j| = \gamma$  for $0\leq j<n$ and $|{\cal N}_i| = \rho$  for $0 \leq  i <m$. 
 
   For $0 \leq  j < n$, define the following set of rows of $\bf H$:
   \begin{equation}
                       {\cal A}^{(j)} = \{{\bf h}^{(j)}_i: i\in {\cal M}_j \}.       
                       \end{equation}
Then it follows from the RC-constraint on the rows of $\bf H$ that ${\cal A}^{(j)}$ has the following structural properties: 1) every row ${\bf h}^{(j)}_i$ in ${\cal A}^{(j)}$ has a 1-component at the position $j$; (2) any 1-component at a position other than $j$th position can appear in at most one row in ${\cal A}^{(j)}$; and (3) for $0 \leq  j_1, j_2  < n$, and $j_1 \neq  j_2$, ${\cal A}^{(j_1)}$ and ${\cal A}^{(j_2)}$ can have \emph{at most one row in common}. The rows in ${\cal A}^{(j)}$ are said to be \emph{orthogonal} on the $j$th code bit $v_j$.

   For $0 \leq  j < n$, define the following set of syndrome-sums:
   \begin{equation}
                              S^{(j)} = \{s^{(j)}_i= {\bf e} \cdot {\bf h}^{(j)}_i: {\bf h}^{(j)}_i\in {\cal A}^{(j)}\}.        
                              \end{equation}
      Then, the $j$th error bit $e_j$ of the error pattern $\bf e$ is checked by (contained in) every syndrome-sum in $S^{(j)}$ and any error bit other than $e_j$ is checked by at most one syndrome-sum in $S^{(j)}$.  Each syndrome-sum in $S^{(j)}$ can be  expressed  as  follows: for $i\in {\cal M}_j$,
      \begin{equation}
              s^{(j)}_i = e_j  +  \sum\limits_{l\in {\cal N}_i, \;l\neq j}  e_l h_{i,l}.                                       
                                       \end{equation}
The syndrome-sums in $S^{(j)}$ are said to be \emph{orthogonal} on the error bit $e_j$ and are called \emph{orthogonal syndrome-sums} on $e_j$.  The RC-constraint on the parity-check matrix $\bf H$ (or property-3 of ${\cal A}^{(j)}$) ensures that any two different orthogonal sets $S^{({j_1})}$ and $S^{(j_2)}$ can have \emph{at most one syndrome-sum in common}. Basically, under the RC-constraint, if  two rows in $\bf H$ have 1-components at two different positions, then the two rows must be \emph{identical}.

   Consider an error pattern ${\bf e} = (e_0, e_1, . . . , e_{n-1})$ with a single error at $j$th position with $0 \leq  j < n$, i.e., $e_j = 1$.  For this single error pattern, all the $\gamma$  syndrome-sums in $S^{(j)}$ orthogonal on $e_j$ are equal to ``1''.  Since $e_j$ is only checked by the syndrome-sums in $S^{(j)}$, all the syndrome-sums in any other orthogonal syndrome set are error free and equal to zero.  Consequently, the trapping set correspond to this single error pattern is a (1,$\gamma$) trapping set ${\cal T}(1,\gamma )$ with one VN and $\gamma$  CNs of degree 1.  Each of these the CNs is connected to the VN $v_j$ and has degree one. It is clear that ${\cal T}(1,\gamma )$ is an elementary trapping set.  If $\gamma  > 4$, it is not a small trapping set of the type defined by Definition 2.

   Next, we consider an error pattern $\bf e$ with two errors at positions, $j_1$ and $j_2$, i.e., $e_{j_1} = e_{j_2} = 1$.  Then all the $\gamma$  syndrome-sums in the orthogonal set $S^{(j_1)}$ check on $e_{j_1}$ and all the $\gamma$  syndrome-sums in the orthogonal set $S^{(j_2)}$ check on $e_{j_2}$.  If $S^{(j_1)}$ and $S^{(j_2)}$ are disjoint, then all the $\gamma$  syndrome-sums in $S^{(j_1)}$ and all the $\gamma$  syndrome-sums in $S^{(j_2)}$ are equal to ``1''.    In this case, the trapping set corresponds to the error pattern $\bf e$ with double errors is a (2,$2\gamma $) elementary trapping set ${\cal T}(2,2\gamma )$ with 2 VNs and $2\gamma$  CNs of degree 1.  If $S^{(j_1)}$ and $S^{(j_2)}$ are not disjoint, then they have exactly one common syndrome-sum which checks on both $e_{j_1}$ and $e_{j_2}$ and hence this common syndrome-sum is equal to zero.  In this case, the trapping set corresponds to the error-pattern $\bf e$ with double errors is a $(2,2(\gamma - 1))$ elementary trapping set with 2 VNs, $2(\gamma  - 1)$ CNs of degree 1 and one CN of degree 2.   For $\gamma  > 5$, it follows from Definition 2 that for either case, the trapping is not small. For $\gamma  > 2$, the number of odd-degree CNs is greater than $\gamma $.  The above analysis shows that the trapping set corresponding to an error pattern with two errors has at least $2(\gamma - 1)$ CNs of odd degrees.

   Consider an error pattern $\bf e$ with three errors at the positions, $j_1$, $j_2$ and $j_3$.  The trapping set corresponds to this error pattern has several possible configurations depending the locations of the three errors.  The first possible configuration is such that the three errors are checked separately by three mutually disjoint orthogonal sets, $S^{(j_1)}$, $S^{(j_2)}$ and $S^{(j_3)}$.  In this case, the trapping set corresponding to the error pattern $\bf e$ is a $(3,3\gamma )$ elementary trapping set ${\cal T}(3,3\gamma )$ with 3 VNs and $3\gamma$  CNs of degree 1, no CN with even-degree.  The second possible configuration is such that two orthogonal sets, say $S^{(j_1)}$ and $S^{(j_2)}$, have a common syndrome-sum and the third orthogonal set $S^{(j_3)}$ is mutually disjoint with $S^{(j_1)}$ and $S^{(j_2)}$. In this case, error bits, $e_{j_1}$ and $e_{j_2}$, are jointly checked by a common syndrome-sum in $S^{(j_1)}$ and $S^{(j_2)}$ and the error bit $e_{j_3}$ is checked only by the syndrome-sums in $S^{(j_3)}$. Then, the trapping set corresponding to this triple-error pattern $\bf e$ is a $(3,3\gamma-2)$ elementary trapping set with three VNs, $3\gamma-2$ CNs of degree 1 and one CN of degree 2.  The third possible configuration is such that all three errors are checked by a syndrome-sum which is contained in all three orthogonal sets, $S^{(j_1)}$, $S^{(j_2)}$ and $S^{(j_3)}$.  In this case, all the $\gamma$ syndrome-sums in each of the orthogonal sets, $S^{(j_1)}$, $S^{(j_2)}$, and $S^{(j_3)}$, are nonzero.  The common syndrome-sum in these three orthogonal sets contains the three errors, $e_{j_1}, e_{j_2}, e_{j_3}$, and all the other syndrome sums contain only one of these three errors.  Consequently, the trapping set corresponding to the error pattern $\bf e$ is a $(3,3\gamma  -2)$ tapping set with $3\gamma - 3$ CNs of degree 1 and one CN of degree 3 (no CN with even degree).  The fourth possible configuration of a trapping set corresponding to a triple error pattern $\bf e$ is such that all three errors $e_{j_1}$, $e_{j_2}$ and $e_{j_3}$ are checked by the syndrome-sum common to two orthogonal sets, say $S^{(j_1)}$ and $S^{(j_2)}$, and one error, say $e_{j_3}$ is checked by $S^{(j_3)}$ alone.  In this case, the trapping set corresponding to this triple error pattern $\bf e$ is a (3,$3\gamma - 1$) trapping set with $3\gamma-2$ CNs of degree 1, one CN of degree 3 and one CN with degree-2. The fifth possible configuration is that $S^{(j_1)}$ and $S^{(j_2)}$ have a common syndrome sum checking on $e_{j_1}$ and $e_{j_2}$, $S^{(j_1)}$ and $S^{(j_3)}$ have a common syndrome sum checking on $e_{j_1}$ and $e_{j_3}$, and $S^{(j_2)}$ and $S^{(j_3)}$ are disjoint.  For this conjuration, the trapping set is a (3,$3\gamma  - 4$) elementary trapping set, with $3\gamma  - 4$ CNs of degree-1 and two CN of degree-2.  The sixth possible configuration is such that the pair of errors, $(e_{j_1},e_{j_2})$, is checked by the common syndrome-sum in $S^{(j_1)}$ and $S^{(j_2)}$, the pair $(e_{j_1}, e_{j_3})$ is checked by the common syndrome-sum in $S^{(j_1)}$ and $S^{(j_3)}$, and the pair $(e_{j_2}, e_{j_3})$ is checked by $S^{(j_2)}$ and $S^{(j_3)}$.  Corresponding to this configuration, the trapping set is a $(3,3(\gamma  - 2))$ elementary trapping set with $3(\gamma - 2)$ CNs of degree 1 and 3 CNs of degree 2.  Consider the 6 possible configurations of three errors, the trapping set with minimum number of odd-degree CNs is the sixth configuration.  In this case, the number of CNs of odd-degree (degree 1) is at least $3(\gamma - 2)$.  If $\gamma  \geq 3$, the number of odd-degree CNs in a trapping set correspond to a triple-error pattern is greater than or at least equal $\gamma$.   For $\gamma  > 4$,  the trapping sets corresponding to the first 4 configurations are not small trapping sets of the types defined by Definition 2.  If $\gamma  > 6$, then the trapping sets corresponding to all 6 configurations are not small trapping sets of the types defined by Definition 2. Among all the 6 possible configurations of errors, the sixth one has the smallest number of CNs of odd degree.  For this configuration, the 3 errors are \emph{uniformly distributed in pairs} among the three sets of orthogonal syndrome-sums.  This maximizes the number of CNs of degree-2 and hence minimizes the number of degree-1.

   For $1 \leq  t \leq  \gamma$  and $0 \leq  j_1, j_2, . . . , j_t \leq  \gamma $, consider an error pattern of $t$ errors at the positions, $j_1, j_2, . . . , j_t$.  For large $t$ and $\gamma $, to analyze all the possible configurations of trapping sets with $t$ VNs is very hard if not impossible.  However, a lower bound on the minimum number of odd-degree CNs can be derived.  A configuration that results in a minimum number of odd-degree CNs is such for $0 \leq  r$, $s \leq  t$, every pair $(e_{j_r}, e_{j_s})$ of errors is checked by the common syndrome-sum in $S^{(j_r)}$ and $S^{(j_s)}$. This configuration actually maximizes the number of CNs with even degrees (all degree 2).  If this configuration exists, it results in a $(t, t(\gamma  - (t - 1)))$ trapping set ${\cal T}(t, t(\gamma  - (t - 1)))$ with $t(\gamma  - (t - 1))$ CNs of degree 1 and $(t-1)!$ CNs of degree 2. It is an elementary trapping set.  Any other configurations of $t$ errors would results in a trapping set with larger number of odd-degree CNs with multiple odd-degrees.  For $\gamma  \geq t$, a trapping set with $t$ VNs in the Tanner graph of an RC-constrained ($\gamma $,$\rho $)-regular LDPC code has at least $t(\gamma  - (t - 1))$ CNs of odd-degrees. For $t < \gamma $, the number of odd-degree CNs is greater than $\gamma $.  For $t = \gamma $, the number of odd-degree CNs is at least $\gamma $.  Based on Definition 2, if $t \leq \sqrt{n}$  and $\gamma  > t + 3$, there is no trapping set with size smaller than $\gamma  - 3$. If the ratio $\tau /\kappa$  requirement for a small trapping set is $\tau /\kappa  \leq  1$, then there is no trapping set with size smaller than $\gamma$.

   The above analysis shows that the structure, the sizes and the distribution of trapping sets of the Tanner graphs of RC-constrained LDPC codes very much depend on the column weights and orthogonal sets of rows of the parity-check matrices of the codes.  Basically, for a ($\gamma$,$\rho $)-regular LDPC code whose parity-check matrix has constant column weight $\gamma $, the RC-constraint on the rows and columns of the parity-check matrix ensures that: 1) the minimum weight of the code is at least $\gamma  + 1$; 2) the girth of the code's Tanner graph is at least 6; 3) there is no trapping set of size smaller than $\gamma  + 1$ with number of odd-degree CNs smaller than $\gamma $, (i.e., a trapping set with number of VNs less than $\gamma  + 1$ must have at least $\gamma$  CNs of odd-degrees connected to it); and 4) no trapping set of the type defined by Definition 2 with size smaller than $\gamma -3 $.  Due to these structural properties, RC-constrained ($\gamma$,$\rho$)-regular LDPC codes with large $\gamma$ in general have much lower error-floors than the unstructured LDPC codes constructed using computer-based method, and furthermore decoding of these codes with iterative message-passing algorithms converges very fast, as demonstrated by extensive simulation results given in [5]-[13], [15], [16], [69]-[71]. 

    Summarizing the above results, we have the following two theorems for trapping sets of an RC-constrained LDPC code.
 \begin{theorem} For an RC-constrained ($\gamma $,$\rho$)-regular LDPC code with $\gamma  > 1$, its Tanner graph contains no trapping set of size smaller than $\gamma  + 1$  for which the number of odd-degree CNs is smaller than $\gamma$.\end{theorem} 
 \begin{theorem} Let $\gamma$  be a positive integer such that $\gamma  > 3$. For an RC-constrained ($\gamma $,$\rho$)-regular LDPC code, its Tanner graph contains no trapping set of size smaller than $\gamma  - 3$ for which the number of odd-degree CNs is smaller than $4(\gamma -3)$.\end{theorem}
 The above results are derived based on only the RC-constraint on the rows and column of a parity-check matrix and its constant column weight $\gamma$  but not on its row weight.  Hence, the results apply to LDPC codes whose parity-check matrix has constant column weight but may have multiple row weights.  
    
       It is important to note that a trapping set induced by an error pattern \emph{does not necessarily prevent} decoding to converge unless the error pattern is uncorrectable to the decoder.  Only the trapping sets induced by the error patterns that are uncorrectable to the decoder may prevent decoding to converge (or fail) and cause an error-floor in the code's error performance.  For long codes, it is hard (or very much time consuming) to identify the configurations of those trapping sets which do trap the decoder and prevent decoding to converge.  However, extensive simulations in many published literatures did observe that in general, trapping sets of small size (relative to minimum weight $\omega_{\min}$ of the code) with small number of odd-degree CNs are the harmful ones.   When an error pattern induces such a small trapping set, the number of failed CNs is so small such that the messages generated by these failed CNs may not strong enough to overcome the messages coming from the satisfied CNs to make the changes of the erroneous VNs checked by the failed CNs to satisfy their check-sum constraints.  As a result, iteration continues.  However, for a trapping set with large number of odd-degree CNs compared to its number of VNs, the messages generated by the failed CNs would be strong enough to overcome the messages from the satisfied CNs to make appropriate changes of some code bits checked by all the CNs in such a way that all the check-sums are zero and decoding converges.  If a trapping set is induced by an undetectable error pattern, it is of the form ${\cal T}(\kappa,0)$, in which the $\kappa$  erroneous VNs form a codeword of weight $\kappa$. When this happens, the resultant syndrome of a hard-decision decoded sequence is zero.  In this case, decoding stops and the decoder converges to an incorrect codeword.  This results in an undetected error. If the minimum weight $\omega_{\min}$ of an LDPC code is small, trapping sets induced by uncorrectable error patterns that correspond to the minimum weight codewords may contribute significantly to the error-floor in the performance of the code.  Based on the above observation and reasoning, a code, in general, has a lower error-floor if it does not have small trapping sets (relative to the minimum weight) with small numbers of odd-degree CNs and its minimum weight  $\omega_{\min}$ is reasonable large.  If a code does not have trapping sets with size smaller than its minimum weight  $\omega_{\min}$, then the error-floor of the code is dominated by its minimum weight  $\omega_{\min}$, i.e., dominated by the trapping sets, ${\cal T}(\omega_{\min},0)$s, induced by the error patterns that are identical to the minimum weight codewords.  In the next two sections, we will show several classes of RC-constrained LDPC codes do have large minimum weights and do not have harmful trapping sets with size smaller than their minimum weights.  Hence, their error-floors are dominated by their minimum weights. 
 
   Since not all the trapping sets defined in Definitions 1 and 2 trap the decoder and prevent decoding to converge, the term ,trapping set, is actually misnamed.

\section{Trapping Sets of Cyclic FG-LDPC Codes and Their Cyclic and QC Descendants}
   Since cyclic FG-LDPC codes and their cyclic and QC descendants are RC-constrained LDPC codes, they have the trapping set structure presented in the last section.  In the following, we consider the trapping set structures of two special subclasses of cyclic FG-LDPC codes.  These two subclasses of FG codes have large minimum distances and no small trapping sets.  They can be decoded with various decoding algorithms ranging from hard-decision, reliability-based to pure soft-decision iterative decoding, such as the OSMLGD, the bit-flipping (BF), the weighted BF, the binary message-passing (SRBI-MLGD), the min-sum and the sum-product algorithms, to provide a wide range of effective trade-offs between error performance and decoding complexity.
   The first subclass of cyclic FG-LDPC codes is the class of cyclic EG-LDPC codes constructed based on the two-dimensional Euclidean geometries presented in Section IV.  Here, we consider the cyclic EG-LDPC code ${\cal C}_{EG}$ (or its QC equivalent ${\cal C}_{EG,qc}$) constructed based on the two-dimensional Euclidean geometry EG(2,$2^s$) over the field GF($2^s$).  The parity-check matrix ${\bf H}_{EG}$  of the code is a $(2^{2s} - 1)\times (2^{2s} - 1)$ circulant over GF(2) (or $(q+1)\times (q+1)$ array of $(q-1)\times (q-1)$ CPMs over GF(2)) whose rows are the incidence vectors of the lines in EG(2,$2^s$)  not passing through the origin of the geometry.  The column and row weights of this circular parity-check matrix ${\bf H}_{EG}$   are both $2^s$. Its rank is $3^s - 1$.  As shown in Section IV. B the null space of ${\bf H}_{EG}$  gives a ($4^s - 1$,$4^s - 3^s$) cyclic EG-LDPC code ${\cal C}_{EG}$ with minimum distance exactly $2^s + 1$.  With the OSMLGD, this code is capable of correcting $2^{s-1}$ or fewer random errors (or $2^s$ or fewer erasures). 
   
   Recall that the columns and rows of ${\bf H}_{EG}$ , as a $(2^{2s} - 1)\times (2^{2s} - 1)$ matrix over GF(2), correspond to the $2^{2s} - 1$ non-origin points and $2^{2s} - 1$ lines (not passing the origin) of EG(2,$2^s$), respectively.  The symbols  of a codeword ${\bf v} = (v_0, v_1, \ldots, v_{2^{2s} -2})$ in ${\cal C}_{EG}$  correspond to the $2^{2s} - 1$ non-origin points of EG(2,$2^s$) and therefore correspond to the columns of ${\bf H}_{EG}$ .  Since any two points in EG(2,$2^s$) are connected by a line, any two code symbols are checked by a row in ${\bf H}_{EG}$ .  Consequently, for any two error symbols, $e_{j_1}$ and $e_{j_2}$, in an error pattern $\bf e$, the two sets of syndrome-sums, $S^{j_1}$ and $S^{j_2}$, orthogonal on $e_{j_1}$ and $e_{j_2}$ have (exactly) one syndrome-sum in common.   
   
   It follows from the trapping set analysis given in the last section, any trapping set corresponding to an error pattern with  $2^s$ or fewer random errors will induce a subgraph of the Tanner graph of the code which contains at least $2^s$ CNs of odd degrees.  This is to say that code has no ($\kappa$,$\tau $) trapping set of size $\kappa$  smaller than $2^s + 1$ with the number $\tau$  of odd-degree CNs smaller than $2^s$.  This implies that for $\kappa  \leq  2^s$, there is no ($\kappa$ ,$\tau$) trapping set with the ratio $\tau /\kappa  \leq  1$.  If the (commonly used) requirements of small value of $\kappa$  and $\tau /\kappa  \leq  1$ are used to define a small trapping set, then the cyclic EG-LDPC code ${\cal C}_{EG}$   has no trapping set of size smaller than $2^s + 1$ (the minimum weight of the code).  For $\kappa  \leq  2^{s-1}$, the number $\tau$  of odd-degree CNs is greater than $ 2^{s-1}$.  Since the code is capable of correcting $2^{s-1}$ or fewer errors with the OSMLGD, all the trapping sets of size equal to or smaller than $2^{s-1}$ are un-harmful (i.e., they do not prevent decoding to converge or converge to an incorrect codeword) if the OSMLGD is performed before each new decoding iteration.   Since the length of the code is  $n = 4^s -1$, the square root of $n$, $\sqrt{n} = \sqrt{4^s - 1}\approx 2^s$.  For  $1 < \kappa  < 2^s - 3$, it follows from the tapping set analysis given in the last section that for a ($\kappa $,$ \tau $) trapping set, the number $\tau$  of CNs of odd-degree is at least $\kappa (2^s - (\kappa  - 1)) > 4\kappa $.  Then, it follows from Definition 2 that the cyclic EG-LDPC code ${\cal C}_{EG}$  has no small trapping set of the type defined by Definition 2 with size smaller than $2^s - 3$.  
 
   Summarizing all the results developed in the last and this sections, we have the following parameters for the cyclic EG-LDPC code ${\cal C}_{EG}$   constructed based on two-dimensional Euclidean geometry EG(2,$2^s$) over GF($2^s$) has the structure parameters: 1) length $4^s - 1$; 2) dimension $4^s - 3^s$; 3) minimum distance $2^s + 1$; and 4) no trapping set of size less than $2^s + 1$ or ($2^s - 3$) with number of odd-degree CNs less than $2^s$ (or less than $4\times 2^s$).  

 In fact, there are many trapping sets of size greater than $2^s$ with number of odd-degree CNs much greater than $2^s$.  As an example, we consider an error pattern $\bf e$ with $2^s + 1$ errors at the positions, $j_0, j_1, \ldots, j_{2^s-1}, j_{2^s}$.  Suppose the errors positions $j_0, j_1,\ldots , j_{2^s-1}$, correspond to the $2^s$ points $\alpha ^{j_0}$, $\alpha ^{j_1}$, $\ldots$, $\alpha ^{j_{2^s-1}}$ of a line $\cal L$ in EG(2,$2^s$) not passing through the origin.  The position $j_{2^s}$ is any other arbitrary position and it corresponds to the point $\alpha ^{j_{2^s}}$.  In this case, there is a single syndrome-sum contains $2^s$ errors at the positions $j_0, j_1, \ldots, j_{2^s-1}$, and this syndrome-sum equals zero.  Since in a finite geometry, any two points are connected by a line.  Then any error at the position in the set $\{ j_0, j_1, \ldots, j_{2^s-1} \}$ and the error at the position $j_{2^s}$ are contained in at most one syndrome-sum and they are the only errors in sum. (Note that the rows of the parity-check matrix ${\bf H}_{EG}$ correspond only to the lines not passing through the origin.)  Consequently, this syndrome-sum is equal to zero.  Recall that each position between 0 and $n - 1$ is checked by $2^s$ rows of ${\bf H}_{EG} $.  Therefore, for each position $j_i$, $0 \leq  i < 2^s$, there are at least $2^s - 2$ syndrome-sums contain only one error in the error pattern $\bf e$.  As a result, the trapping set induced by the error pattern $\bf e$ consists of $2^s + 1$ VNs, at least $2^s(2^s - 2)$ CNs of degree-1, at most $2^s$ CNs of degree-2 and one CN of degree $2^s$. If $s \geq 3$, the number of degree-1 CNs is much larger than the number of VNs in the trapping set. This error pattern is correctable with the OSMLGD. 
 
 Using the geometric structures, configurations of some trapping sets of an EG-LDPC code may be analyzed. Consider another case.  Let $\bf e$ be an error pattern with $2^s +2$ errors at the positions, $j_0, j_1, \cdots, j_{2^s-1}, j_{2^s}, j_{2^s+1}$.  Suppose the errors positions $j_0, j_1, \ldots, j_{2^s-1}$, correspond to the $2^s$ points $\alpha ^{j_0}, \alpha ^{j_1}, \ldots, \alpha ^{j_{2^s-1}}$ of a line $\cal L$ in EG(2,$2^s$) not passing through the origin.  The positions $j_{2^s}$ and $j_{2^s+1}$ are two arbitrary positions which correspond to the points $\alpha ^{j_{2^s}}$ and $\alpha ^{j_{2^s+1}}$.  Assume that $\alpha ^{j_{2^s}}$ and $\alpha ^{j_{2^s+1}}$ are not on the same line.  Then each point on $\cal L$ may pair with either point $\alpha ^{j_{2^s}}$ or $\alpha ^{j_{2^s+1}}$ appearing on a line.  Based on this, we can readily see that the trapping set induced by the error pattern $\bf e$ has at least $2^s(2^s- 3)$ CNs of degree-1, at most $2^{s+1}$ CNs of degree-2 and one CN with degree $2^s$.  If points $\alpha ^{j_{2^s}}$ and $\alpha ^{j_{2^s+1}}$ lie on the same line, then the number of CNs with degree-1 is at least $2^s(2^s - 2)$.  In either case, for $s \geq 3$, the number of odd-degree CNs is much greater than the number of VNs of the trapping set. 
 
   Now we consider a more general case. For $0 \leq  r < 2^s - 2$, consider an error pattern $\bf e$ with $2^s + r$ errors positions at the positions, $j_0, j_1, \ldots, j_{2^s-1}, j_{2^s}, \ldots,j_{2^s+r-1}$.  Again, we assume that the $2^s$ positions $j_0, j_1, \ldots, j_{2^s-1}$, correspond to the $2^s$ points $\alpha ^{j_0}, \alpha ^{j_1}, \ldots, \alpha ^{j_{2^s-1}}$ of a line $\cal L$ in EG(2,$2^s$) not passing through the origin. Following the same analysis given above, we can easily show that the trapping set induced by this error pattern with $2^s + r$ errors consists of at least $2^s(2^s - r)$ CNs of degree-1 and at most $r 2^s$ CN's of degree-2. Since $r < 2^s - 2$, the number of degree-1 CNs is much larger than the number of VNs.  For the case $r = 0$, the trapping set induced by the error pattern $\bf e$ whose error locations corresponding to the $2^s$ points of a line not passing through the origin of the geometry has exactly $2^s(2^s -1)$ CNs of degree-1 and one CN of degree $2^s$.  Since there are $2^{2s} - 1$ lines not passing through the origin (the rows of the parity-check matrix ${\bf H}_{EG}$ are the incidence vectors of these lines), there are $2^{2s} - 1$ such trapping sets of size $2^s$.  For such a trapping set, the number of CNs of degree-1 is $2^s - 1$ times larger than the number of VNs.  Error patterns corresponding to these trapping sets are correctable with the OSMLGD.

\begin{example} 
 Consider the (63,37) cyclic EG-LDPC code constructed based on the two-dimensional Euclidean geometry EG(2,$2^3$) over GF($2^3$).  The parity-check matrix of this code is a $63\time 63$ circulant over GF(2) with both column and row weights 8.  The minimum weight of this code is 9.  The code is capable of correcting 4 or fewer errors with OSMLGD.  By computer search, we have found all the trapping sets induced by error patterns with 3 up to 22 errors.  Table 1 gives a partial list of the found trapping sets.  From the Table 1, we see that for $\kappa  < 9$, the number $\tau$  of odd-degree CNs associated to every trapping set is greater than $\kappa$ .  For $\kappa  = 9$, there are (9,0) trapping sets which correspond to minimum weight codewords of the code.  The square root $\sqrt{63} \approx 8$. From Table 1, we see that for $\kappa  < 6$, the number $\tau$  of odd-degree CNs associated with each trapping set of size $\kappa$ smaller than 6 is greater than $4\kappa $.  Then, it follows from Definition 2 that the Tanner graph of the code does not contain small trapping set with size $\kappa  < 6$ of the type defined by Definition 2. In decoding of the (63,37) cyclic EG-LDPC code using 50 iterations of the SPA, none of the trapping sets with size smaller than 9 prevents decoding to converge (or trap the decoder) and the error patterns corresponding to these trapping sets are all correctable.  The trapping sets ${\cal T}(9,0)$, ${\cal T}(10,0)$, ${\cal T}(11,0)$, ${\cal T}(12,0)$ and ${\cal T}(14,0)$ result in undetected error (incorrect decoding). The error performance of the (63,37) cyclic EG-LDPC code is shown in Figure 11. 
 
   Suppose we consider the $(255,175)$ cyclic EG-LDPC code constructed based on the two-dimensional Euclidean geometry EG(2,$2^4$) over GF($2^4$).  This code has minimum weight 17.  Extensive computer search found no trapping set of size smaller than 17 which prevents decoding to converge or cause decoding failure.  We found some large trapping sets with very large numbers of odd-degree CNs but are not harmful.  These trapping sets are: ${\cal T}(16,102)$, ${\cal T}(18,110)$, ${\cal T}(21,102)$, ${\cal T}(30,120)$, ${\cal T}(29,120)$, and ${\cal T}(33,130)$.  All but ${\cal T}(30,120)$ have $\tau  > 4\kappa$.  Therefore, only the trapping set ${\cal T}(30,120)$ is a small trapping set by Definition 2.  \twotriangle\end{example}

\begin{example} Consider the (4095,3367) Cyclic-EG-LDPC code with minimum weight 65 constructed based on the 2-dimensional EG(2,$2^6$) over GF($2^6$) given in Example 2.  The parity-check matrix of this code has column weight 64.  The Tanner graph of this code has no trapping set of size smaller than 64 with number of odd-degree CNs smaller than 64. Note that $\sqrt{4095} \approx 64$.  It follows from Definition 2 that the code has no trapping set with size smaller than 61. As shown in figure 1, decoding of this code with either the SPA or the SMA converges very fast. Consider the (1365,765) cyclic descendant of the (4095,3367) cyclic EG-LDPC code given in Example 3.  The parity-check matrix of this code is a $1365\times 1365$ circulant with both column and row weights 16.  For this code, any trapping set of size smaller than 17 has at least 16 odd-degree CNs associated with it. Note that $\sqrt{1365}>17$. Based on Definition 2, it has no trapping set with size smaller than 13.  \twotriangle\end{example}

   Next, we consider the trapping set structure of a cyclic PG-LDPC code ${\cal C}_{PG}$ constructed based on the 2-dimensional projective geometry PG(2,$q$) over GF($q$) with $q=2^s$.  The parity-check matrix of this code is a $(q^2 + q + 1)\times (q^2 + q + 1)$ circulant ${\bf H}_{PG}$  over GF(2) with both column and row weights equal to $q + 1$. The null space of ${\bf H}_{PG}$ gives an RC-constrained cyclic PG-LDPC code ${\cal C}_{PG}$ of length $n = q^2 + q +1$ and minimum weight at least $q + 2$, whose Tanner graph has a girth of at least 6.  Since the ${\bf H}_{PG}$  satisfies the RC-constraint and its column weight is $q +1$, it follows from the analysis given in Section VII that ${\cal C}_{PG}$ has no trapping set ${\cal T}(\kappa ,\tau )$ for which both the size $\kappa$  and the number $\tau$  of odd-degree CNs smaller than $q  + 1$.  The square root of the length of the code is $\sqrt{n}\approx q$.  For $\kappa  < q -2$, it easy to check that the number $\tau$  of odd-degree CNs of a trapping set ${\cal T}(\kappa ,\tau )$ is greater than $4\kappa $.  Then, ${\cal C}_{PG}$ has no trapping set of the type defined by Definition 2 with size smaller than $q - 2$.  The results on trapping sets of the cyclic PG-LDPC code are exactly the same obtained in [20] derived in  a different approach.  Our derivation of the results are simply based on the RC-constraint on the parity-check matrix which is much simpler and less mathematical.
 
   For $q = 2^s$, the cyclic PG-LDPC code ${\cal C}_{PG}$ has the following structural parameters: 1) length $n = 2^2s + 2^s + 1$; 2) dimension $n - 3^s - 1$; 3) minimum weight at least $2^s + 2$; 4) girth at least 6; 5) no trapping set of size less than $2^s +2$ with number of odd-degree CNs less than $2^s+1$; and 6) no trapping set of the type defined by Definition 2 with size less than $2^s - 2$.

\section{Other RC-Constrained LDPC Codes and their Trapping Sets}
 
   Besides EG- and PG-LDPC codes, there are other classes of structured RC-constrained LDPC codes.  These classes of codes are either constructed based on finite fields [11]-[13], [69]-[71] or experimental designs [72]-[76].  Codes in most of these classes are QC-LDPC codes.  Since the parity-check matrices of the codes in these classes satisfy the RC-constraint, their trapping sets have the structure as described in Section VII. B.  The constructions based on finite fields given in [11]-[13], [69]-[71] are of the same nature and they give several large classes of RC-constrained QC-LDPC codes.  Among them, several subclasses have large minimum weights.  In this section, we choose the first class of QC-LDPC codes given in [11] for illustration of their trapping set structure.
 
   Consider the first construction of QC-LDPC codes given in [11].  Let $\alpha$ be a primitive element of the Galois field GF($q$) Then, $\alpha ^{-\infty}  = 0$, $\alpha ^0 = 1,  \alpha , \cdots   , \alpha ^{q-2}$ give all the elements of GF($q$).  Let ${\cal C}_{rs}$   be the cyclic $(q-1, 2, q-2)$ RS code over GF($q$) with two information symbols whose generator polynomial ${\bf g}(X)$ has $\alpha , \alpha ^2, \cdots  , \alpha ^{q-3}$ as roots.  Then, for $0 \leq  i < q-1$, the two $(q-1)$-tuples over GF($q$),
 \[
                {\bf u}_i = (\alpha ^i, \alpha ^{i+1}, \cdots  , \alpha ^0, \alpha ^{q-2}, \cdots  , \alpha ^{i-1}),
 \]
and
 \[
               {\bf  v}_i = (\alpha ^i, \alpha ^i, \cdots  , \alpha ^i),
\]
are two nonzero codewords in ${\cal C}_{rs}$   with weight $q-1$.  Note that ${\bf u}_1, \cdots  , {\bf u}_{q-2}$ are cyclic-shifts of ${\bf u}_0$.    For $i = 0$, ${\bf v}_0 = (1, 1, \cdots  , 1)$.  The subscript ``$rs$'' of ${\cal C}_{rs}$   stands for ``Reed-Solomon''.
 
   For  $0 \leq  i < q-1$, ${\bf u}_i-{\bf v}_0$ is a codeword in ${\cal C}_{rs}$   with weight $q-2$ (minimum weight). Form the following $(q-1)\times (q-1)$ matrix over GF($q$) with ${\bf u}_0-{\bf v}_0$, ${\bf u}_1-{\bf v}_0, \cdots  , {\bf u}_{q-2}-{\bf v}_0$ as rows:
 \begin{equation} 
 {\bf W}_{rs} = \left[  \begin{array}{l}
{\bf w}_0\\
{\bf w}_1\\
\vdots\\
{\bf w}_{q-2}
\end{array}  \right]    
= \left[  \begin{array}{llll}
 \alpha ^0 - 1      &    \alpha  - 1    &   \cdots   &      \alpha ^{q-2} - 1\\
  \alpha ^{q-2} - 1   &  \alpha ^0 - 1  &    \cdots  &        \alpha ^{q-3} - 1\\
\vdots& & \ddots & \vdots\\
 \alpha  - 1    &      \alpha ^2 - 1  &     \cdots     &     \alpha ^0 - 1
\end{array}  \right].       
\end{equation}
This matrix ${\bf W}_{rs}$    is the matrix (with rows permuted) given by Eq. (4) in [11] for the construction of the first class of QC-LDPC codes.  Every row (or column) of ${\bf W}_{rs}$    consists $q-2$ distinct nonzero elements and one 0-element of GF($q$). The $q-1$ zero entries of ${\bf W}_{rs}$    lie on its main diagonal. Therefore, both column and row weights of ${\bf W}_{rs}$    are $q-2$. This matrix satisfies the following constraint on the Hamming distance between two rows [11]: for $0\leq  i, j < q-1$, $i \neq  j$ and $0 \leq  c, l < q-1$, the Hamming distance between the two $(q-1)$-tuples over GF($q$), $\alpha ^c {\bf w}_i$ and $\alpha ^l {\bf w}_j$, is at least $q-2$, (i.e., $\alpha ^c {\bf w}_i$ and $\alpha ^l {\bf w}_j$ differ in at least $q-2$ places).  This constraint on the rows of matrix ${\bf W}_{rs}$ is called the \emph{row-distance (RD)-constraint} and ${\bf W}_{rs}$ is called an RD-constrained matrix.

   Let $\bf P$ be a $(q - 1)\times (q - 1)$ CPM whose top row is given by the $(q - 1)$-tuple $(0 1 0 \cdots  0)$ over GF(2) where the components are labeled from 0 to $q - 2$ and the single 1-component is located at the 1st position.  Then $\bf P$ consists of the $(q - 1)$-tuple $(0 1 0 \cdots  0)$ and its $q - 2$ right cyclic shifts as rows.  For $1 \leq  i < q$, let ${\bf P}^i = {\bf P}\times {\bf P}\times  \cdots  \times {\bf P}$ be the product of $\bf P$ with itself $i$ times, called the $i$th power of $\bf P$.  Then, ${\bf P}^i$ is also a $(q - 1)\times (q - 1)$ CPM whose top row has a single 1-component at the $i$th position.  For $i = q -1$, ${\bf P}^{q -1} = {\bf I}_{q-1}$, the $(q - 1)\times (q - 1)$ identity matrix.  Let ${\bf P}^0 = {\bf P}^{q-1} = {\bf I}_{q -1}$.  Then the set ${\cal P} = \{{\bf P}^0, {\bf P}, {\bf P}^2, \cdots  , {\bf P}^{q -2}\}$ of CPMs forms a cyclic group of order $q-1$ under matrix multiplication over GF(2) with ${\bf P}^{q -1-i}$ as the multiplicative inverse of ${\bf P}^i$ and  ${\bf P}^0$ as the  identity element.

  For $0 \leq  i < q -1$, we represent the nonzero element $\alpha ^i$ of GF($q$) by the $(q - 1)\times (q - 1)$ CPM ${\bf P}^i$. This matrix representation is referred to as the $(q - 1)$-fold binary \emph{matrix dispersion} (or simply binary matrix dispersion) of $\alpha ^i$.  Since there are $q - 1$ nonzero elements in GF($q$) and there are exactly $q - 1$ different CPMs over GF(2) of size $(q - 1)\times (q - 1)$, there is a one-to-one correspondence between a nonzero element of GF($q$) and a CPM of size $(q - 1)\times (q - 1)$.  Therefore, each nonzero element of GF($q$) is uniquely represented by a CPM of size $(q - 1)\times (q - 1)$.  For a nonzero element $\delta$  in GF($q$), we use the notation $ {\bf B}{(\delta)} $ to denote its binary matrix dispersion.  If $\delta  = \alpha ^i$, then $ {\bf B}{(\delta)}  = {\bf P}^i$.   For the 0-element of GF($q$), its binary matrix dispersion is defined as the $(q - 1)\times (q - 1)$ ZM, denote ${\bf P}^{-\infty }$. 

      Dispersing each nonzero entry of ${\bf W}_{rs}$ into a $(q - 1)\times (q - 1)$ CPM over GF(2) and each 0-entry into a $(q - 1)\times (q - 1)$ ZM, we obtain the following $(q-1)\times (q-1)$ array of CPMs and/or ZMs over GF(2) of size $(q - 1)\times (q - 1)$:
\begin{equation}        
      {\bf H}_{rs} =    \left[ \begin{array}{cccc}            {\bf B} _0     &    {\bf B} _1  &   \cdots   &     {\bf B} _{q-2}\\
                      {\bf B} _{q-2} &       {\bf B} _0 &      \cdots    &    {\bf B}_{q-3}   \\
                          \vdots  &        &     \ddots  &\vdots\\
                       {\bf B} _1 &        {\bf B} _2&       \cdots      &    {\bf B}_0
\end{array}
\right],
\end{equation}
where ${\bf B}_j={\bf B} (\alpha^j -1)$ for $0\leq j<q-1$. ${\bf H}_{rs}$  is called the binary $(q - 1)$-fold \emph{array dispersion} of ${\bf W}_{rs}$    (or simply binary array dispersion of ${\bf W}_{rs}   $). This array has $(q-1)$ ZMs which lie on its main diagonal.  It is a $(q - 1)^2\times (q - 1)^2$ matrix over GF(2) with both column and row weights equal to $q-2$.  Based on the RD-constraint on the rows of ${\bf W}_{rs}$    and the binary CPM matrix dispersions of the entries of ${\bf W}_{rs}   $, it was proved in [10], [11], [69], [71] that ${\bf H}_{rs} $, as a $(q - 1)^2\times (q - 1)^2$ matrix over GF(2), satisfies the RC-constraint.  Hence, its associated Tanner graph has a girth of at least 6. The RD-constrained matrix ${\bf W}_{rs}$ used for constructing the RC-constrained array ${\bf H}_{rs}$ of CPMs is called the base matrix for array dispersion.

  For any pair ($\gamma $,$\rho $) of integers $\gamma$  and $\rho$  with $1\leq  \gamma $, $\rho  < q $, let ${\bf H}_{rs}(\gamma ,\rho )$ be a $\gamma \times \rho$  subarray of ${\bf H}_{rs} $.  ${\bf H}_{rs} (\gamma ,\rho )$ is a $\gamma (q-1)\times \rho (q-1)$ matrix over GF(2) which also satisfies the RC-constraint.  The null space of ${\bf H}_{rs} (\gamma ,\rho )$ gives a QC-LDPC code ${\cal C}_{rs,qc}$ of length $\rho (q-1)$ with rate at least $(\rho -\gamma )/\rho $, whose Tanner graph has a girth of at least 6.  If ${\bf H}_{rs} (\gamma ,\rho )$ does not contain any of the ZMs of ${\bf H}_{rs} $, then ${\bf H}_{rs}$  has constant column weight $\gamma$  and constant row weight $\rho $.  In this case, ${\cal C}_{rs,qc}$ is a $(\gamma ,\rho )$-regular QC-LDPC code. If ${\bf H}_{rs} (\gamma ,\rho )$ contains ZM(s) of ${\bf H}_{rs} $, it has two different column weights, $\gamma - 1$ and $\gamma $, and/or two different row weights, $\rho  - 1$ and $\rho $.  In this case, the null space of ${\bf H}_{rs} (\gamma ,\rho )$ gives a near-regular binary QC-LDPC code.  

   For a given finite field GF($q$), the above construction gives a family of structurally compatible RC-constrained QC-LDPC codes.  Consequently, the construction gives a large class of binary QC-LDPC codes. Since their parity-check matrices satisfy the RC-constraint, they have the same trapping set structure presented in VII. B.

   A very special case is the QC-LDPC code ${\cal C}_{rs,qc,f}$ given by the null space of the full array ${\bf H}_{rs}$  with $q = 2^s$.  For this case, the length of the code is $n = (2^s-1)^2$ and its minimum weight is at least $2^s-1$.  Using the technique presented in [13], we find that the rank of ${\bf H}_{rs}$  is
\begin{equation}
                             rank({\bf H}_{rs} ) = 3^s-3.  
\end{equation}
(The derivation of the expression of (73) is given in a separate paper.)  Since the column weight of ${\bf H}_{rs}$  is $2^s-2$, it follows from the analysis of trapping set structure of an RC-constrained LDPC code given in VII that for $\kappa  \leq  2^s-2$, ${\cal C}_{rs,qc,f}$ has no trapping set ${\cal T}(\kappa ,\tau )$ of size $\kappa$  with number of odd-degree smaller than $2^s-2$. Note that $\sqrt{n} = 2^s-1$. Then, for $\kappa  < 2^s-5$, there is no trapping set ${\cal T}(\kappa ,\tau )$ of the type defined by Definition 2 with number of odd-degree CNs smaller than $4\kappa $.  That is to say that there is no trapping set with size smaller than $2^s-5$.  

Summarizing the above results, the QC-LDPC code ${\cal C}_{rs,qc,f}$ given by the full array ${\bf H}_{rs}$  of (72) for $q = 2^s$ has the following parameters: 1) length $n = (2^s-1)^2$; 2) dimension $(2^s-1)^2-3^s + 3$; 3) minimum weight at least $2^s-1$; 4) any trapping set ${\cal T}(\kappa ,\tau )$ with $\kappa  \leq  2^s -2$ must have more than $2^s-2$ CNs of odd-degrees; and 5) no trapping sets of the type defined by definition 2 with size smaller than $2^s-5$.

\begin{example} Let GF($2^5$) be the field for code construction.  Based on this field, we can construct a $31\times 31$ array ${\bf H}_{rs}$  of CPMs and ZMs of size $31\times 31$.  ${\bf H}_{rs}$  is a $961\times 961$ matrix over GF(2) with both column and row weights 30. The null space of ${\bf H}_{rs}$  gives a (30,30)-regular (961,721) QC-LDPC code with minimum distance at least 31.  This code is the code given in Example 1 of [11].  For this code, any trapping set ${\cal T}(\kappa ,\tau )$ with $\kappa  < 30$ must have more than 30 CNs of odd-degrees.  The code has no trapping sets of the type defined by Definition 2 with size smaller than 27.  None of the trapping sets with size smaller than 31 traps the decoder. The error performances of this code with 5, 10 and 50 iterations are shown in Figure 12. 
\twotriangle \end{example}

Besides the class of RD-constrained base matrices given above, several other classes of RD-constrained base matrices for constructing RC-constrained arrays of CPMs have been proposed in [10]-[13].  Based on these arrays of CPMs, several large classes of RC-constrained QC-LDPC codes have been constructed.  Codes in these classes perform well with iterative decoding using either the SPA or MSA.  In the following, we describe another method for constructing a large class of RD-constrained base matrices for array dispersions to construct RC-constrained QC-LDPC codes.  This method is based on a class of Latin squares over finite fields and is proposed in [13].

   An array is called a Latin square of order $n$ if each row and each column contains every element of a set of $n$ elements exactly once [77].  Latin squares can be constructed from finite fields. Consider the field GF($q$).  Let $\alpha$  be a primitive element of GF($q$)  and $\eta$  be nay nonzero element of GF($q$).  Form the following $q\times q$ matrix over GF($q$):
                     \begin{equation}
{\bf W}_{LS}  =\left[\begin{array}{ccccc}
                                  \alpha ^0\eta  - \alpha ^0       &  \alpha ^0\eta  - \alpha  &  \ldots  &       \alpha ^0\eta  - \alpha ^{q-2}   &     \alpha ^0\eta  - \alpha ^{-\infty}  \\
                                    \alpha \eta  - \alpha ^0     &         \alpha \eta  - \alpha   & \ldots &        \alpha \eta  - \alpha ^{q-2}    &       \alpha \eta  - \alpha ^{-\infty}\\
\vdots  & & \ddots & &\vdots\\ 
                              \alpha ^{q-2}\eta  - \alpha ^0   &  \alpha ^{q-2}\eta  - \alpha   &   \ldots & \alpha ^{q-2}\eta  - \alpha ^{q-2}   &  \alpha ^{q-2}\eta  - \alpha ^{-\infty} \\
                                 \alpha ^{-\infty} \eta  - \alpha ^0   &     \alpha ^{-\infty} \eta  - \alpha    &\ldots &      \alpha ^{-\infty} \eta  - \alpha ^{q-2}    &   \alpha ^{-\infty} \eta  - \alpha ^{-\infty} 
\end{array}\right].
\end{equation}
Then, ${\bf W}_{LS}$ is a Latin square of order $q$ over GF($q$)  .  Every element of GF($q$)  appears in a row and a column once and only once. In [13], it was proved that ${\bf W}_{LS}$ satisfies the RD-constraint.  Binary array dispersion of ${\bf W}_{LS}$ gives a $q\times q$ array ${\bf H}_{LS}$ of CPMs and ZMs of size $(q-1)\times (q-1)$.  Each row or column of ${\bf H}_{LS}$ contains one only one ZM.  ${\bf H}_{LS}$ is a $q(q-1)\times q(q-1)$ matrix over GF(2) with both column and row weights $q-1$.

   For any pair ($\gamma$,$\rho $) of positive integers with $1 \leq  \gamma , \rho  < q$, let ${\bf H}_{LS}(\gamma ,\rho )$ be a $\gamma \times \rho$  subarray of ${\bf H}_{LS}$.  ${\bf H}_{LS}(\gamma ,\rho )$ is a $\gamma (q-1) \times \rho (q-1)$ matrix matrix over GF(2).  If ${\bf H}_{LS}(\gamma ,\rho )$ does not contain any ZM of ${\bf H}_{LS}$, then ${\bf H}_{LS}(\gamma ,\rho )$, as a $\gamma (q-1)\times \rho (q-1)$ matrix, has column and row weights $\gamma$  and $\rho $, respectively.  The null space of ${\bf H}_{LS}(\gamma ,\rho )$ gives an RC-constrained ($\gamma $,$\rho $)-regular QC-LDPC code ${\cal C}_{LS,qc}$ of length $\rho (q-1)$. The code has the trapping set structure as described in Section VII.

   For $q = 2^s$, the QC-LDPC code ${\cal C}_{LS,qc,f}$ given by the null space of the full array ${\bf H}_{LS}$ has the following parameters [13]:

    \begin{center}
    
 \begin{minipage}{7cm}
             Length: $n  = 2^s(2^s - 1)$,
 
                    Number of parity symbols:  $n-k = 3^s-1$,
 
             Minimum distance $d_{\min} \geq  2^s + 2$.
 \end{minipage}
 
 \end{center}

   It follows from the trapping set analysis given in Section VII, any trapping set corresponding to an error pattern with  $2^s - 1$ or fewer random errors will induce a subgraph of the Tanner graph of the code which contains at least $2^s - 1$ CNs of odd degrees.  If the requirements of small value of $\kappa$  and $\tau /\kappa  \leq  1$ are used to define a small trapping set, then the QC-LDPC code ${\cal C}_{LS,qc,f}$  has no trapping set of size smaller than $2^s - 1$.
 
   Since the length of the code is $n = 2^s(2^s -1)$, the square root of $n$, $\sqrt{n} \approx  2^s$.  For $1 < \kappa  < 2^s - 4$, the number $\tau$  of CNs of odd-degrees is at least $\kappa (2^s - 1 - (\kappa  - 1)) > 4\kappa $.  Then, it follows from Definition 2 that the QC-LDPC code ${\cal C}_{LS,qc,f}$ has no trapping set of the type defined by Definition 2 with size smaller than $2^s - 4$.
 
 \begin{example}
 The code constructed based on the Latin square of order 32 over GF($2^5$) is an RC-constrained  (992,750) QC-LDPC code with minimum weight at least 34.  Extensive computer search found no trapping sets with size smaller than 34 that trap the decoder.  Two trapping sets ${\cal T}(36,0)$ are found.  This says that the minimum weight of the code is 36.  Also found are 1595 ${\cal T}(40,0)$ trapping sets. Since there are no harmful trapping sets with sizes smaller than the minimum weight, the error-floor of the code is dominated by the minimum weight of the code which is 36. The error performances of this code over the AWGN channel with 50 iterations of the SPA and the MSA are shown in Figure 13.  We see that there is no visible error-floor down to the BER of $10^{-11}$.  At the BLER of $10^{-9}$ (decoded with a min-sum FPGA decoder), the code performs 1.1 dB from the sphere packing bound. \twotriangle \end{example}

\section{Conclusion  and Remarks}

    In this paper, we have shown that cyclic and quasi-cyclic descendant codes can be derived from a known cyclic code through decomposition of its parity-check matrix in circulant form using column and row permutations.  We have analyzed some structural properties of descendant cyclic codes of a cyclic code, particularly in characterization of the roots of their generator polynomials.  By decomposition of cyclic finite geometry LDPC codes, we are able to enlarge the repertoire of  cyclic finite geometry LDPC codes and to construct new quasi-cyclic LDPC codes.  The cyclic and quasi-cyclic structures allow the implementation of encoding of LDPC codes with simple shift registers with linear complexity.  These structures also simplify the hardware implementation of LDPC decoders.  Quasi-cyclic structure simplifies wire routing of an LDPC decoder and allows partial parallel decoding that offers a trade-off between decoding complexity and decoding delay.  We have shown that a cyclic LDPC code can be put in quasi-cyclic form through column and row permutations and vice versa.  In encoding, we use its cyclic form and in decoding, we use its quasi-cyclic form. This allows us to have both advantages in encoding and decoding implementations.   In this paper, we have also analyzed the trapping set structure of LDPC codes whose parity-check matrices satisfy the RC-constraint.   We have shown that several classes of finite geometry and finite field LDPC codes don't have trapping sets with sizes smaller than the minimum weights of the codes.  The codes in these classes have large minimum weights.  Consequently, codes in these classes have very low error-floors which are pertinent to some communication and storage systems where very low error-rates are required.
    
Finally, we would like to point out that there are two large classes of structured LDPC codes [8], [69] which satisfy the RC-constraint but are not quasi cyclic.  The class of LDPC codes given in [8] was constructed based on finite geometry decomposition and the class of LDPC codes given in [69] was constructed based on Reed-Solomon codes with two information symbols. These two classes of codes have large minimum distances.  It follows from our trapping set analysis, they don't have trapping sets of sizes smaller than their minimum distances.

\newpage

\begin{table}
\caption{A partial list of trapping sets of the (63,37) cyclic EG-LDPC codes}\centering
\begin{tabular}{c||c||c||c}
\hline
                                                     Size                     &          Number of odd-degree CNs & Size                     &          Number of odd-degree CNs\\
                                                        $ \kappa $            &                                         $ \tau $ &      $ \kappa $            &                                         $ \tau $\\
\hline\hline
                                                         3                   &                                  18 & 10 & 0\\
                                                                                 &                               20 &  & 14\\
                                                                                 &	22 & &\\
\hline
                                                        4                       &                               20 & 11 & 0\\
                                                                     &                                     22 &&\\
   &                                  24 &&\\
   &                                  26 &&\\
   &		28&&\\
\hline
                                                       5                   &                                    22 &12 & 0\\
                                                       &	24&&\\
                                                                             &                                   26&&\\
                                                                            &                                    28&&\\
                                                                          &                                       30&&\\
\hline
                                                      6              &                                          22 & 13 &26\\
                                                      						&		24&&\\
                                                                        &                                        26&&\\
                                                                          &                                      28&&\\
                                                          &30&&\\
                                                                            &                                    32&&\\
\hline                                                     7               &                                          18& 14 & 0\\
                                                                                  &                              22&&\\
                                                                                  &24&&\\
                                                                              &                                  26&&\\
                                                                              &28&&\\
                                                                              &30&&\\
                               &      32&&\\
                                 &    34&&\\
                                   &  36&&\\
\hline
                                                    8                            &                              26 & 22 &32\\
                      &               30&&\\
\hline
                                                    9                          &                                 0&&\\
                                                                    &                                         26&&\\
\hline
\end{tabular}
\end{table}

  \begin{figure}
 \centering
 \includegraphics[width = 3.3in]{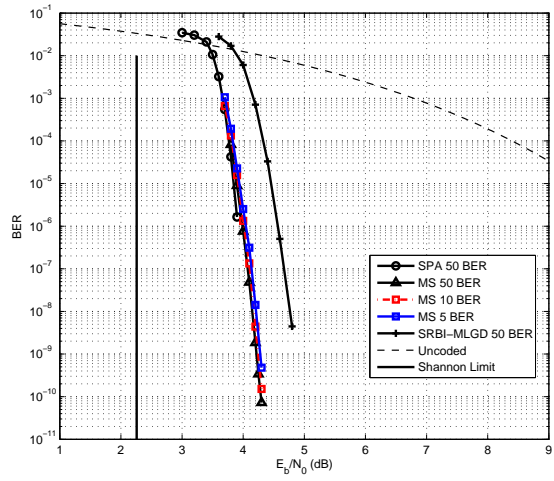}
 \caption{Bit error performances of the binary (4095,3367) cyclic EG-LDPC code given in Example 1 decoded with the SPA and the scaled MSA.}\label{fig:3367}
 \end{figure}

        \begin{figure}
 \centering
  \includegraphics[width = 3.3in]{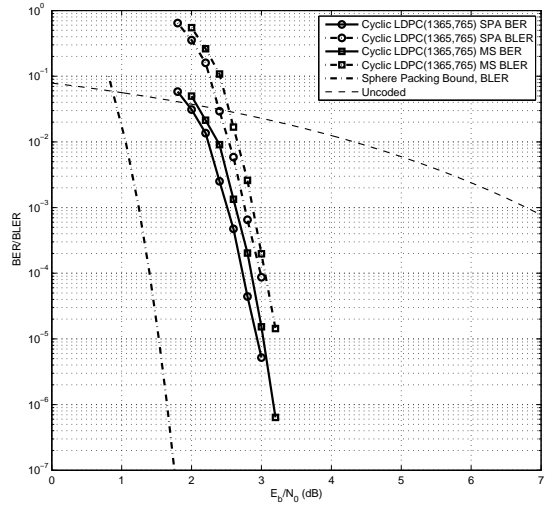}
 \caption{(a) The error performances of the binary (1365,765) cyclic EG-LDPC code given in Example 3 decoded with 50 iterations of the SPA and the MSA.}\label{fig:1365}
 \end{figure}
 
   \setcounter{figure}{1}
        \begin{figure}
 \centering
  \includegraphics[width = 3.3in]{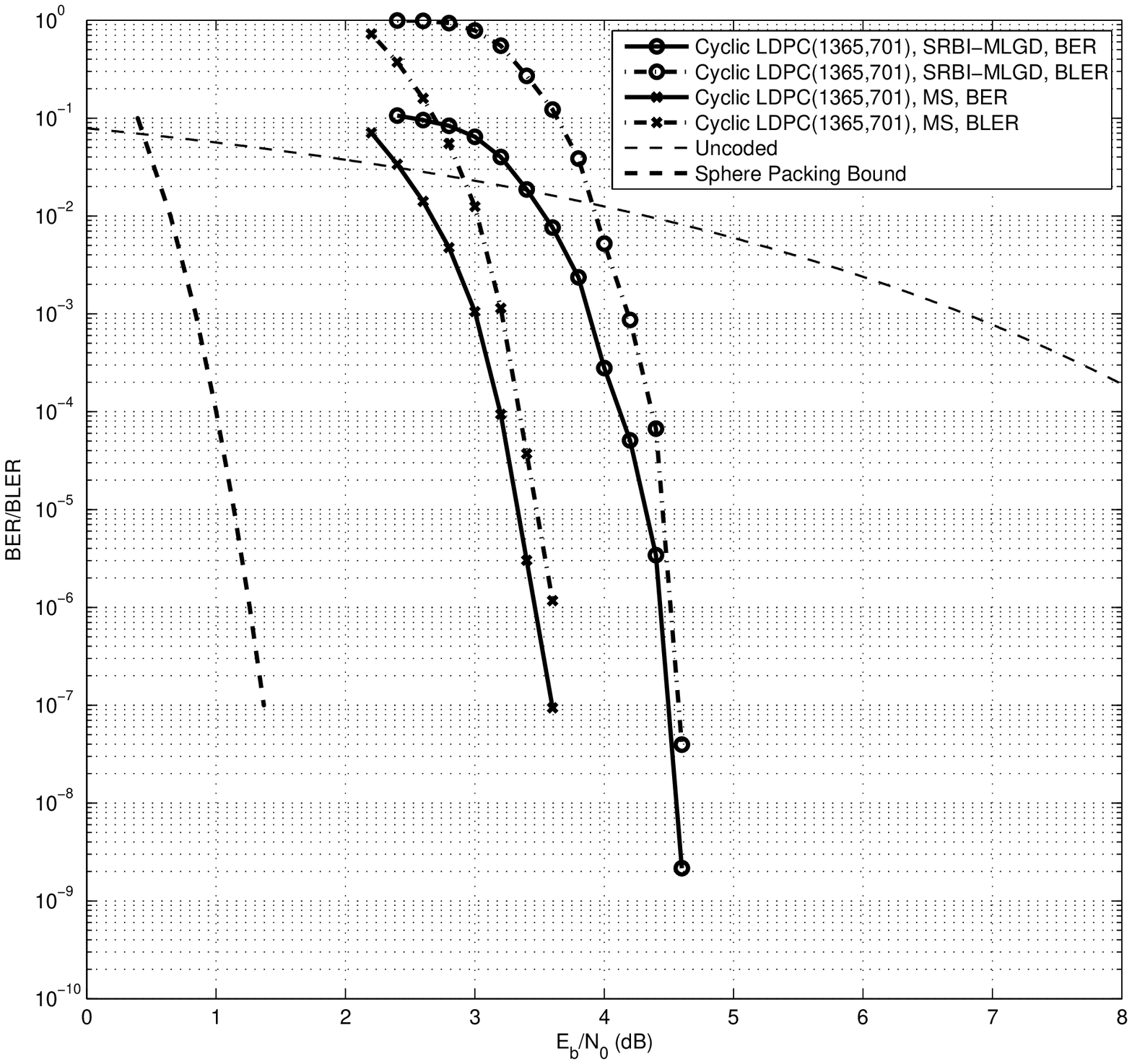}
 \caption{(b) The error performances of the binary (1365,701) cyclic EG-LDPC code given in Example 3 decoded with the MSA and the SRBI-MLGD-algorithm.}\label{fig:701}
 \end{figure}
 
    \setcounter{figure}{1}
        \begin{figure}
 \centering
  \includegraphics[width = 3.3in]{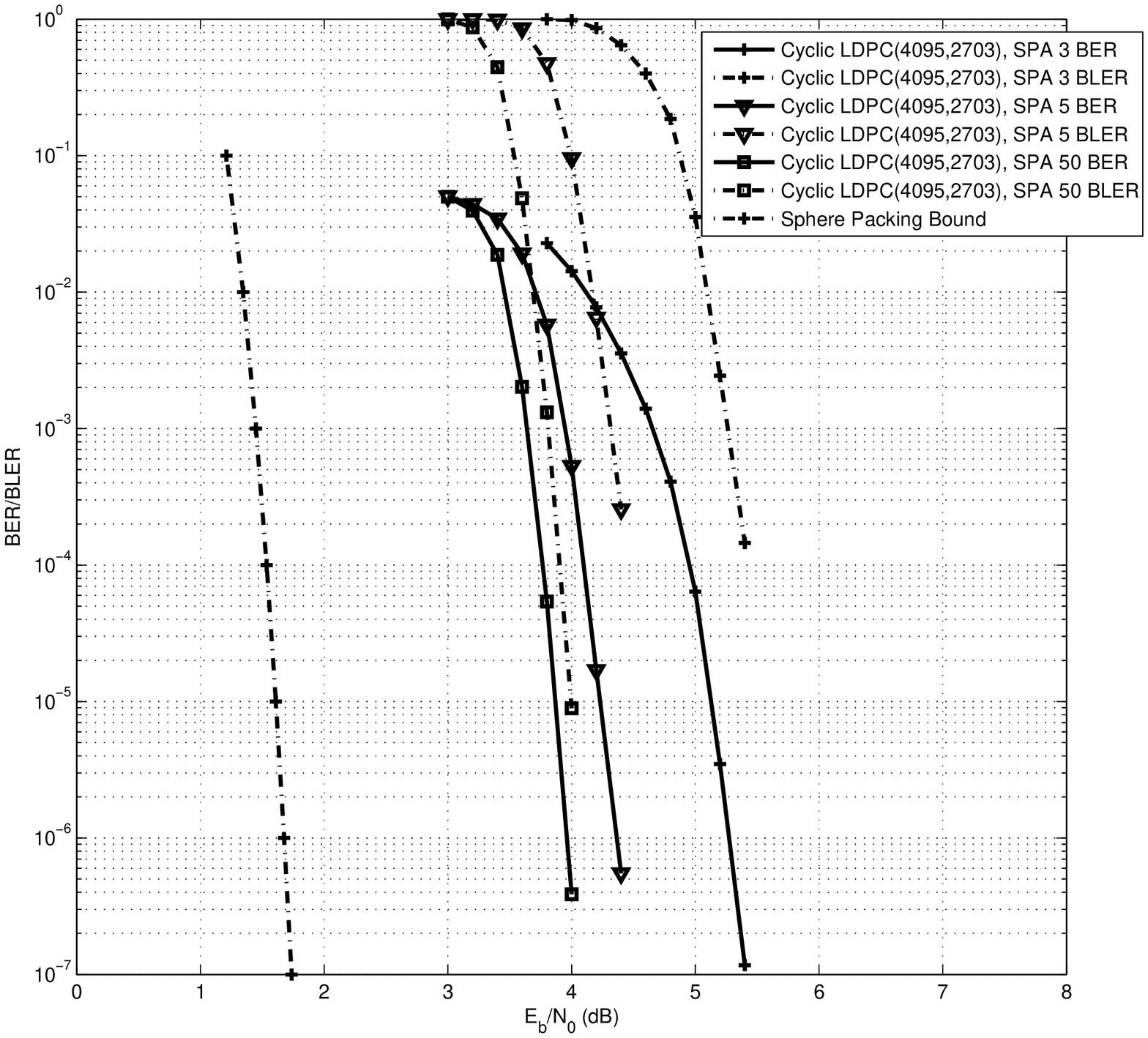}
 \caption{(c) The bit and block error performances of the binary (4095,2703) cyclic EG-LDPC code given in Example 3.}\label{fig:2703}
 \end{figure}

        \begin{figure}
 \centering
  \includegraphics[width = 3.3in]{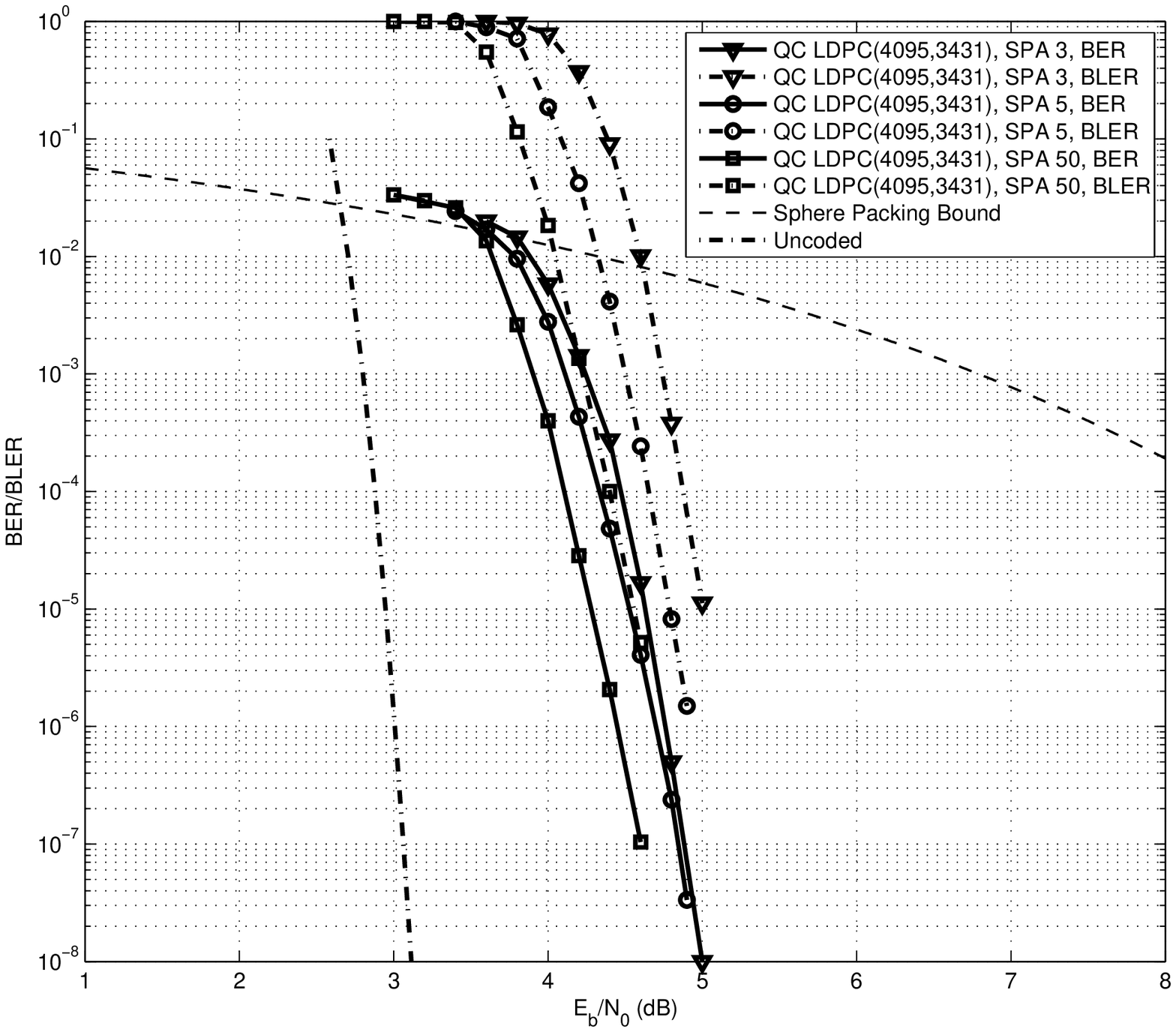}
 \caption{ The bit and block error performances of the binary (4095,3431) QC EG-LDPC code given in Example 4.}\label{fig:3431}
 \end{figure}
 
  \begin{figure}
 \centering
 \includegraphics[width = 3.3in]{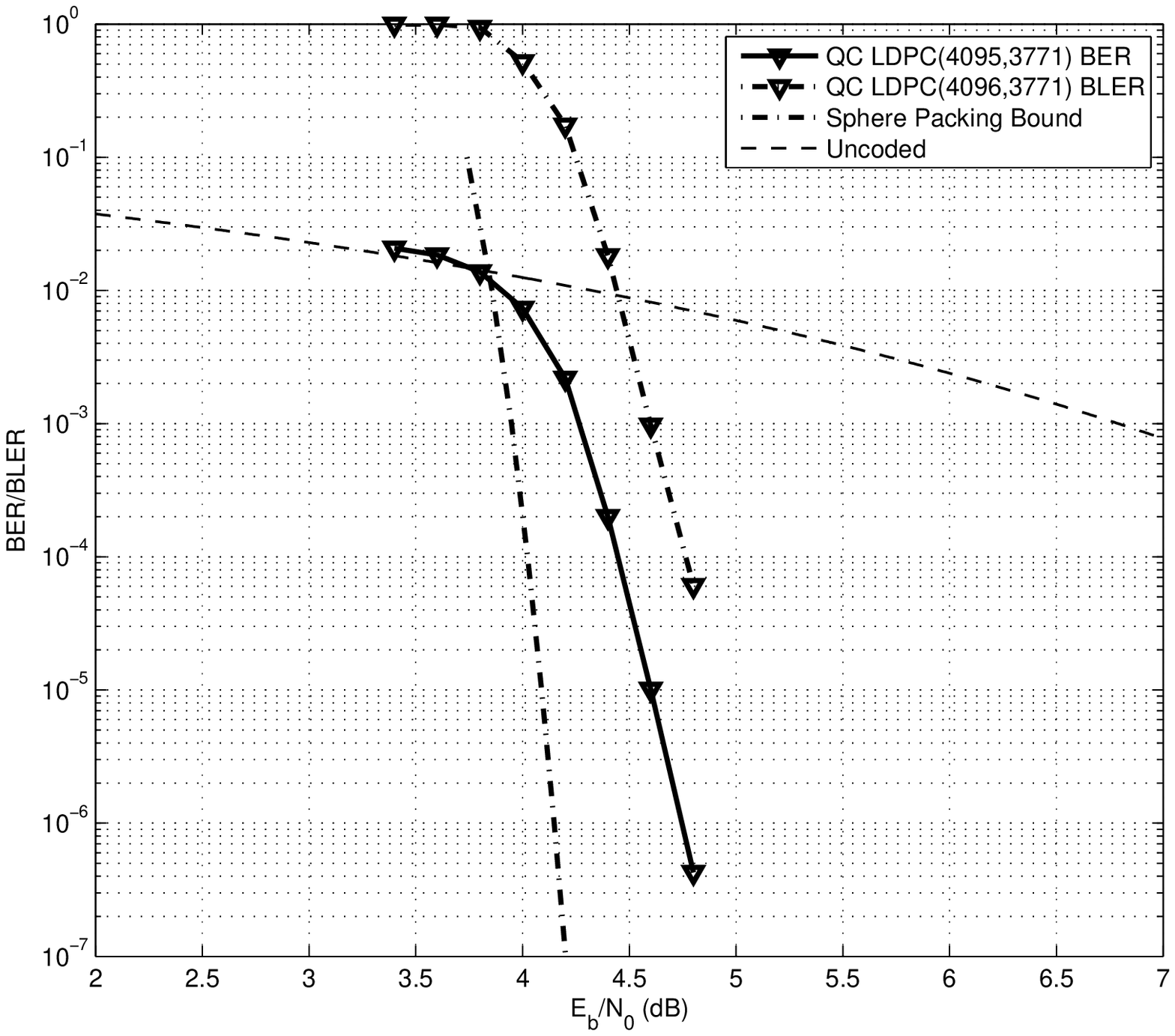}
 \caption{The bit and block error performance of the binary (4095,3771) QC-LDPC code given in Example 5.}\label{fig:3771}
 \end{figure}
 
       \begin{figure}
 \centering
 \includegraphics[width = 3.3in]{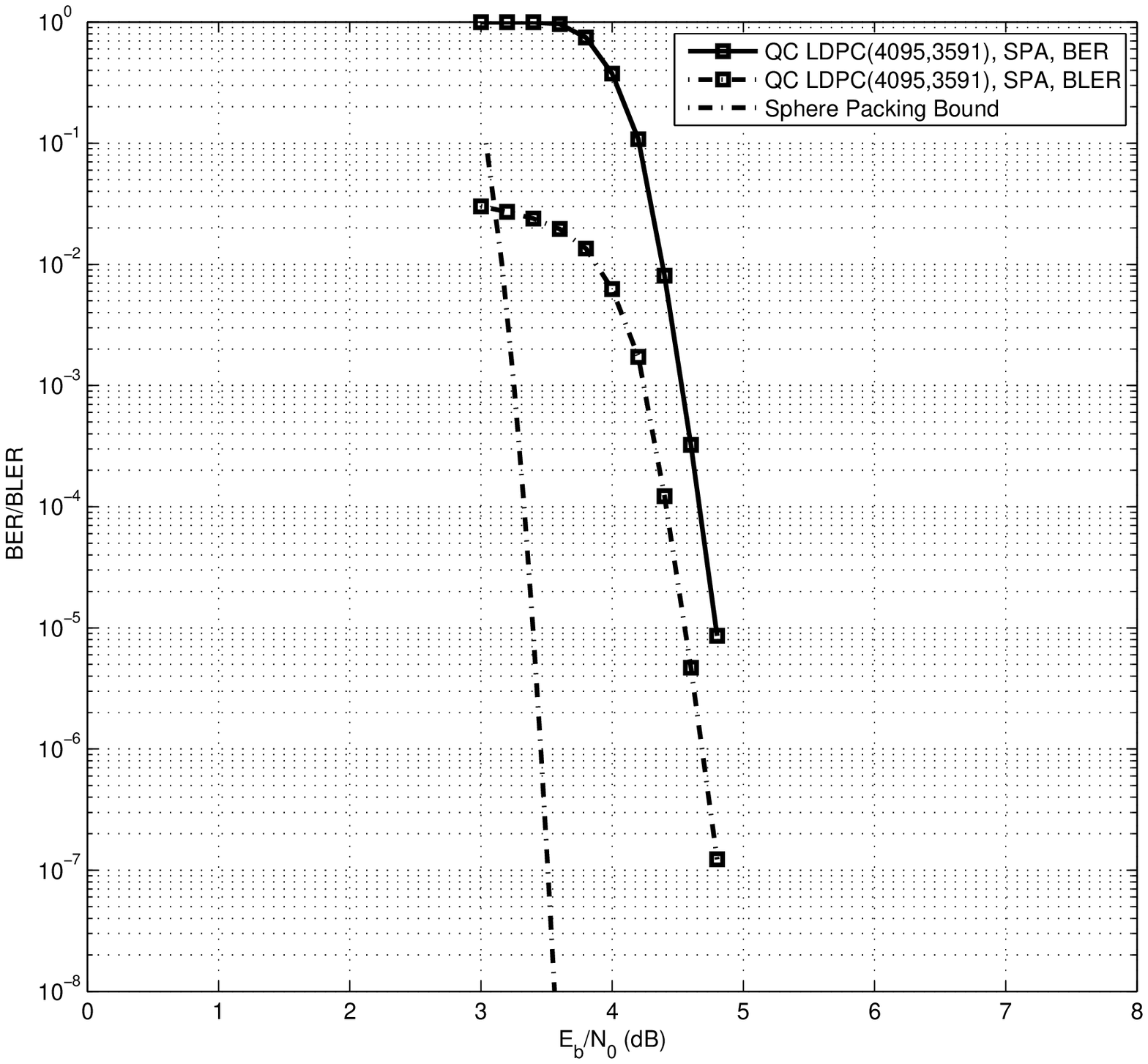}
\caption{The bit and block error performance of the binary (4095,3591) QC-LDPC code given in Example 6.}\label{fig:3591}
 \end{figure}
    
             \begin{figure}
 \centering
 \includegraphics[width = 3.3in]{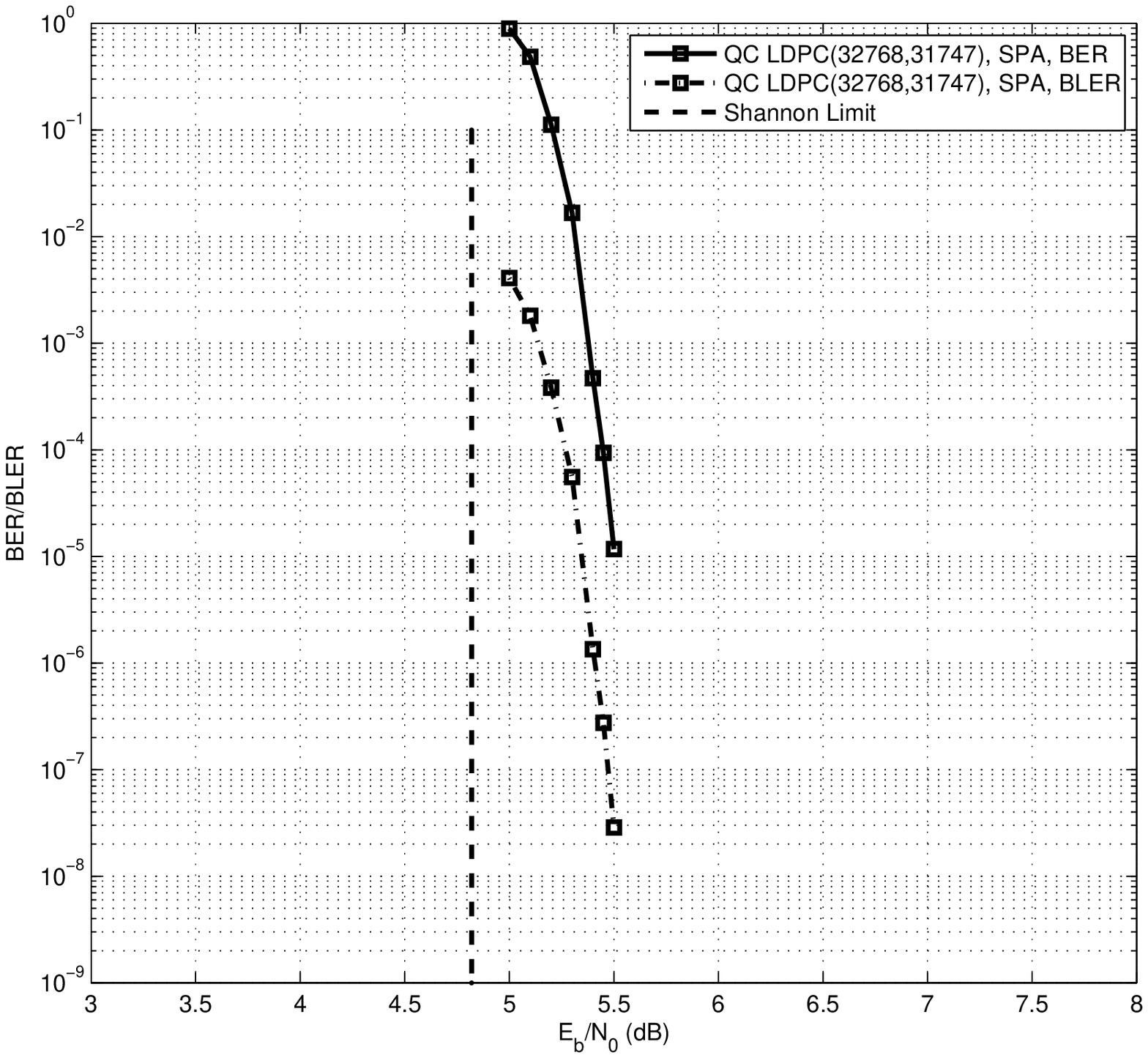}
 \caption{The bit and block error performances of the binary (32768,31747) QC-LDPC code given in Example 7.}\label{fig:31747}
 \end{figure}   
 
     \begin{figure}
 \centering
 \includegraphics[width = 3.3in]{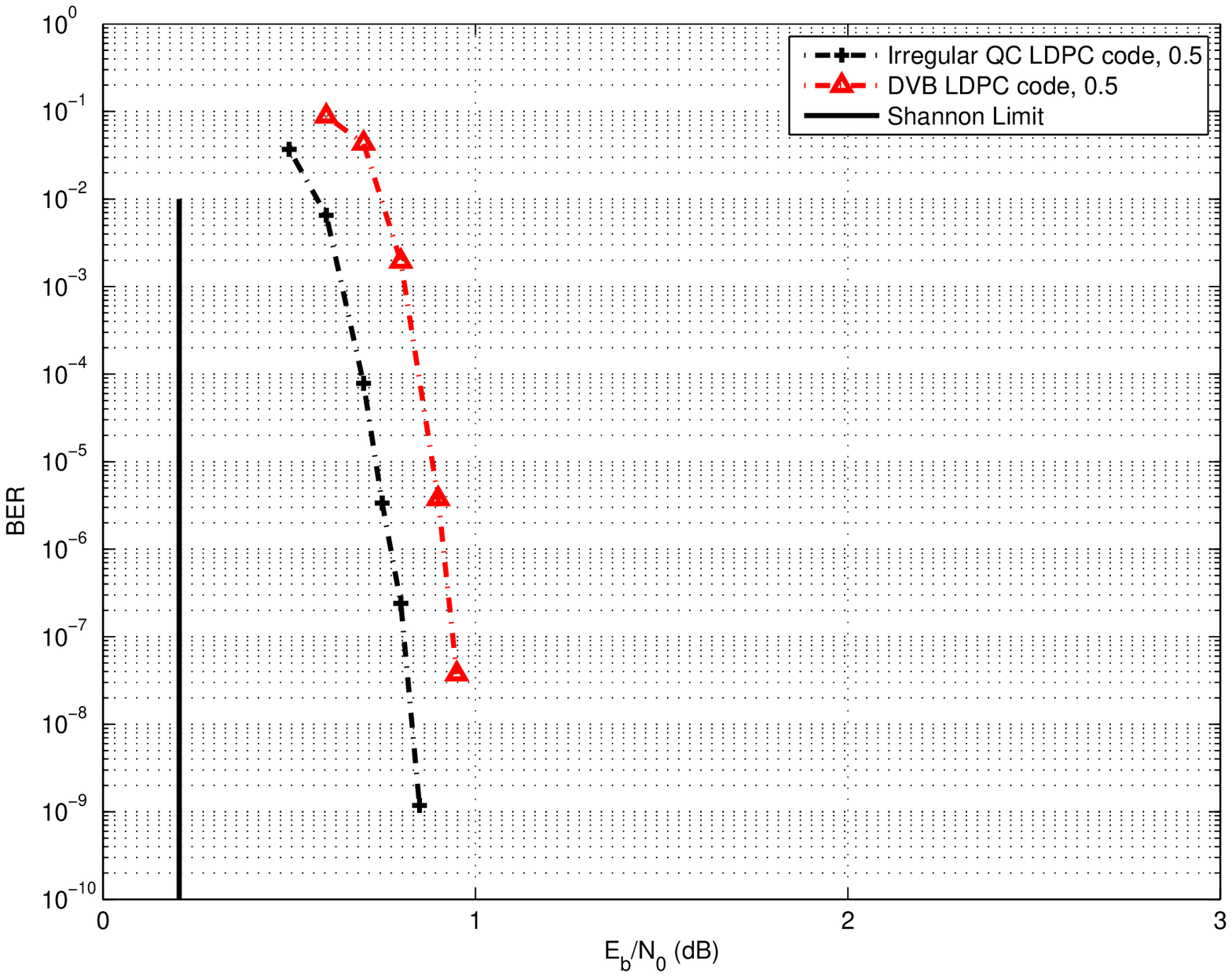}
 \caption{The error performances of the binary (65536,32768) QC-LDPC code and the DVB S-2 standard code given in Example 8.}\label{fig:65536}
 \end{figure}

      \begin{figure}
 \centering
 \includegraphics[width = 3.3in]{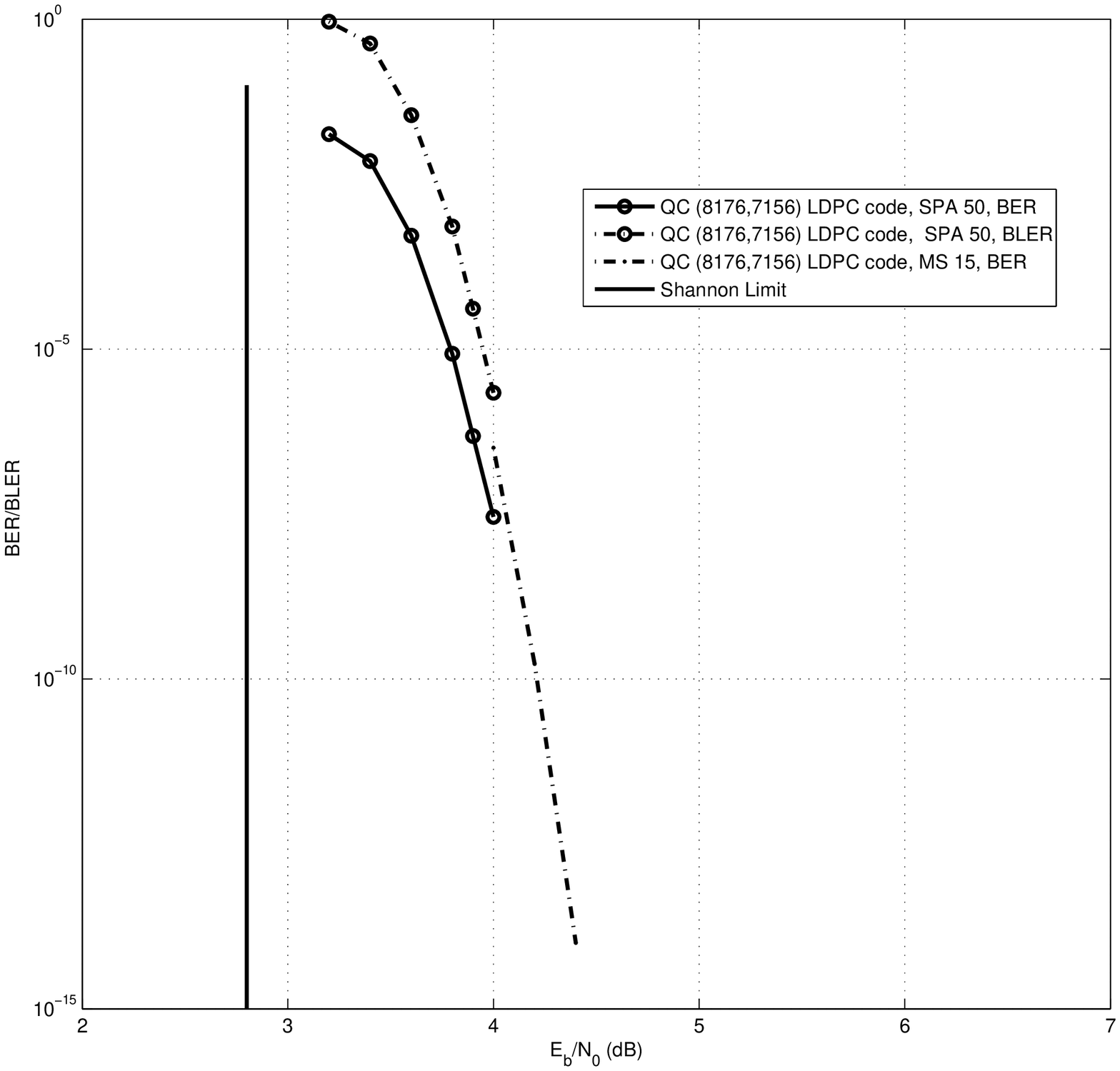}
 \caption{The error performances of the binary (8176,7156) QC-LDPC code given in Example 9.}\label{fig:7156}
 \end{figure}

      \begin{figure}
 \centering
 \includegraphics[width = 3.3in]{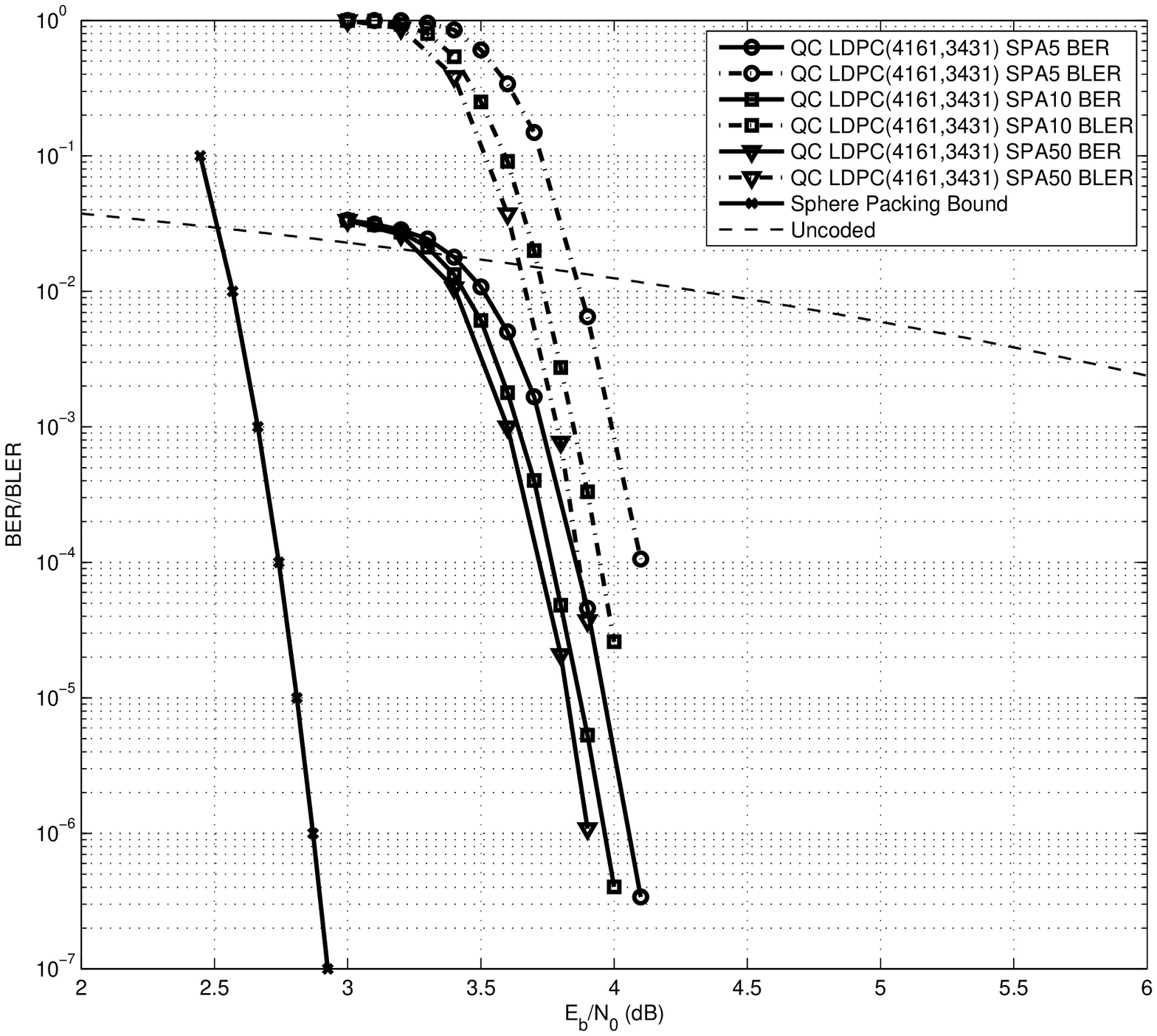}
 \caption{(a) The error performances of the (4161,3431) cyclic PG-LDPC code given in Example 10 decoded with various number of iterations of the SPA.}\label{fig:pg3431}
 \end{figure}
   
    \setcounter{figure}{8}
   \begin{figure}
 \centering
 \includegraphics[width = 3.3in]{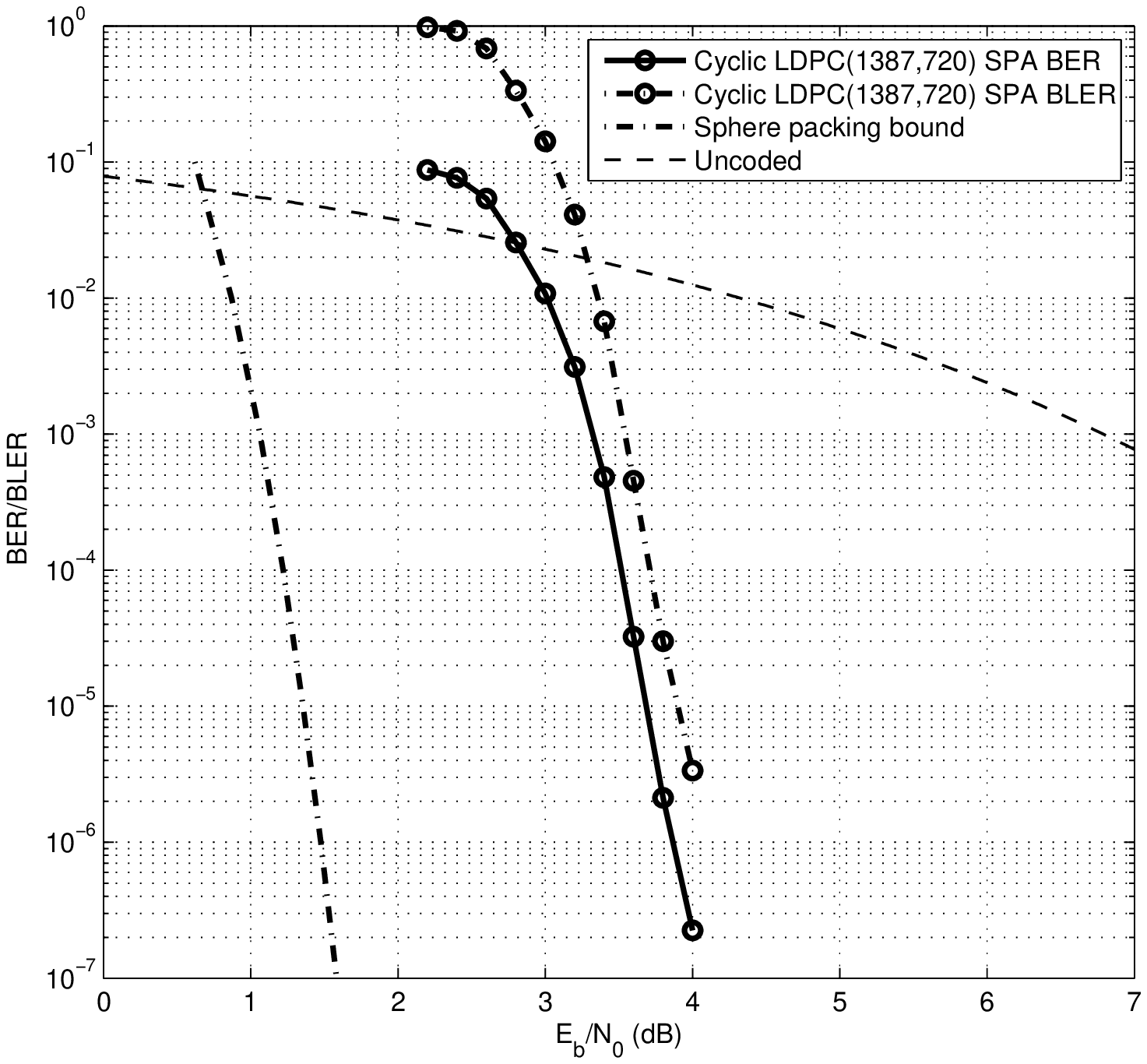}
 \caption{(b) The error performances of the binary (1387,720) cyclic LDPC code given in Example 10.}\label{fig:1387}
 \end{figure}

\setcounter{figure}{9}
   \begin{figure}
 \centering
 \includegraphics[width = 5in]{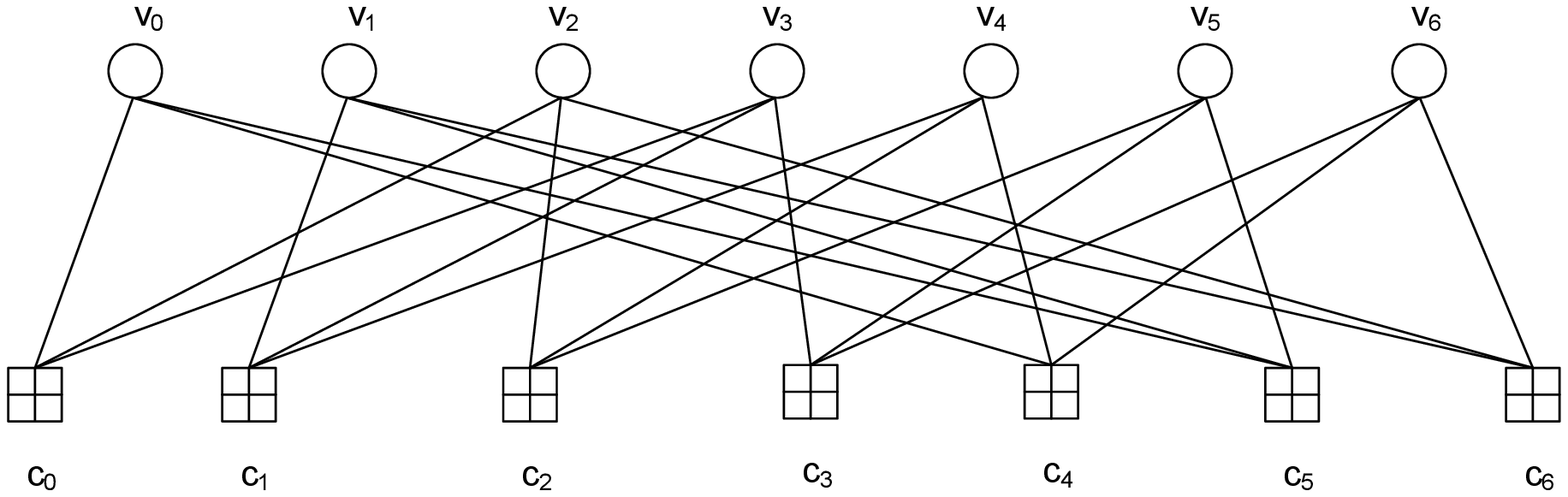}
 \caption{(a) The Tanner graph of a (3,3)-regular (7,3) LDPC code.}
 \end{figure}

\setcounter{figure}{9}
   \begin{figure}
 \centering
 \includegraphics[width = 3.6in]{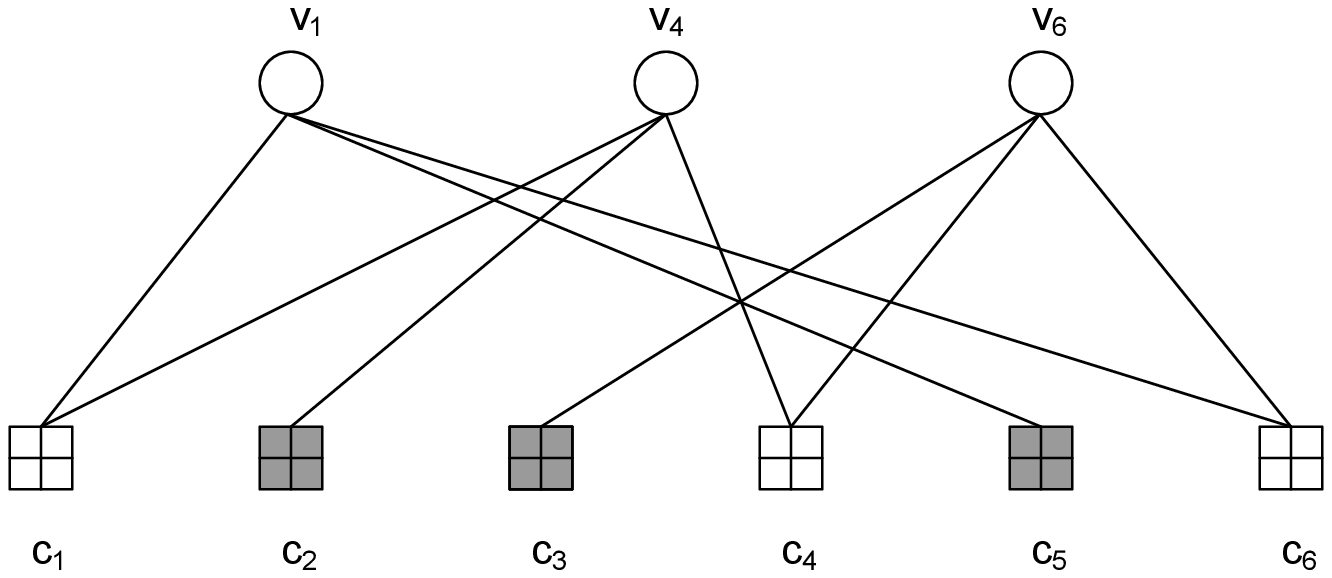}
 \caption{(b) A elementary (3,3) trapping set.}
 \end{figure}

\setcounter{figure}{9}
   \begin{figure}
 \centering
 \includegraphics[width = 3.6in]{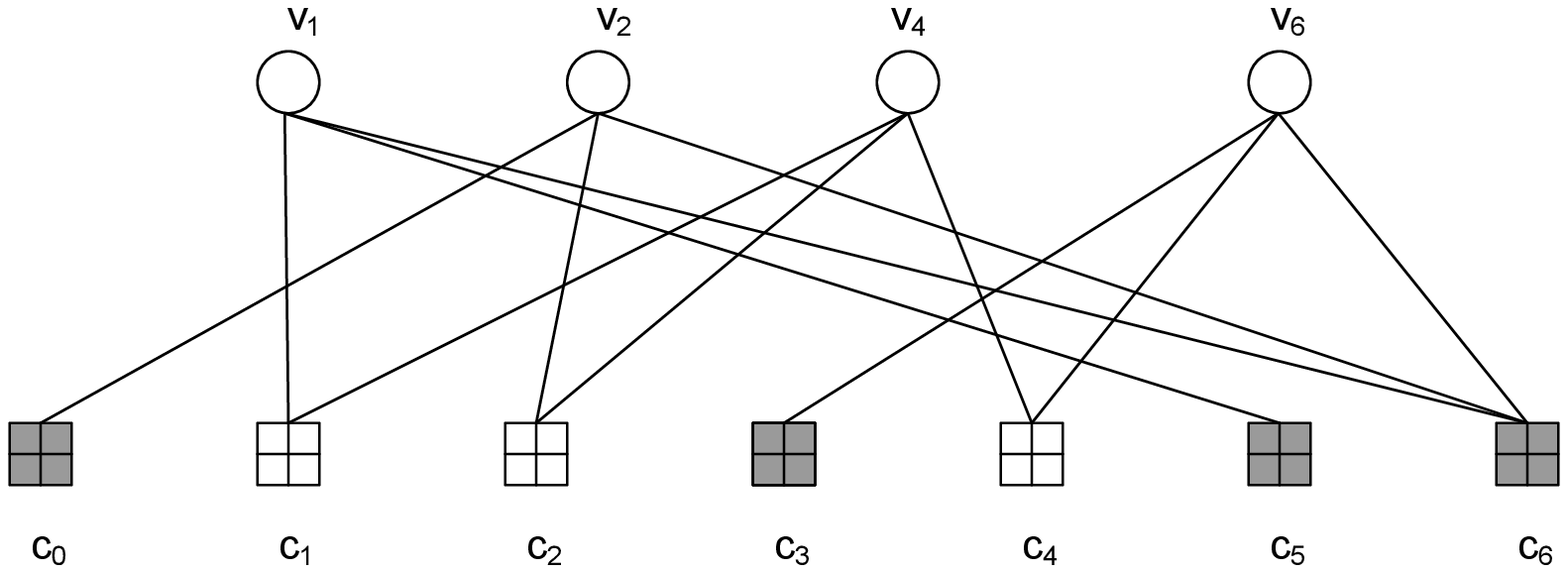}
 \caption{(c) A (4,4) trapping set.}
 \end{figure}

      \begin{figure}
 \centering
 \includegraphics[width = 3.3in]{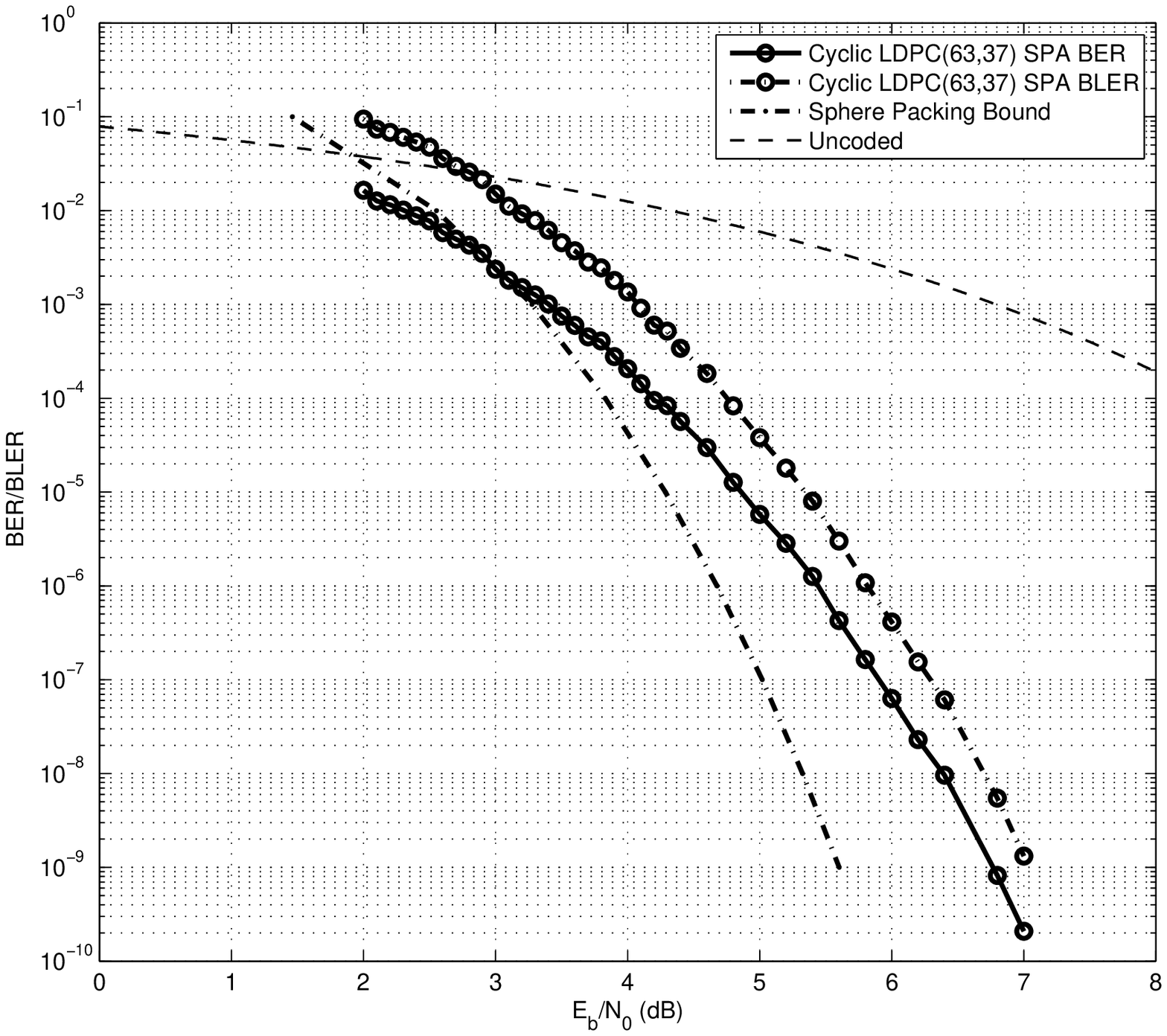}
 \caption{The bit and block error performances of the (63,37) cyclic EG-LDPC code given in Example 11.}\label{fig:63}
 \end{figure}
 
 \clearpage
 \begin{figure}
 \centering
 \includegraphics[width = 3.3in]{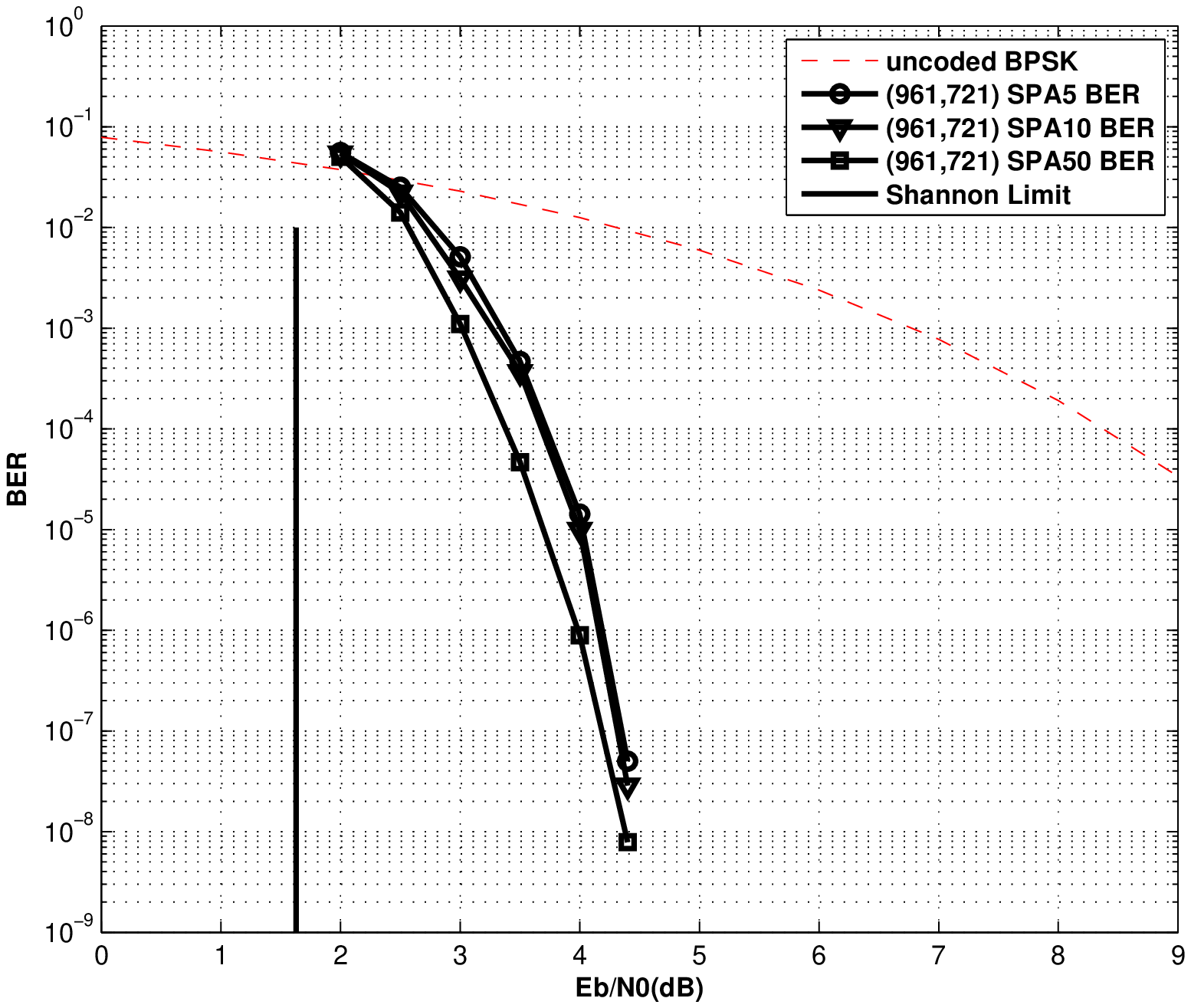}
 \caption{The bit error performance of the binary (961,721) QC-LDPC code given in Example 13 decoded with 5, 10 and 50 iterations of the SPA.}
  \end{figure}

         \begin{figure}
 \centering
 \includegraphics[width = 3.3in]{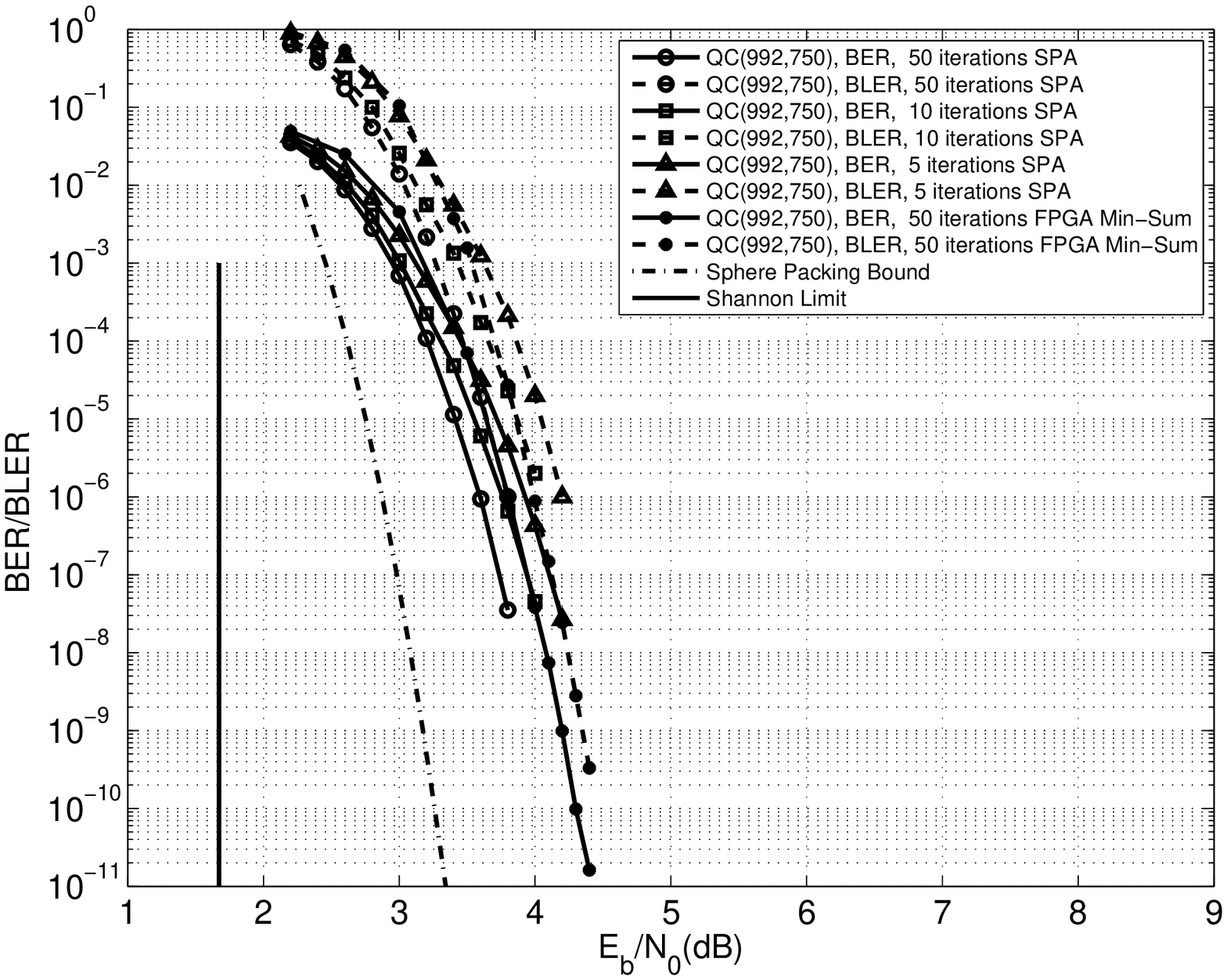}
 \caption{The bit and block error performances of the binary (992,750) QC-LDPC code given in Example 14.}\label{fig:992}
 \end{figure}

 \end{document}